\documentclass[10pt, letterpaper]{article}

\usepackage[letterpaper,margin=1in]{geometry} 
\usepackage{setspace} 

\usepackage{yhmath} 
\usepackage{amsmath}
\usepackage{amsthm}
\usepackage{amssymb}
\usepackage{dsfont} 

\usepackage[hyperfootnotes=False]{hyperref}
\usepackage{tabularx}
\usepackage{booktabs}
\usepackage{siunitx}  
\usepackage{graphicx} 
\usepackage{caption}
\usepackage{subcaption}
\usepackage{threeparttable}
\usepackage{xcolor}

\newtheorem{mainassumption}{Assumption}

\newtheorem{maintheorem}{Theorem}

\newtheorem{corollary}{Corollary}

\newtheorem{lemma}{Lemma}[section]
\newtheorem{assumption}{Assumption}[section]
\newtheorem{remark}{Remark}[section]

\DeclareMathOperator*{\argmin}{arg\,min}

\usepackage{natbib}
\usepackage[USenglish]{babel}

\usepackage{orcidlink}

\title{
  Transformer-based CoVaR:\\
    Systemic Risk in Textual Information
}

\author{%
  Junyu Chen \footnotemark[1]%
  \and
  Tom Boot \footnotemark[1]%
  \and
  Lingwei Kong \footnotemark[1]%
  \and
  Weining Wang \footnotemark[2]%
}

\begin{document}

\renewcommand{\thefootnote}{\fnsymbol{footnote}}
\maketitle
\footnotetext[1]{%
    Department of Economics, Econometrics and Finance, Faculty of Economics and Business, University of Groningen, The Netherlands. 
    Emails: 
    \href{mailto:junyu.chen@rug.nl}{junyu.chen@rug.nl}, 
    \href{mailto:t.boot@rug.nl}{t.boot@rug.nl},
    \href{mailto:l.kong@rug.nl}{l.kong@rug.nl}.
}

\footnotetext[2]{%
    Corresponding author. School of Economics, University of Bristol, UK. Address: Beacon House, Queens Road, Bristol BS8 1QU, UK.
    Email: \href{mailto:weining.wang@bristol.ac.uk}{weining.wang@bristol.ac.uk}. Phone: +44 7413 466969.
}

\begin{abstract}
\noindent 
Conditional Value-at-Risk (CoVaR) quantifies systemic financial risk by measuring the loss quantile of one asset, conditional on another asset experiencing distress. We develop a Transformer-based methodology that integrates financial news articles directly with market data to improve CoVaR estimates. Unlike approaches that use predefined sentiment scores, our method incorporates raw text embeddings generated by a large language model (LLM). We prove explicit error bounds for our Transformer CoVaR estimator, showing that accurate CoVaR learning is possible even with small datasets. Using U.S. market returns and Reuters news items from 2006--2013, our out-of-sample results show that textual information impacts the CoVaR forecasts. With better predictive performance, we identify a pronounced negative dip during market stress periods across several equity assets when comparing the Transformer-based CoVaR to both the CoVaR without text and the CoVaR using traditional sentiment measures. Our results show that textual data can be used to effectively model systemic risk without requiring prohibitively large data sets.
\\

\\


\end{abstract}
\newpage

\renewcommand{\thefootnote}{\arabic{footnote}}

\section{Introduction}

Between 2007 and 2009, the crisis in the U.S. mortgage market set off a chain reaction that destabilized financial institutions and the global economy. To better anticipate these events, regulators, policymakers, and researchers focused their attention on measuring the systemic risk that arises in a strongly connected economy. While some systemic risk factors can be effectively distilled from numerical market data, additional information can be extracted from textual data, one well-known example being “market sentiment” \citep{2017_Keiber_and_Samyschew}. However, this text-based risk information is inherently difficult to quantify, and a key open question is how we can effectively extract it from textual sources to obtain better measures and forecasts of systemic risk.



In this paper, we propose a formal econometric framework that uses a Transformer-based architecture \citep{2017_Vaswani_et_al} to improve systemic risk analysis by integrating structured financial variables and unstructured textual information such as financial news articles. Our approach proceeds in two steps. First, we employ a large language embedding model that maps variable-length news articles into fixed-size, high-dimensional vectors representing semantic content. Second, we extract the key latent information from these high-dimensional input sequences by leveraging the Transformer's ability to capture nonlinear dependencies and contextual relationships across diverse data sources (i.e., \textit{multimodal data}\footnote{By multimodal data, we mean inputs of different types, combining structured numerical variables (e.g., returns) with unstructured information such as text.}).

A Transformer is a neural network that relies on the self-attention mechanism to compute latent representations of its input. Traditional neural network architectures (such as convolutional or recurrent networks) rely on fixed parameterized transformations of their inputs, whereas the self-attention mechanism in a Transformer dynamically adapts the transformation based on the input. In other words, self-attention allows the importance of a variable to change based on the surrounding data. For example, a negative news item may be weighted as highly significant for predicting risk when other variables mention "default", whereas the same news item might be down-weighted if the context is a "holiday". By identifying these context-dependent relationships, the Transformer provides the latent representations necessary for the risk prediction task.

A main motivation for using the Transformer is that traditional sentiment scores often miss critical information in textual data or may even misrepresent sentiment indicators relevant for market prediction \citep{2019_Ke_et_al, 2011_Loughran_and_McDonald, 2007_Tetlock}. 
Transformers can effectively learn signals from noisy, high-dimensional textual inputs by down-weighting irrelevant content \citep{2022_Edelman_et_al}, even when the sample size is limited \citet{2025_Mousavi-Hosseini_et_al}. To show that such results can be also be expected when estimating the CoVaR on relatively small financial data sets, we establish a new convergence result that complements recent theory on convergence rates for deep neural networks in high dimensions \citep{2020_Schmidt-Hieber, 2019_Suzuki}. Specifically, we show how the quantile regression framework of \citet{2022_Padilla_et_al} can be extended to incorporate text information through the Transformer architecture. We then show that under the architecture we propose, the convergence rate of the CoVaR estimator only depends logarithmically on the dimension of the input data. As a result, a researcher does not manually need to summarize unstructured text into a low-dimensional object, for example via a sentiment scores, and then specify how these summaries affect risk.

Empirically, we estimate the Conditional Value-at-Risk (CoVaR), a widely used measure of systemic risk, for eight Global Systemically Important Banks (G-SIBs) in the U.S. We consider daily returns between 2006--2013, including both the 2008 financial crisis and the 2011 debt ceiling crisis. We augment the standard numerical data used to estimate CoVaR with a rich set of \textit{Reuters} financial news articles. The Transformer-based method that incorporates the full text data yield more negative CoVaR and $\Delta$CoVaR values during crisis periods than approaches that either exclude news information or reduce it to sentiment scores prior to risk modeling. After addressing potential look-ahead bias, we show that our empirical results remain robust. Furthermore, we demonstrate that incorporating text into the Transformer architecture achieves generally lower prediction loss compared to the model relying only on structured financial return data.

In summary, this paper makes three contributions to advancing the understanding and practical application of multimodal data in systemic risk assessment. First, to our knowledge, we are the first to introduce an econometric framework for analyzing CoVaR using the Transformer-based architecture. Second, we establish upper bounds for its $\ell_2$ risk and also generalization error. Finally, we empirically demonstrate that incorporating textual information leads to better predictive performance and alters the CoVaR risk measures. In particular, we identify a pronounced negative dip during the crisis period across several equity assets when comparing the Transformer-based CoVaR to both the CoVaR without text and the CoVaR using the traditional sentiment measure. 

\textbf{Outline.} The remainder of the paper is organized as follows. Section \ref{sec:literature_review} reviews the related literature. Section \ref{sec:2_methodology} introduces the model and estimators. Section \ref{sec:theory} presents the main theoretical results. Section \ref{sec:application} reports the findings of the empirical application. Section \ref{sec:conclusion_future_work} concludes and discusses directions for future research. Appendix \ref{sec:transformer_details} provides the description of an alternative Transformer architecture. Appendix \ref{sec:news_content} discusses the interaction between news content and risk prediction. Appendix \ref{sec:simulation} provides additional results from the simulation study. Proofs of all theoretical results are provided in the Supplemental Appendix (SA) to streamline the main exposition; these results may also be of independent technical interest.

\section{Literature Review} \label{sec:literature_review}

Broadly speaking, we can distinguish three approaches to risk modeling in a connected world which we refer to as network-based models, simulation-based approaches and market-based approaches. Early network-based models \citep{2000_Allen_and_Gale, 2012_Billio_et_al, 2014_Diebold_and_Yilmaz, 2015_Acemoglu_et_al} highlight the importance of capturing interconnectedness and nonlinear contagion effects. Simulation-based approaches \citep{2011_Gai_et_al, 2015_Greenwood_et_al} underscore amplification mechanisms like fire sales and liquidity spirals, yet they often assume simplified dynamics.

Market-based approaches offer another perspective by deriving systemic risk measures from asset prices and market information: \cite{2017_Acharya_et_al} proposed Systemic Expected Shortfall (SES) and its variant to estimate capital shortfalls in systemic crises, and \cite{2009_Segoviano_and_Goodhart} developed the Distress Insurance Premium (DIP) using option-implied risk-neutral probabilities. Among market-based approaches, the CoVaR framework proposed by \citet{2016_Adrian_and_Brunnermeier} remains a central tool for quantifying systemic risk. Subsequent contributions, including \citet{2015_Hautsch_et_al}, \citet{2016_Haerdle_et_al}, and \citet{2022_Keilbar_and_Wang}, extend CoVaR to high-dimensional environments. The high-dimensional setting is essential in applied practice, as systemic risk emerges from the complex interconnectedness of global markets and macroeconomic factors that low-dimensional models fail to address simultaneously.
However, these methodologies do not accommodate multimodal data structures, which are increasingly relevant in modern financial systems. 

A popular strategy to accommodate information from text is the use of sentiment analysis, where texts are for example classified as positive, negative, or neutral. Early work constructs sentiment indices using dictionary-based word counts: \cite{2007_Tetlock} shows that pessimistic media tone predicts short-term market declines, while \cite{2011_Loughran_and_McDonald} refine sentiment lexica to better capture financial context. Later studies, such as \cite{2019_Ke_et_al}, employ topic modeling to assign sentiment weights to terms and aggregate them into document-level scores for return prediction. Even recent neural-network-based studies are sentiment-focused: \cite{2024_Schuettler_et_al} find that LLMs enhance stock return prediction, \cite{2024_Wang_and_Ma} propose MANA-Net to aggregate sentiment dynamically by market relevance, and \cite{2025_Jha_et_al} measure popular financial sentiment across languages using LLMs. These findings highlight the value of richer linguistic and contextual signals through market sentiment. However, reducing text to low-dimensional sentiment categories may overlook other dimensions of information.

Much of the literature has found evidence of other channels affects financial risk and market outcomes beyond sentiment. For example, linguistic characteristics such as grammatical structure, linguistic complexity, etc. affect how information is processed and priced by markets \citep{2014_Loughran_and_McDonald, 2024_Fedyk, 2024_AlexKim_et_al}. In particular, \cite{2011_Engelberg_and_Parsons} and \cite{2012_Dougal_et_al} establish causal links between differential reporting of events by media or journalists and financial markets. Beyond linguistic characteristics, recent studies demonstrate the informational value of textual data for asset pricing and financial risk (e.g., \cite{2025_Adammer_et_al}; \cite{2023_van_Dijk_and_de_Winter}; \cite{2020_Bybee_et_al}), consistent with psychological evidence linking emotions to financial decisions \citep{2000_Lerner_and_Keltner, 2001_Lerner_and_Keltner}. In addition, topic-based analyses such as \citet{2009_Kogan_et_al} and \citet{2019_Larsen_and_Thorsrud} show that latent news topics predict volatility and macroeconomic variables. Thus, focusing solely on simplified sentiment analysis may be insufficient to capture the dynamics of risk, particularly during periods of financial stress. 

Recent advances in multimodal learning reinforce the view that textual information can provide valuable predictive signals for financial markets beyond traditional sentiment indicators. Transformers and LLMs have been applied to financial prediction using diverse inputs, including earnings-call audio, social media, and news \citep{2020_Li_et_al, 2023_Zou_and_Herremans, 2024_Wang_et_al, 2024_Omranpour_et_al, 2025_Wang_et_al, 2020_Yang_et_al, 2025_Cao_et_al}. Beyond multimodal prediction, \citet{2026_Pele_et_al} examine LLMs guided by human prompts for financial forecasting. Our empirical application builds on these approaches by employing a large language embedding model that maps news articles into high-dimensional vectors representing semantic content, rather than collapsing them into a single sentiment score.
\section{Methodology} \label{sec:2_methodology}

In this section, we first introduce the systemic risk modeling framework and then the Transformer-based architecture used for approximating a quantile function that integrates both financial variables and market-related textual information. In the next section, we formally study its approximation and generalization properties.

\subsection{Systemic Risk Modeling}


Our CoVaR estimation procedure proceeds in two steps, which builds on the framework of \cite{2022_Keilbar_and_Wang} but modifies the second estimation step. In the first step,  we estimate the VaR for each bank using a linear quantile regression model. In the second step, we estimate CoVaR by first fitting a Transformer-based quantile regression conditional on other banks’ returns and all financial news, and then plugging in the corresponding VaR estimates from the first step to replace the returns. We do not incorporate text data in the first step to maintain the standard definition of VaR as a quantile of the marginal return distribution. Our focus is therefore on the second step, where we use the Transformer's ability to process textual data. 

In the following, we let $R_{j,t}\in\mathbb{R}$ denote the return of bank $j$ at time $t$, where $j=1,\ldots,J$ with $J$ fixed and finite, and $t=1,\ldots,T$. To incorporate textual information, we utilize an embedding mapping, which is referred to in the deep learning literature as an \emph{embedding layer}. This embedding layer maps individual news items into $d_e$--dimensional numerical vectors (i.e., embedding vectors) that represent their semantic content in a continuous coordinate space. Specifically, let $n$ denote the number of news items observed for bank $j$ at time $t$. Applying the embedding layer to each news item yields a news-embedding matrix $E_{j,t} \in \mathbb{R}^{d_e \times n}$, where each column represents a single news item and $d_e$ is the embedding dimension. In our empirical application (see Section \ref{sec:application}), we employ a \textit{pre-trained}\footnote{A pre-trained model is trained in advance on large, general data.} mapping $\pi_e$ and treat its parameters as fixed. In practice, one may instead \textit{fine-tune}\footnote{Fine-tuning means starting from a pre-trained model and then slightly updating its parameters using task-specific data.} its embedding parameters.

\textbf{Step 1: Estimation of VaR}\\
A variety of methods has been developed to estimate VaR, including the historical simulation method \citep{1999_Hull_and_White}, parametric volatility models such as ARCH and GARCH \citep{1982_Engle_et_al,1986_Bollerslev}, approaches based on Extreme Value Theory \citep{2000_McNeil_and_Frey}, and linear quantile regression \citep{2015_Chao_et_al}. Any of the approaches described above can be used to construct our VaR estimator. Among these approaches, linear quantile regression provides a transparent, flexible, and easily interpretable alternative \citep{1978_Koenker}. It has been widely used in empirical studies linking macro-financial conditions to tail risks \citep{2016_Adrian_and_Brunnermeier, 2016_Haerdle_et_al}. Following \citet{2016_Adrian_and_Brunnermeier}, let $M_{t-1} \in \mathbb{R}^{m}$ denote a vector of lagged macroeconomic state variables that reflect the general state of the economy. The VaR for bank $i$ at time $t$ is defined as the $\tau$-quantile of its return distribution conditional on $M_{t-1}$:
\begin{equation}
    \mathrm{P}\bigl(R_{i,t} \leq \operatorname{VaR}_{i,t}^{\tau}|M_{t-1}\bigr) = \tau, \quad \tau \in (0,1).
\end{equation}
By convention, $\operatorname{VaR}_{i,t}^{\tau}$ is typically negative, reflecting a potential loss. Although alternative sign conventions exist, we follow the one stated above. 
We employ a linear quantile regression with lagged macroeconomic variables to estimate $\operatorname{VaR}_{i,t}^{\tau}$:
\begin{equation} \label{eq:var}
    R_{i,t} = \alpha_{i}(\tau) + \gamma_{i}(\tau) ^{\top}M_{t-1} + \epsilon_{i,t}(\tau)
\end{equation}
with $\alpha_{i}(\tau) \in \mathbb{R}$, $\gamma_{i}(\tau)\in \mathbb{R}^{m}$. We assume $F^{-1}_{\epsilon_{i,t}(\tau) \mid M_{t-1}}(\tau) = 0$ with $\epsilon_{i,t}(\tau)$ being the error term for the $\tau$-th quantile. The VaR estimate is then the fitted value of the quantile regression
\begin{equation}
\widehat{\operatorname{VaR}}_{i,t}^{\tau} = \widehat{\alpha}_{i}(\tau) + \widehat{\gamma}_{i}(\tau)^{\top} M_{t-1}.
\end{equation}

\textbf{Step 2: Estimation of CoVaR}\\
CoVaR, introduced as a systemic extension of standard VaR by \cite{2016_Adrian_and_Brunnermeier}, is defined as the VaR of bank $j$ conditional on a specific state of the financial system. In particular, we consider the specific state as the distress scenario of the financial system, while also taking the market-related textual information (i.e., financial news) into account. For the distress scenario of financial systems, we assume that all other banks are at their VaR level. Formally, we define CoVaR as 
\begin{equation}
    {\mathrm{P}\Bigl(R_{j,t} \leq \operatorname{CoVaR}_{j,t}^{\tau} \,\Big|\, R_{-j,t} = \operatorname{VaR}_{-j,t}^{\tau}, \,  E_{j,t}\Bigr) = \tau,}
\end{equation}
where $R_{-j,t}, \operatorname{VaR}_{-j,t}^{\tau} \in \mathbb{R}^{J-1}$ denote the returns and VaRs of all banks other than $j$ at time period $t$, respectively, and {$E_{j,t}$ relates to the embedded information for $j$-th bank up to at most time period $t$ (e.g., embedded financial news at time period $t-1$)}. We estimate $\operatorname{CoVaR}_{j,t}^{\tau}$ using a Transformer-based quantile regression described in Section~\ref{sec:transformer}. Specifically, the model is given by
\begin{equation} \label{eq:f_target}
    R_{j,t} = f_{j,\tau}^*\Bigl(R_{-j,t}, \, E_{j,t}\Bigr) + \nu_{j,t}(\tau),
\end{equation}
where $f_{j,\tau}^*(\cdot)$ is the true conditional quantile function for bank $j$, and $\nu_{j,t}$ is an error term satisfying $F^{-1}_{\nu_{j,t}(\tau) \mid R_{-j,t}, E_{j,t}}(\tau) = 0$. The function $f_{j,\tau}^*$ is essential for capturing the complex tail behavior of returns and news. It serves as the channel through which non-sentiment-related risk, as introduced earlier, transmits to the target bank return. Under the distress scenario, we replace $R_{-j,t}$ with the VaR estimates of all other banks $\widehat{\operatorname{VaR}}_{-j,t}^{\tau}$, so that the estimated CoVaR is computed as  
\begin{equation} \label{eq:fhat}
    \widehat{\operatorname{CoVaR}}_{j,t}^{\tau} = \widehat{f}_{j,\tau}\Bigl(\widehat{\operatorname{VaR}}_{-j,t}^{\tau}, \, E_{j,t}\Bigr),
\end{equation}
where $\widehat{f}_{j,\tau} \in \mathcal{G}$, with $\mathcal{G}$ denoting a class of Transformer-based neural network functions. The class $\mathcal{G}$, along with the model architecture and estimation procedure, is introduced in the next section.

\subsection{Transformer Architecture} \label{sec:transformer}

We adopt the Transformer architecture for the following reasons. Firstly, a key challenge for our CoVaR framework is integrating numerical and textual inputs without imposing restrictive dependence assumptions. Conventional architectures like CNNs (assuming spatial locality) and recurrent architectures like RNNs (assuming temporal order) \citep{2017_Yin_et_al} embed strong functional constraints that may not hold for our data. In contrast, the Transformer architecture enables direct pairwise interactions among all input elements, thereby capturing global dependencies \citep{2017_Vaswani_et_al, 2020_Cordonnier_et_al}. This may make the Transformer architecture better suited to settings where neither specific spatial nor temporal structure is clearly known a priori.

Secondly, our application necessitates a model that can learn sparse functions from noisy data. Financial news contains significant redundancy and irrelevance, requiring the model to identify and focus on meaningful information and signals. Additionally, the volume of news varies daily, requiring the use of \textit{padding}\footnote{Padding refers to extending shorter sequences with special meaningless placeholders so that all sequences have the same length. This is necessary because neural networks typically require inputs of uniform size.} to create uniform input sequences. This introduces artificial \textit{tokens}\footnote{In Transformers, a token is the basic unit of text that the model processes. It is analogous to a variable or feature in an econometric dataset. A token is typically represented as a vector. Padding tokens act as placeholders.} that the model must disregard. The Transformer’s self-attention mechanism directly addresses both requirements. It enables the model to dynamically weight input importance, effectively ignoring both irrelevant content and padding tokens. This property aligns with its theoretically established capacity to learn sparse functions (\citep{2022_Edelman_et_al}).

Thirdly, another challenge stems from the high-dimensional input, characterized by a large number of news articles $n$ (i.e., the model input length or number of tokens) and a large embedding dimensions $d_e$. Recent theoretical work suggests this challenge can be mitigated.  \cite{2024_Trauger_and_Tewari} demonstrate that, under suitable parameter constraints, the complexity of a Transformer model can even grow independently of the input dimension. In addition, \cite{2025_Mousavi-Hosseini_et_al} shows that Transformers can have smaller sample complexity that almost independent of the input length, implying superior data efficiency compared to feedforward or recurrent architectures. Inspired by these insights, we adopt a function class composed of Transformer blocks where the parameter norms are constrained. This design ensures that the model's complexity—and consequently,  the input dimensions can have little effect on the model complexity, as demonstrated in Corollary \ref{cor:convergence_rate}.

The Transformer architecture contains two components that we discuss in depth below: a multi-head self-attention mechanism (MSA) and a position-wise feed-forward network (FFN). We refer to these two components as sublayers. Each sublayer can be wrapped with a residual (skip) connection and layer normalization \citep{2016_He_et_al, 2016_Ba_Kiros_et_al}, which stabilize training and mitigate issues such as vanishing gradients. One MSA sublayer and one FFN sublayer together form a Transformer layer. A Transformer model can contain $L$ such layers stacked sequentially, indexed by $l=1, \ldots, L$. In the following, we denote the rectified linear unit (ReLU) activation function by $\sigma_{R}(x)=\max (x, \ 0)$ which operates element-wise, regardless of whether the input is a vector or a matrix. We also define the SoftMax function as
\begin{equation}
    \sigma_S: \mathbb{R}^{n \times n} \to [0,1]^{n \times n}, \quad Z \mapsto \left(\frac{e^{Z_{i,j}}}{\sum_{k} e^{Z_{k,j}}}\right)_{i,j=1}.
\end{equation}

We now zoom in on the Transformer architecture. We consider the Transformer architecture without residual connections and layer normalization in this paper. For a discussion on architectures including these components, please refer to \cite{2025_Jiao_et_al}. While these operations provide additional transformations within each sublayer, they do not alter the fundamental input–output mapping structure of the Transformer block \citep{2022_Edelman_et_al}. For completeness, we present in Appendix \ref{sec:transformer_details} the standard Transformer architecture, including residual connections and layer normalization.

We take the initial input to our Transformer-based neural network to be a matrix $Z_{j,t}\in\mathbb{R}^{d\times n}$ formed by concatenating the returns $R_{-j,t}$ and the text embeddings $E_{j,t}$. In practice, $R_{-j,t}$ and $E_{j,t}$ may not have compatible dimensions, so we apply a mapping $\Pi$ to obtain a representation of the desired size. In the empirical application in Section \ref{sec:application}, we consider a function $\Pi$ of the following form:
\begin{equation} \label{eq:concatenation}
    Z_{j,t} := \Pi(R_{-j,t}, E_{j,t}) = \left[
    \begin{array}{c}
        R_{-j,t} \cdot \mathbf{1}_{n}^{\top} \\ E_{j,t}
    \end{array}
    \right] \quad \in \mathbb{R}^{d \times n},
\end{equation}
where $\mathbf{1}_{n} \in \mathbb{R}^{n}$ is a vector of ones and $d = (J-1)+d_e$. This construction replicates the same return vector $R_{-j,t}$ across the $n$ columns and stacks it above each column of $E_{j,t}$. We refer to the resulting columns as \emph{return-augmented news embeddings}.

\subsubsection{Self-Attention Sublayer}
The self-attention sublayer allows each position (or token) in the input sequence to attend to, and aggregate information from, all other positions in the same sequence \citep{2017_Vaswani_et_al}. In our application, the input sequence is the concatenated vector of embedded news items and contemporaneous returns from other banks (see Equation \eqref{eq:concatenation}). The self-attention mechanism computes pairwise \textit{attention weights} between all elements in this sequence. Intuitively, these weights act as similarity or correlation scores, enabling the model to determine how much focus to place on a given news article when processing another (see Figure \ref{fig:words_score_matrix_news}). A weighted sum of the input values, using these attention weights, produces the sublayer's output, blending relevant information from across the entire input context.
\begin{figure}[t]
    \centering
    \caption{Example illustration of self-attention scores}
    \label{fig:words_score_matrix_news}
    \includegraphics[width=0.6\textwidth]{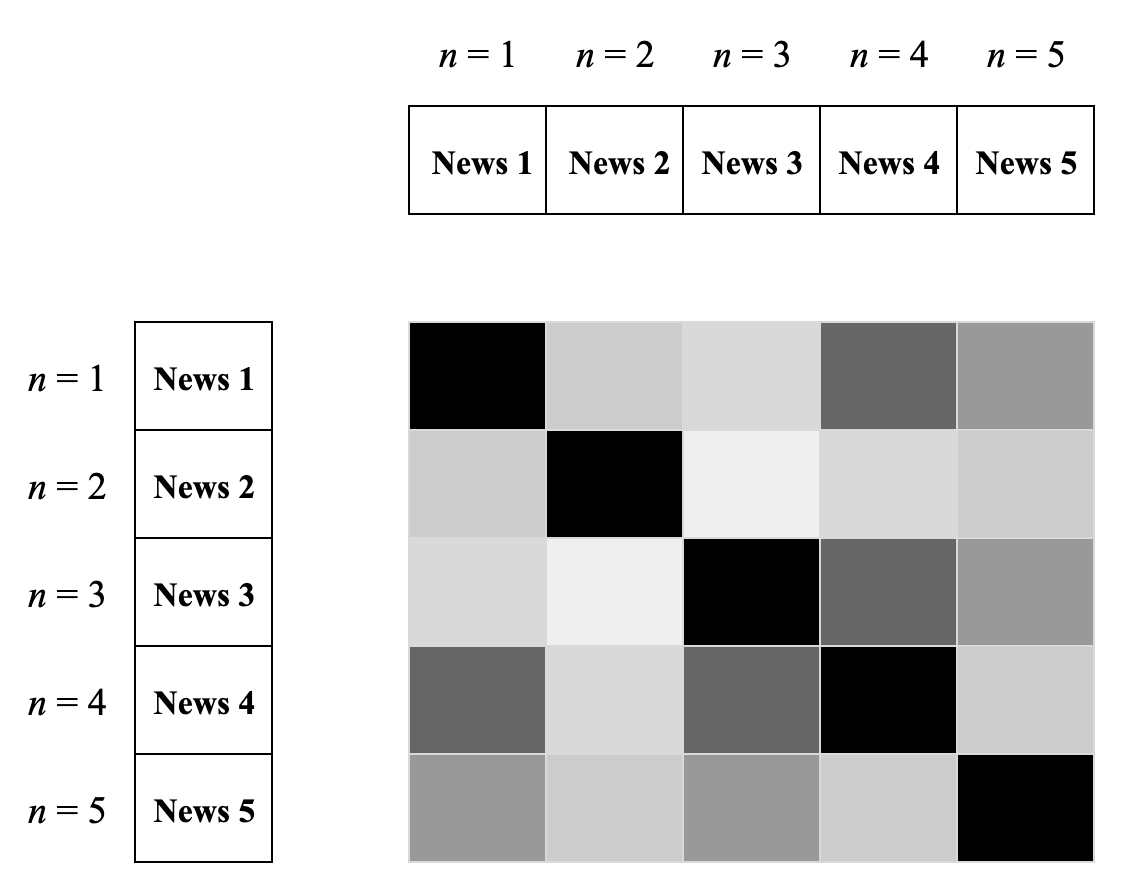}
    \caption*{\textit{Notes:} The self-attention mechanism computes pairwise \textit{attention weights} among all elements in the input sequence. In this figure, the attention mechanism learns similarity scores between (return-augmented) news items. Different cell colors correspond to different score values. Darker values indicate stronger interactions.}
\end{figure}

Multi-head self-attention extends the above mentioned idea by splitting the model dimensions $d$ into multiple heads. Let the number of heads be $H$. We define the $l$-th multi-head self-attention class $\mathcal{G}_{l}^{(MSA)}(n,d,H)$ as the set of map $g_l^{(msa)}: \mathbb{R}^{d \times n} \to \mathbb{R}^{d \times n}$ of the form
\begin{equation}
    g_l^{(msa)}(Z) = \sum_{h=1}^H W_{l,h}^{(O)} \, W_{l,h}^{(V)} Z \cdot \sigma_S \left(\frac{Z^{\top} {W_{l,h}^{(K)}}^{\top} W_{l,h}^{(Q)} Z}{\sqrt{d}} \right).
\end{equation}
For each head $h \in \{1, \ldots, H\}$, we denote by $W_{l, h}^{(V)}, W_{l, h}^{(K)}, W_{l, h}^{(Q)} \in \mathbb{R}^{\frac{d}{H} \times d}$, and ${W_{l,h}^{(O)}} \in \mathbb{R}^{d \times \frac{d}{H}}$ the value, key, query, and projection matrices, respectively. Both the query and key matrices are learned from data during training as weighted transformations of the input matrix $Z$. The term $Z^{\top} W_{l,h}^{(K)\top} W_{l,h}^{(Q)} Z$ captures pairwise interactions between all columns of $Z$, measuring how similar or relevant one observation (e.g., news item) is to another. Estimating \( W_{l,h}^{(K)\top} \) and \( W_{l,h}^{(Q)} \) helps us understand which inputs contribute to the risk of bank \( j \). For example, if there is a strong correlation between News~3 and News~4, the corresponding weights will reflect the importance of both components in determining the bank’s risk (Figure \ref{fig:words_score_matrix_news}). As another illustration, News~2 and News~3 exhibit weaker interactions relative to other pairs in Figure~\ref{fig:words_score_matrix_news}. After dividing by $\sqrt{d}$ and applying the softmax function to normalize the weights, the self-attention mechanism amplifies strong dependencies and filters out weak or noisy signals. 

In our risk analysis framework, the input matrix $Z_{j,t,s}$ jointly incorporates numerical return information and textual news content. By concatenating returns with news embeddings, we construct return-augmented news representations. The self-attention mechanism then learns the interactions among these return-augmented news inputs, allowing cross-bank spillovers and textual information to be modeled in a unified manner. In particular, the model identifies which latent information in bank returns and news articles the strongest influence on the risk exposure of firm $j$. Through this mechanism, the model uncovers dynamic transmission channels through which financial and informational shocks propagate across the system.

\subsubsection{Feed-Forward Network Sublayer}
Following the attention sublayer, a feed-forward network (FFN) sublayer applies the same operation to each column (i.e., token) of the input matrix. Its primary role is to provide additional non-linearity and transformation capacity. Recent work by \citet{2023_Sonkar_and_Baraniuk} shows that the FFN prevents the output from collapsing into a narrow subspace of the total embedding space, meaning they become highly correlated and lose their distinctiveness.

The $l$-th feed-forward network class $\mathcal{G}_{l}^{(FF)}(d, d_h)$ is defined as the set of maps $g_l^{(ff)}:\mathbb{R}^{d \times n} \to \mathbb{R}^{d \times n}$ of the form
\begin{equation}
    g_l^{(ff)}(Z) = W_l^{(2)} \, \sigma_{R} \left(W_l^{(1)} \cdot Z + b_l^{(1)} \right) + b_l^{(2)}.
\end{equation}
where $d_h$ denotes the hidden dimension of the feed-forward network. The parameters are $W_l^{(1)} \in \mathbb{R}^{d_h \times d}$, $b_l^{(1)} \in \mathbb{R}^{d_h}$ and $W_l^{(2)} \in \mathbb{R}^{d \times d_h}$, $b_l^{(2)} \in \mathbb{R}^{d}$. 

\subsection{Transformer-based Estimator} \label{sec:estimator}
\begin{figure}[!t]
    \centering
    \caption{Illustration of the Transformer-based CoVaR prediction model}
    \label{fig:transformer_details}
    \includegraphics[width=0.8\textwidth, keepaspectratio]{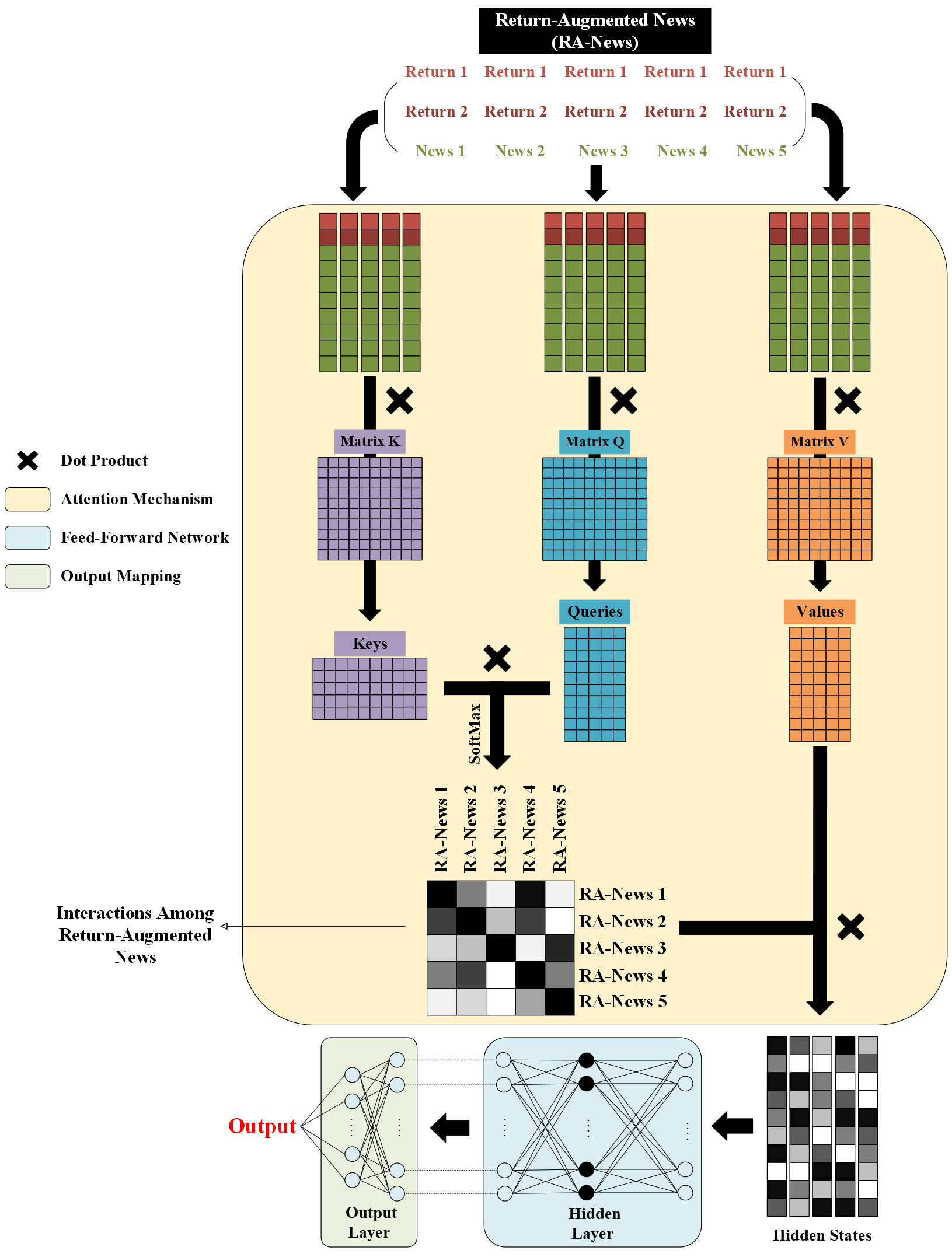}
    \caption*{\textit{Notes:} 
    The model takes numerical inputs (stock returns) and textual inputs (news articles). As an example, consider two bank returns and five news articles as model inputs. News are augmented with return information via concatenation, forming return-augmented news embeddings. The Transformer attention mechanism captures dependencies within the return-augmented news and learns latent representations (hidden states). These hidden states are subsequently passed through a feed-forward network (FFN) and a multilayer perceptron (MLP) to produce a scalar prediction.}
\end{figure}
We employ a Transformer-based estimator consisting of the previously introduced Transformer architecture followed by a multilayer perception (MLP). To bridge the Transformer output matrix $Z \in \mathbb{R}^{d \times n}$ and the MLP, we apply a linear projection. The linear projection maps the Transformer output onto a vector, which is then passed through the MLP to produce the scalar result. As shown in Figure~\ref{fig:transformer_details}, our estimator first lets the MSA learn the interactions among the return-augmented embeddings, which yields hidden representations that are then passed through the FFN and the MLP to produce the risk prediction.

Formally, we define the Transformer-block function class $\mathcal{G}^{(TF)}(n, d, H, d_h, L)$ as the set of maps $g^{(tf)}: \mathbb{R}^{d \times n} \to \mathbb{R}^{d \times n}$ formed by the composition of $L$ Transformer layers. Each layer consists of a multi-head self-attention sublayer and a feed-forward network sublayer:
\begin{equation}
    g^{(tf)} = g_L^{(ff)} \circ g_L^{(msa)} \circ \cdots \circ g_1^{(ff)} \circ g_1^{(msa)}.
\end{equation}
And we define the MLP function class $\mathcal{G}^{(MLP)}(D, \widetilde{d}_{m})$ as the set of maps $g^{(mlp)}: \mathbb{R}^d \to \mathbb{R}$ of the form
\begin{equation} \label{eq:mlp_class}
    g^{(mlp)}(x) = W_{D} \sigma_R( \dots \sigma_R(W_1 x + b_1) \dots) + b_D,
\end{equation}
where the dots denote repeated compositions across layers, $D$ is the MLP network depth and $\widetilde{d}_{m}:=\left(d, d_{m}^{(1)}, \ldots, d_{m}^{(D-1)}, 1\right)$ represents the MLP network width\footnote{Depth refers to the number of layers in the MLP, and width indicates the number of neurons within each of those layers.}. If all hidden layers have the same width, we denote it by $d_m$. For $i=1, \ldots, D$, we have $W_i$ as a $d_{m}^{(i+1)} \times d_{m}^{(i)}$ weight matrix and $b_i$ as a vector of dimension $d_{m}^{(i)}$. Thus, the full class of our Transformer-based neural network functions of a given architecture $(n, d, H, d_h, L, D, \widetilde{d}_{m})$ takes the following form:
\begin{equation} \label{eq:function_class}
    \mathcal{G} := \mathcal{G}(n, d, H, d_h, L, D, \widetilde{d}_{m}) := \bigl\{ g: \mathbb{R}^{d \times n} \to \mathbb{R} \mid g := g^{(mlp)} \circ (g^{(tf)} \cdot w) \bigl\}.
\end{equation}

For convenience, we additionally define for a given concatenation function $\Pi$ (see, e.g., equation (\ref{eq:concatenation})),
\begin{equation} \label{eq:function_class II}
    \mathcal{F} := \mathcal{F}(n, d, H, d_h, L, D, \widetilde{d}_{m}) := \bigl\{ g \circ \Pi: g \in \mathcal{G} \bigl\},
\end{equation}
so that for any $f \in \mathcal{F}(n, d, H, d_h, L, D, \widetilde{d}_{m})$, there exists $g \in \mathcal{G}(n, d, H, d_h, L, D, \widetilde{d}_{m})$ such that
\begin{equation} \label{eq:g_and_f}
    f(R_{-j,t},E_{j,t}) = g \circ \Pi (R_{-j,t},E_{j,t}) = g(Z_{j,t}).
\end{equation}
The Transformer-based quantile estimator is the solution to the following minimization problem for a given class $\mathcal{F}$:
\begin{equation}
     \widehat{f}_{j, \tau}= \argmin_{f \in \mathcal{F}} \sum_{t=1}^{T} \rho_{\tau} \Bigl( R_{j,t} - f\left(R_{-j,t},\, E_{j,t} \right)\Bigr),
\end{equation}
We apply the stochastic gradient descent (SGD) algorithm, see, e.g.,  \citep{2022_Garg_et_al, 2023_Bai_et_al, 2024_Li_et_al}, for solving the optimization problem in the later applications.

In summary, we construct a Transformer-based quantile regression model for CoVaR that combines peer returns and financial news within a unified architecture (see Figure \ref{fig:transformer_details}). By explicitly defining the model architecture and its associated function class, we establish a foundation for analyzing both the theoretical properties and empirical performance of the proposed estimator. {In the following section, we introduce the consistency property of our CoVaR estimator.}



\section{Theoretical Results} \label{sec:theory}


\textbf{Notation.} We consider the following notation. Let $\mathbb{R}^{+} := \{x \in \mathbb{R} : x > 0\}$ denote the set of positive real numbers. For vectors, $\|\cdot\|_{p}$ denotes the $\ell_p$ norm, and $\|\cdot\|_\infty := \max_{1 \le i \le d} |\cdot|$ is the sup-norm. For a matrix $A \in \mathbb{R}^{a \times b}$, $\|\cdot\|_{2}$ denotes the spectral norm. Let $\|\cdot\|_{p, q}$ denote the ($p, q$) matrix norm defined by first taking the $p$-norm of each row of the matrix, and then taking the $q$-norm of the resulting vector. We also define $\|A\|_0 = \left|\left((i, j): A_{i, j} \neq 0\right)\right|$ and $\|A\|_{\infty} := \max _{i=1, \ldots, p, j=1, \ldots, q}\left|A_{i, j}\right|$. For two sequences, $\{a_n\}_n$ and $\{b_n\}_n$, we write $a_n \lesssim b_n$ if there exists a constant $C$ such that $a_n \leq C b_n$ for all $n$. Moreover, for deterministic sequences $\{a_n\}_n$ and $\{b_n\}_n$ with $b_n>0$, we write $a_n = O(b_n)$ if there exists a constant $C>0$ such that $|a_n| \le C b_n$ for all sufficiently large $n$. For random variables $\{X_n\}_n$ and deterministic $\{b_n\}_n$, 
we write $X_n = O_p(b_n)$ if $X_n/b_n$ is bounded in probability, and write $X_n \lesssim_P b_n$ if $X_n = O_p(b_n)$.


Define the functional class $V_{(\mathcal{F}, \|\cdot\|_{\infty})}(\delta)$ as the logarithm of the $\delta$-covering number of a function class $\mathcal{F}$ under the sup-norm, which measures the complexity of $\mathcal{F}$ by counting the minimum number of balls of radius $\delta$ needed to cover it. Specifically, it is defined as $V_{(\mathcal{F}, \|\cdot\|_{\infty})}(\delta) := \log N_{\infty}(\delta,\mathcal{F})$, where $N$ is the smallest number of functions $f_1, \dots, f_N$ such that every $f \in \mathcal{F}$ is within $\delta$ (in sup-norm) of some $f_i$. 


Our main theorem is in the fashion of \cite{2019_Suzuki}, \cite{2020_Schmidt-Hieber} and \cite{2022_Padilla_et_al} but we adapt these results to the CoVaR framework for Transformers. We first outline the required assumptions.

\begin{mainassumption}[Cumulative Distribution Function Boundedness] \label{assp:cdf} 
    There exists a constant $\kappa>0$ such that for any $\delta_t>0$ satisfying $\max_t |\delta_t|_{\infty} \leq \kappa$ we have that, a.s.,  for all $j = 1,\dots,J$,
    \begin{equation}
          \left| F_{R_{j,t} \mid x_{j,t}} (f_\tau^*(x_{j,t})+\delta_t) - F_{R_{j,t} \mid x_{j,t}} (f_\tau^*(x_{j,t})) \right| \, \geq \, \underline{p} \cdot \left|\delta_t\right|
    \end{equation}
    for $t=1, \ldots, T$ and for some constant $\underline{p}>0$ with $F_{R_{j,t} \mid x_{j,t}}$ being cumulative distribution function (CDF) of $R_{j,t}$ conditioning on $x_{j,t}$. We also require that  for all $j = 1,\dots,J$,
    \begin{equation}
 \sup_{z \in \mathbb{R}} p_{R_{j,t} \mid x_{j,t}}(z) \leq \bar{p}, \quad \quad \text{a.s.}
    \end{equation}
    for some constant $\bar{p}>0$, where $p_{R_{j,t} \mid x_{j,t}}$ is the probability density function of $R_{j,t}$ conditioning on $x_{j,t}$.
\end{mainassumption}

The above CDF Boundedness assumption ensures that the conditional distribution of returns is sufficiently smooth and strictly increasing near the quantile of interest \citep{2022_Padilla_et_al}. This allows well-defined CoVaR values and ensures that risk estimates respond meaningfully to fluctuations in the underlying data. 

\begin{mainassumption}[Lipschitz] \label{assp:Liptschiz}
Denote $\mathcal{E}_j$ the support of $E_{j,t}$, and for all $j = 1,\dots,J$, the target function in (\ref{eq:f_target}) satisfies that for all $r_1, r_2$ within the support of $R_{-j,t}$,
    \begin{equation}
 \sup_{e\in \mathcal{E}_j}  \big|f_{j,\tau}^*(r_1,e) - f_{j,\tau}^*(r_2,e)\big| \lesssim  \big\|r_1 - r_2\big\|_1,
    \end{equation}
    and for any $f \in \mathcal{F}$ for the given class $\mathcal{F}$ used for fitting, we have that
    \begin{equation}
 \sup_{e\in \mathcal{E}_j}   \big|f(r_1,e) - {f}(r_2,e)\big| \leq  L_{\mathcal{F}} \big\|r_{1} - r_{2}\big\|_1,
    \end{equation}
    with $L_{\mathcal{F}}$ depending upon the class $\mathcal{F}$.   
\end{mainassumption}


The second part in Assumption \ref{assp:Liptschiz} is trivial satisfied if we impose norm bounds over parameters used in $\mathcal{F}$, as the class $\mathcal{F}$ consists of functions $f$ formed by composing layers that are constructed via linear maps and Lipschitz nonlinear activations/functions. By imposing norm bounds on the parameters of each layer, we ensure each layer is Lipschitz-continuous. Consequently, the overall Lipschitz constant $L_{\mathcal{F}}$ for any $f \in \mathcal{F}$ is bounded by the product of the layer-wise constants. See Appendix~\ref{sec:corollary} for an explicit example where we consider a one-layer Transformer followed by an MLP with spectral norm bound $B$.

\begin{mainassumption}[Consistency of VaR]
    \label{assp:var}
    The estimated VaR is consistent, i.e., for all banks $j$, 
    \begin{equation}
  \left\| \widehat{\operatorname{VaR}}_{-j,t}^{\tau} - \operatorname{VaR}_{-j,t}^{\tau} \right\|_1 = O_p(b_{T}),
    \end{equation}
    where $b_T$ is a deterministic sequence such that $b_T \to 0$ as $T \to \infty$.
\end{mainassumption}


Under the Lipschitz Continuity and VaR Consistency assumptions, the sampling errors from VaR estimation vanish as the sample size increases, provided the overall Lipschitz constant $L_{\mathcal{F}}$ is of sufficiently small magnitude. This guarantees that the "plug-in" step does not introduce persistent bias, allowing the CoVaR estimator to remain stable and accurate \citep{2019_Patton_et_al, 2025_Dimitriadis_and_Hoga}.

\begin{maintheorem}[Out-of-Sample CoVaR Estimator Consistency] 
    \label{thm:covar_convergence}
    Suppose that Assumptions \ref{assp:cdf}--\ref{assp:var} and conditions in Lemma \ref{lem:theo1} hold, then the CoVaR prediction satisfies:
    \begin{equation}
        \max_{1\leq j\leq J} \left[ \left|\operatorname{CoVaR}_{j,t'}^{\tau} - \widehat{\operatorname{CoVaR}}_{j,t'}^{\tau}\right|\right] \lesssim_{P} \max_{1\leq j\leq J}\widetilde{R}_{j,\delta_T,T} + L_{\mathcal{F}} \, b_T,
    \end{equation}
    where $\widetilde{R}_{j,\delta_T,T} =\; \inf_{f  \in \mathcal{F}} \Bigl\|f  - f_{j,\tau}^*\Bigl\|_{\infty}^2 + \Biggl( \delta_T +   \sqrt{\frac{V_{\left(\mathcal{F},\|\cdot\|_{\infty}\right)}(\delta_T)}{T}} \Biggl)$ with $ \delta_T>0,  \delta_T \rightarrow 0, \delta_T^2 T \rightarrow \infty$ as $T\to \infty$, and $\widehat{\operatorname{CoVaR}}_{j,t'}^{\tau} = \widehat{f}_{j,\tau}(\widehat{\operatorname{VaR}}_{-j,t'}^{\tau}, E_{j,t'})$.
\end{maintheorem}
\begin{proof}
    See Appendix Section \ref{sec:covar_consistency}.
\end{proof}

Theorem \ref{thm:covar_convergence} connects the complexity of the estimator to an upper bound on the generalization error of our estimator. In this sense, it is comparable to Theorem 2.6 in \cite{2020_Hayakawa_and_Suzuki}, but is established here under a quantile-loss optimization framework. The result is useful for deriving convergence rates once we impose structure on both the target function and the function class $\mathcal{F}$ used for fitting. Corollary \ref{cor:convergence_rate} further shows that, under explicit sparsity and smoothness restrictions on the target function, one can obtain a rate whose dependence on the input dimension $d$ and input length $n$ is at most logarithmic.

Although we acknowledge that this rate stated in Corollary \ref{cor:convergence_rate} can be further sharpened, it already indicates that a plausible approximation remains feasible even with relatively high-dimensional input. We defer the full proof to Appendices \ref{sec:proof_cover} and \ref{sec:bias}, but first briefly summarize the main mechanisms that make such a rate attainable. 

In Corollary \ref{cor:convergence_rate}, we follow two core setups similar to the ones employed in \cite{2022_Edelman_et_al}. First, the complexity of the function class $\mathcal{F}$ is governed by the combined capacity of the Transformer blocks and the MLP output layer, and we use norm-based Transformer blocks with sparse MLP, where the $\ell_{2,1}$ norms of the Transformer weights are explicitly bounded. This enables the model complexity to scale only logarithmically in \(d\) and \(n\), both of which may be large in practice (e.g., \(d = 71\) and \(n = 97\) in our application of Section~\ref{sec:application}), as captured by the bound
\begin{equation}\label{eq:variance-cover-simplified}
    V_{\left(\mathcal{F},\|\cdot\|_{\infty}\right)}(\delta)  = O \left( d_{m} \log(d_m) \cdot \log \Big( \delta^{-1} D \, d \, d_{m}^D\Big) + \frac{\log(dn)}{\delta^2} \right),
\end{equation}  
where $d_m, D$ are the network width and depth of MLP that possibly grow with $T$, e.g., in Corollary \ref{cor:convergence_rate} we choose $D \lesssim \log T$. Second, we assume that the target function has a low-dimensional sparse structure: it depends only on an unknown subset of $s$ relevant input coordinates with $s$ being a finite number, and is $\beta$-H\"older smooth. Together, these assumptions imply that self-attention can identify and aggregate the relevant information from the high-dimensional input, while the MLP only needs to learn a smooth mapping on an $s$-dimensional signal.  Hence, the following bias rate is attainable under Assumptions~\ref{assp:lemma_target_function} and \ref{assp:lemma_recoverability} for a given class $\mathcal{F}$ satisfying the aforementioned complexity bound:
    \begin{equation}
        \inf_{f \in \mathcal{F}} \Bigl\|f  - f_{j,\tau}^*\Bigl\|_{\infty}^2 = O\left( d_{m}^{-\beta / s} \right).
    \end{equation}

Based on the bias-variance decomposition provided in Theorem \ref{thm:covar_convergence} and results outlined above, we can achieve the following convergence rate.
\begin{corollary}[Convergence Rate] \label{cor:convergence_rate} 
    Suppose that the assumptions of Theorem \ref{thm:covar_convergence} and Lemma \ref{lemma:edelman_B2_continuous} hold with $b_T = O(T^{-\alpha})$ and $\alpha\ge \frac{2\beta}{s+4\beta}$, then out-of-sample CoVaR predictor achieves the convergence rate: 
    \begin{equation} 
        \max_{1\leq j\leq J} \left| \operatorname{CoVaR}_{j,t'}^{\tau} - \widehat{\operatorname{CoVaR}}_{j,t'}^{\tau} \right| = O_p\left( C_{n,d,T}(T^{-\frac{2\beta}{s+4\beta}} + T^{-1/4}) \right),
    \end{equation}
     where $ C_{n,d,T}$ is a rate that grows logarithmically with $n$, $d$, and $T$.
\end{corollary}
\begin{proof}
    See Appendix \ref{sec:corollary}.
\end{proof}

\section{Empirical Study} \label{sec:application}


\subsection{Data}

We study eight (see Table~\ref{tab:gsib_banks}) U.S. global systemically important banks (G-SIBs), as identified by the Financial Stability Board (FSB) with the Basel Committee on Banking Supervision (BCBS) and national regulators. These banks are considered as ``too big to fail" because of their role in the global financial system. Our sample covers daily log returns from October 2006 to November 2013, giving $T = 1776$ observations. The period includes major stress events or market drawdowns such as the 2008 financial crisis, the Goldman Sachs SEC lawsuit in 2010, the Greek debt crisis, and the U.S. debt ceiling disputes of 2011 and 2012. The sample period ensures us the diversity in news topics to study a set of systemic risks.
\begin{table}[!t]
    \centering
    \caption{List of U.S. Global Systemically Important Banks}
    \label{tab:gsib_banks}
    \begin{tabular}{l r}
        \toprule
        \textbf{Ticker} & \textbf{Bank Name} \\
        \midrule
        GS   & Goldman Sachs \\
        WFC  & Wells Fargo \\
        JPM  & JPMorgan Chase \\
        BAC  & Bank of America \\
        C    & Citigroup \\
        BK   & Bank of New York Mellon \\
        MS   & Morgan Stanley \\
        STT  & State Street \\
        \bottomrule
    \end{tabular}
\end{table}

We use a set of common macro-state variables as risk factors for VaR estimation in \eqref{eq:var} \citep{2016_Adrian_and_Brunnermeier, 2016_Haerdle_et_al}: the Implied Volatility Index (VIX), daily returns on the S\&P 500, the spread between Moody’s Baa Corporate Bond Yield and the 10-Year Treasury yield, and the yield spread between the 10-Year and 3-Month Treasury yields\footnote{VIX, S\&P 500 returns, and stock prices are from Yahoo Finance. Credit spreads and yield spreads are from the Federal Reserve Bank of St. Louis (FRED).}. For CoVaR estimation, we use both bank stock returns with the financial news articles released from \textit{Reuters}. The news dataset contains 105,359 articles, originally collected by \citet{2014_Ding_et_al}. After removing duplicates, we retain 97,264 articles.

On average, 38 articles are published per day, with daily counts ranging from 1 to 97 and a median of 39. Figure~\ref{fig:daily_frequency} shows that article volume rises sharply during market stress and drawdowns, such as the 2008 crisis and the 2011--2012 U.S. debt ceiling standoff. These spikes coincide with declines in the S\&P 500, suggesting that media coverage intensifies in volatile periods. The average article length is 449 words, with a median of 415.
\begin{figure}[!t]
    \centering
    \caption{Daily news article frequency and stock market index}
    \label{fig:daily_frequency}
    \includegraphics[width=0.95\textwidth]{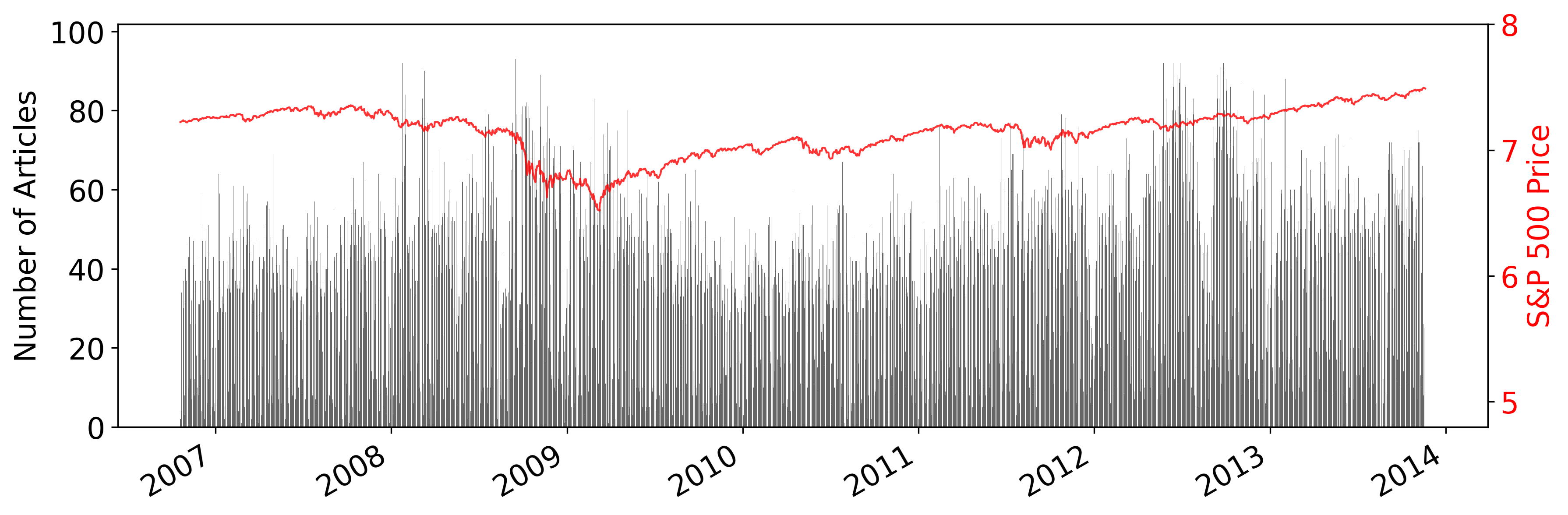}
    \caption*{\textit{Notes:} The left y-axis shows the daily frequency of \textit{Reuters} financial news articles.
    The right y-axis shows the logarithm of the S\&P 500 price index.}
\end{figure}

\subsection{Data Preprocessing}

{Preparing the input for our model proceeds in three steps. First, we employ the pre-trained embedding model \texttt{gemini-embedding-001} as the embedding mapping $\pi_{e}$ to convert $n$ ($n=97$) textual news used at time $t$ into numerical news embeddings $E_{j,t} \in \mathbb{R}^{d_{e} \times n}$ (we choose $d_{e}=64$). {For days on which $n$ is smaller than 97, we apply padding to the news embeddings (see the explanation of padding in the footnote in Section~\ref{sec:transformer}).} We consider all news articles observed up to time $t$ within a five-day look-back window. This choice is motivated by \cite{2013_Garcia}, who documents that market reactions to news typically reverse within four trading days, suggesting that such an effect is short-lived and largely driven by non-fundamental factors.  Second, because the embeddings do not inherently have temporal order, we inject positional information to preserve the chronological information of the news sequence. Following \cite{2017_Vaswani_et_al}, we add fixed sinusoidal positional encodings (PE)\footnote{The positional encoding for position $k$ and dimension $i$ is defined as:
$$\begin{aligned}
    PE(k, 2i) & =\sin \left(\frac{k}{10000^{2 i / d_{e}}}\right), \\
    PE(k, 2i+1) & =\cos \left(\frac{k}{10000^{2 i / d_{e}}}\right),
\end{aligned}$$
where $k$ is the temporal position within the $5$-day look-back window. These encodings are scaled by the norm of the input embeddings.} to the sequence $E_{j,t}$. This preserves the temporal sequence information of news articles. Third, we augment the news embeddings with numerical return information by concatenating the returns of the remaining $J-1$ ($J = 8$) banks with $E_{j,t}$, yielding the return-augmented news embeddings $Z_{j,t} = \Pi(R_{-j,t}, E_{j,t})$. We then take $Z_{j,t}$ as our model input, which consists of 97 input tokens, each with 71 (i.e., 64 plus 7) dimensions.}

As mentioned above, we use \texttt{gemini-embedding-001}, a pre-trained embedding model released in June 2025 by \textit{GoogleAI}. This model generates dense vectors of 3072 dimensions. It takes a maximum of 2048 tokens per input, which corresponds to about 1228--1638 English words\footnote{According to GoogleAI documentation, 100 tokens are about 60–80 English words.}. The majority of articles in our dataset are shorter than this limit, with only about 1\% exceeding it (see Figure~\ref{fig:article_length}). The longest article has 9471 words.
\begin{figure}[!t]
    \centering
    \caption{Distribution of article lengths}
    \label{fig:article_length}
    \includegraphics[width=0.6\textwidth]{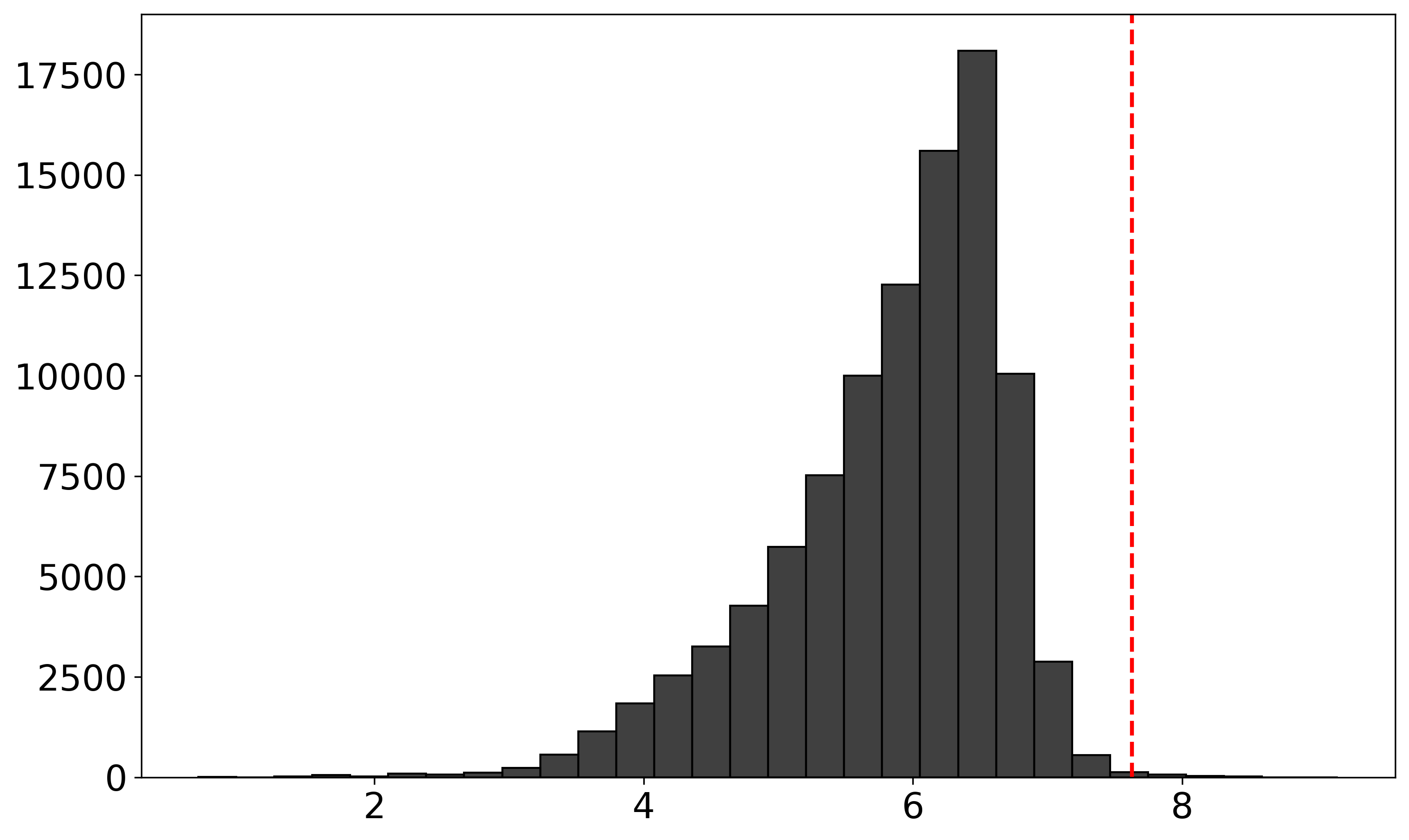}
    \caption*{\textit{Notes:} Article length is measured in number of words. The x-axis shows the logarithm of article length, and the y-axis shows frequency. The dashed red line indicates the 1228 word cutoff, corresponding to 2048 tokens.}
\end{figure}

Pre-trained embeddings are known to be effective general-purpose semantic representations \citep{2025_Zhao_et_al} and have shown good performance in various natural language processing tasks such as semantic search, text classification, and clustering \citep{2025_Lee_et_al}. An important property of \texttt{gemini-embedding-001} is that the embedding dimensions are ordered by importance, with the most informative content appearing in the first dimensions. This property enables dimensionality reduction by keeping only the top leading dimensions of each vector embedding \citep{2024_Kusupati_et_al, 2024_Kim_et_al}. Following this idea, we keep the top 64 dimensions of each embedding, since using all dimensions would make the model too large and require much more training data. We can then represent a collection of $n$ articles by a semantic embedding matrix with size $\mathbb{R}^{d_{e} \times n}$, where each row is a $d_{e}=64$ dimensional embedding. This matrix is then used as the input in the downstream risk prediction task. 

To validate that the embeddings capture meaningful semantic patterns in the text data, we study their spatial structure in the embedding space. Since the embedding model is designed to map semantically similar texts to nearby points in the vector space, we expect that texts with similar content will produce embeddings with smaller distances (e.g. Euclidean, cosine distances etc.). For example, the sentences ``The market is going to crash" and ``The market is expected to plunge" should yield embedding vectors that are close to each other, reflecting their semantic similarity.

First, we apply the t-distributed Stochastic Neighbor Embedding (t-SNE) algorithm to project the high-dimensional embeddings into two dimensions for visualization purposes. t-SNE is one of the most popular methods for revealing intrinsic structures in high-dimensional datasets, including trends and patterns, through a nonlinear dimensionality reduction technique \citep{2022_Cai_and_Ma}. This is useful for high-dimensional data that lie on several different but related low-dimensional manifolds \citep{2008_Van_der_Maaten_and_Hinton}. For example, news items can report various sub-events or highlight different facets within the same topic. 

\begin{figure}[!t]
    \centering
    \caption{t-SNE visualization of Goldman Sachs–related news articles}
    \label{fig:tsne}
    \includegraphics[width=0.95\textwidth]{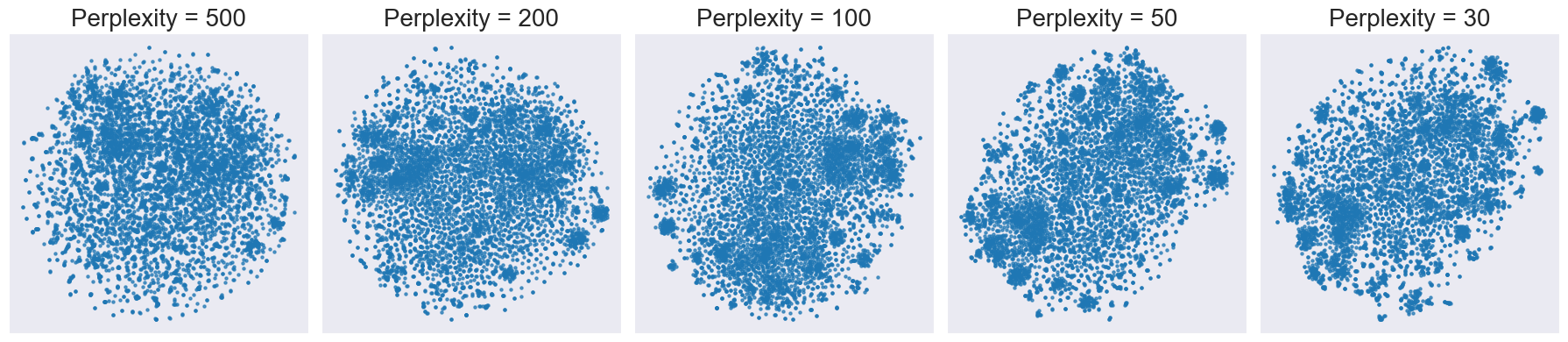}
    \caption*{\textit{Notes:} The figure shows t-SNE embeddings of 6{,}124 news articles related to Goldman Sachs under different perplexity parameter settings.}
\end{figure}
t-SNE includes a parameter called ``perplexity”, which balances the focus between local and global structures among data points. Local structure refers to the relationships between each point and its nearest neighbors, capturing fine-grained similarities. Global structure refers to the broader arrangement of clusters relative to one another. Lower perplexity values emphasize local relationships, producing smaller, tighter clusters, whereas higher perplexity values emphasize global structure, resulting in larger, more diffuse clusters. As shown in Figure \ref{fig:tsne}, clusters become smaller and more concentrated with lower perplexity values, and larger and more diffuse when the algorithm emphasizes global structure. However, in all cases, the visualizations reveal a pattern of clustering, suggesting that the embeddings form coherent groups in the reduced space.

Next, we apply the KMeans clustering algorithm directly in the original high-dimensional embedding space. Since KMeans is an unsupervised learning algorithm, the number of clusters must be specified in advance. Choosing a smaller number of clusters results in broader, more general themes, while a larger number yields more granular clusters with narrower topics. As an example, we examine all 6124 news articles related to Goldman Sachs\footnote{These articles are selected based on whether the title, summary, or body text contains the phrase “Goldman Sachs.”}. Figure \ref{fig:kmeans} shows the result of applying KMeans with 30 clusters to these articles. For each cluster, we use the LLM \texttt{gemini-2.0-flash} to assign a common theme to the 15 most representative articles. These representative articles are defined as those closest to the cluster centroid. The LLM generates a thematic summary based on the topic shared among the 15 selected articles.

\begin{figure}[!t]
    \centering
    \caption{KMeans clustering on Goldman Sachs–related news (30 clusters)}
    \label{fig:kmeans}
    \includegraphics[width=0.9\textwidth]{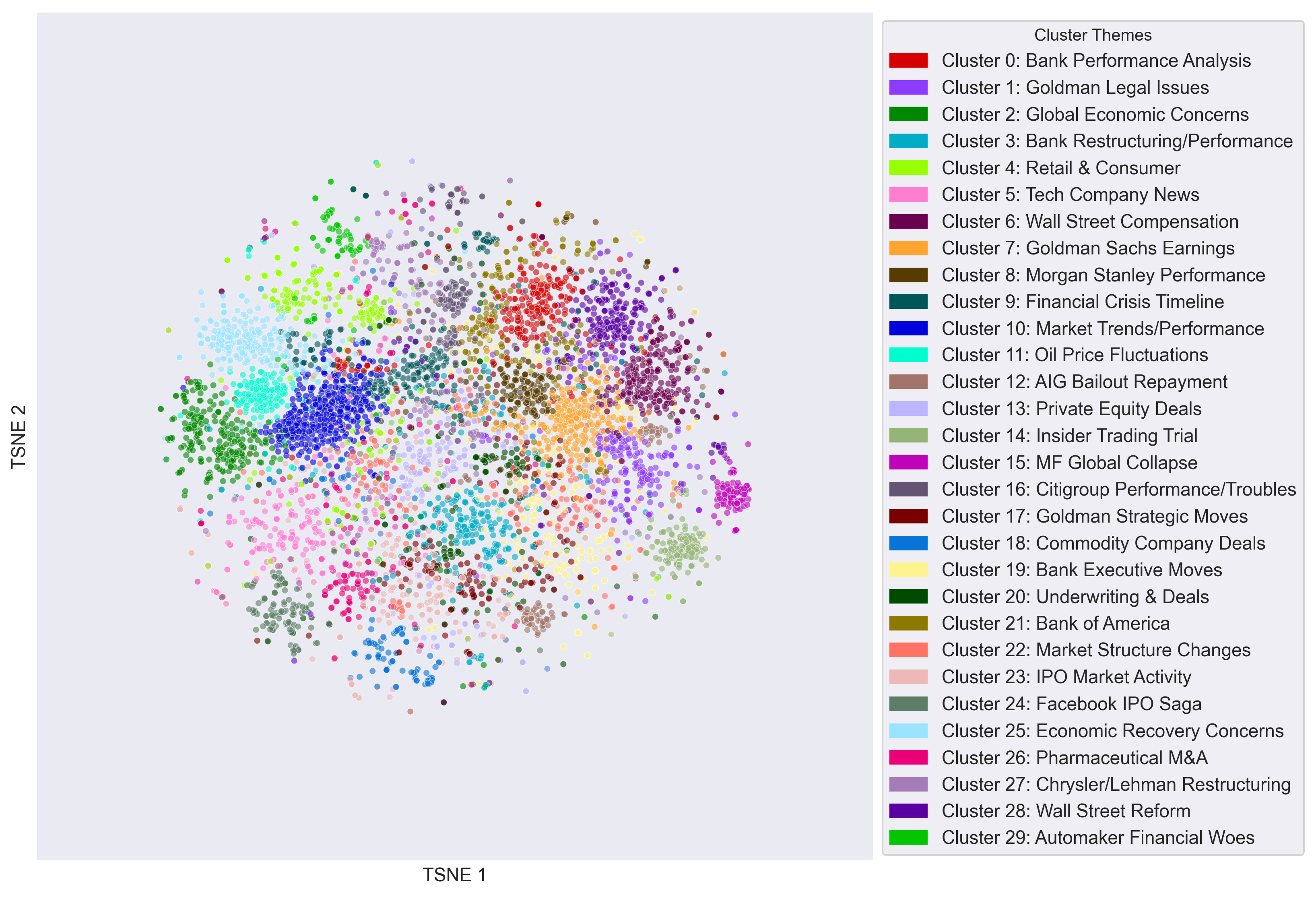}
    \caption*{
    \textit{Notes:} The clustering is visualized using the t-SNE algorithm, and each cluster represents a group of semantically similar news articles.}
\end{figure}
Taken together, the t-SNE visualizations and the KMeans clustering results provide qualitative evidence that the embeddings capture meaningful semantic structure in the news corpus. The presence of stable clustering patterns across different perplexity settings, along with the emergence of interpretable themes within clusters in the original embedding space, suggests that semantically related articles are mapped to nearby regions.

\subsection{CoVaR Prediction}
We input the high-dimensional return-augmented news embeddings into our Transformer-based model with one Transformer-block layer and a two-layer MLP. We set the hidden dimensions of both the FFN and the MLP to $d_h = d_m = 64$. We estimate CoVaR in two steps as described in Section \ref{sec:2_methodology}. First, we estimate $\operatorname{VaR}$ for all banks using linear quantile regression on macro-state variables, following \citet{2016_Adrian_and_Brunnermeier} and \citet{2016_Haerdle_et_al}. Second, we estimate $\operatorname{CoVaR}$ for our target bank using a Transformer-based model that incorporates both numerical (i.e., stock returns) and textual information (i.e., financial news articles). 

Our primary objective is to evaluate the added value of incorporating textual information in the second step of CoVaR framework. The dependent variable is the log of the stock return of the target bank. The independent variables include the contemporaneous log stock returns of other banks and news embeddings from the past five days. 

To evaluate out-of-sample performance, we split the dataset into three subsets: training ($\mathcal{D}_{1}$), validation ($\mathcal{D}_{2}$), and test ($\mathcal{D}_{3}$). The training set covers 40\% of the full dataset and is used to fit the model. The validation set takes 20\% and is used for regularization (i.e., early stopping) and hyperparameter tuning. We use the remaining 40\% as the test set to evaluate out-of-sample predictive performance of our Transformer. Each subset contains at least one major stock market downturn. We highlight the first and third big market downturns with gray-shaded areas in the following figures and skip the second one, since the second occurs in the validation set and is not of interest to us. The first downturn corresponds to the 2008 financial crisis, which is determined by the NBER (National Bureau of Economic Research) recession period. The third downturn corresponds to the 2011 U.S. debt ceiling crisis, which we determine using quarterly windows. Note that we train the model only once and then use the trained model to perform one-step-ahead predictions on the test data in a rolling-window manner. This means that for each test period, we use all available information up to time $t$ to predict the outcome at time $t+1$.

We train our Transformer-based model using SGD with mini-batches \citep{2025_Marek_et_al}. In the hyperparameter search, we experiment with three learning rates (0.00015, 0.0015, and 0.015) and three batch sizes (32, 64, 128). We then select the optimal hyperparameters that achieve the lowest loss on the validation data. We choose 200 training epochs with early stopping when the loss does not decrease for 50 consecutive epochs. Because our Transformer incorporates textual information, the optimization becomes much more difficult, and the training objective may include multiple local minima. We acknowledge that some training runs may not converge to the same solution. 

\begin{figure}[t]
    \centering
    \caption{VaR, CoVaR, and $\Delta$CoVaR estimates for Goldman Sachs at $\tau=5\%$}
    \label{covar_gs}
    \includegraphics[width=0.95\textwidth]{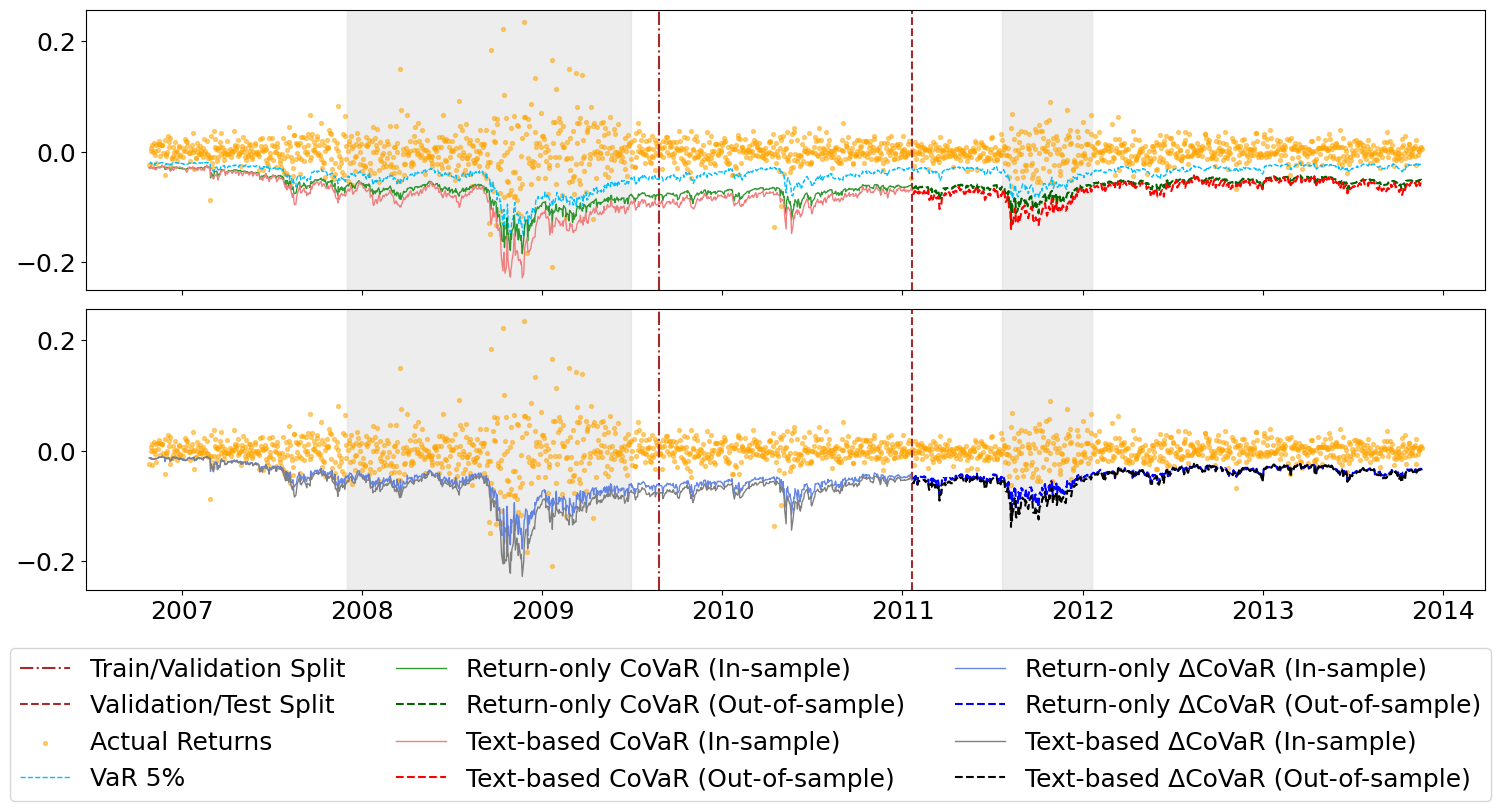}
    \caption*{
    \textit{Notes:} The estimates are shown with and without textual information. 
    The time series is divided into three parts: the training period (in-sample), the validation period (in-sample), and the test period (out-of-sample), separated by red lines. The \textbf{top panel} shows returns, VaR, and CoVaR, while the \textbf{bottom panel} shows $\Delta$CoVaR. Grey areas indicate major market downturns.}
\end{figure}
Figure~\ref{covar_gs} shows the $\operatorname{CoVaR}$ and $\operatorname{\Delta CoVaR}$ estimates for Goldman Sachs, with and without textual information (see Figure \ref{fig:all_banks} in Appendix \ref{sec:tables_and_figures} for other assets). {$\operatorname{\Delta CoVaR}$ is an alternative risk measure proposed by \cite{2016_Adrian_and_Brunnermeier} that builds upon the CoVaR framework. It quantifies the spillover effect by measuring the change in the system's risk
(CoVaR) when a specific bank $j$ shifts from its median state to a state of financial distress. We denote $\operatorname{\Delta CoVaR}$ as
\begin{equation}
    \operatorname{\Delta CoVaR}_{j, t}^\tau = \operatorname{CoVaR}_{j, t}^\tau\left(\operatorname{VaR}_{-j, t}^\tau\right) - \operatorname{CoVaR}_{j, t}^\tau\left(\operatorname{VaR}_{-j, t}^{0.5}\right).
\end{equation}
Importantly in this context, the $\tau$ used for the conditional systemic risk (CoVaR) and the $\tau$ defining the banks' distress levels (VaR) can be distinct.}

The text-based CoVaR is estimated using our Transformer-based model introduced in Section~\ref{sec:transformer}, and the non-text-based (i.e., return-only) estimates are computed from a two-layer MLP neural network with hidden dimensions 64. During periods of non-crisis and normal market activity, the gap between VaR and CoVaR tends to widen, reflecting the silent accumulation of systemic risk. This supports the perspective of \citet{2016_Adrian_and_Brunnermeier}, who argue that systemic risk accumulates silently during calm periods and materializes during crises. The latter is evident during market downturns, where the gap narrows as individual institutions approach their VaR thresholds, reflecting the realization of systemic risk.

Comparing the Transformer-based CoVaR model with text against the baseline MLP without text shows that incorporating financial news leads to noticeably different risk estimates during periods of market stress (See Figure \ref{fig:covar_diff}). The largest absolute differences appear during the 2008 crisis, especially in October, when the differences reach about 0.06--0.065 for both CoVaR and $\Delta$CoVaR. At that time, the global markets were highly unstable and concern about the viability of major financial institutions intensified. 
\begin{figure}[!t]
    \centering
    \caption{Differences in CoVaR and $\Delta$CoVaR estimates for Goldman Sachs}
    \label{fig:covar_diff}
    \includegraphics[width=0.95\textwidth]{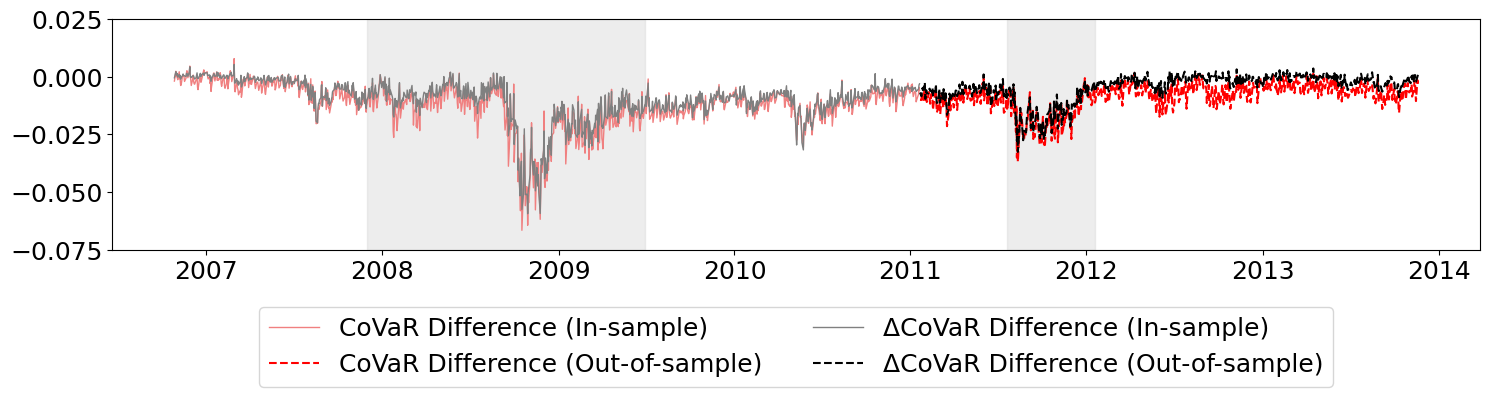}
    \caption*{\textit{Notes:} Differences are computed as the text-based Transformer estimates minus the return-only MLP estimates.}
\end{figure}

During the 2008 crisis, the media narrative was dominated by panic about the global banking system, fears of a deep recession, and widespread uncertainty about whether government interventions would be sufficient to stabilize financial markets. For example, \textit{Reuters} reported news such as ``Germany slashes 2009 growth forecast as crisis bites", ``Recession looms despite global interventions”, and ``Swiss banks raise emergency funds to fight crisis”. These news articles carry signals about systemic fragility, which push the model toward more negative CoVaR and $\Delta$CoVaR estimates compared with the model that relies purely on numerical variables. Other late-2008 dates (e.g., November, December) show similar gaps. In relative terms, the two models differ by 40--50\% during this crisis.

Another significant drop shows up in mid-2010. The news then focused on the euro-zone debt crisis and systemic risk in Europe. For example: ``EU urges Greece to stick to austerity, pension plan”, ``Euro-zone troubles may roil stocks”, and reports on flash-trading risks such as “Naked swaps, flash trading and other key terms.” This information again increase perceived risk, and the text-based model reflects that in more negative CoVaR estimates.

For the out-of-sample period, the largest absolute differences occur during the volatility episode from July 2011 to January 2012. The absolute gaps reach about 0.033--0.037 for CoVaR and $\Delta$CoVaR, with relative deviations of around 33\% to 40\% for $\Delta$CoVaR. The dominant news themes during this time concern the U.S. debt-ceiling crisis, such as ``United States loses prized AAA credit rating from S\&P” and ``Obama officials attack S\&P’s credibility after downgrade”, as well as the euro-zone sovereign crisis, such as ``Euro zone investors gloomiest for nearly 2 years” and ``ICE bars MF Global from floor, customers can unwind only”\footnote{Based on \textit{Reuters} news reports, the collapse of MF Global in 2011 was caused by risky 6 billion euro bets on European sovereign bonds, heavy leverage, and the alleged misuse of client funds to cover losses, culminating in bankruptcy.}.

Across both crisis periods (and also the one in 2010 included in the validation set), the sign of the differences is consistently negative, showing that the model with text produces more negative CoVaR and $\Delta$CoVaR estimates than the model without text. In other words, during crisis episodes, incorporating financial news systematically amplifies the perceived downside tail risk, particularly during crisis episodes. We show an extra news content analysis in Appendix \ref{sec:news_content} to better understand why the inclusion of news embeddings leads to more negative CoVaR predictions during periods of financial distress.

\subsection{Average Quantile Loss} \label{sec:backtesting}
We show that our CoVaR measure can be accurately estimated using a residual-based analysis. In particular, we evaluate the average quantile (AVQ) loss on the test data for CoVaR models with and without textual information. Starting from January 21, 2011 (the beginning of the test period), we compute quarterly average losses and extend the evaluation cumulatively every three months.

Table \ref{tab:three_month_loss_GS} reports the loss comparison results for Goldman Sachs. Over the full test period, the text-based CoVaR model exhibits slightly lower losses than the CoVaR model without text. However, we observe a noticeable increase in losses between the 9th and 15th months for both models. These periods cover exactly a major market downturn: on August 5, 2011, Standard \& Poor’s downgraded U.S. sovereign debt from AAA to AA+ for the first time in history; and later, until December 2011, the U.S. debt limit was raised three times after Congress failed to block the increases or pass a Balanced Budget Amendment\footnote{See details in Congressional Research Service \href{https://www.congress.gov/crs_external_products/RL/PDF/RL31967/RL31967.139.pdf}{\textit{The Debt Limit: History and Recent Increases}}, and U.S. Government Accountability Office \href{https://www.gao.gov/products/gao-12-701}{\textit{Debt Limit: Analysis of 2011–2012 Actions}}}.
\begin{table}[t]
    \centering
    \caption{Cumulative three-month average losses for Goldman Sachs}
    \label{tab:three_month_loss_GS}
    \begin{tabular}{l S[table-format=1.2] S[table-format=1.2]}
    \toprule
     & \textbf{Text-based Transformer} & \textbf{Return-only MLP} \\
    \midrule
    3 months   & \textbf{0.10} & 0.13 \\
    6 months   & \textbf{0.11} & 0.14 \\
    9 months   & \textbf{0.12} & 0.15 \\
    12 months  & \textbf{0.12} & 0.15 \\
    \addlinespace
    15 months  & \textbf{0.12} & 0.16 \\
    18 months  & \textbf{0.12} & 0.15 \\
    21 months  & \textbf{0.12} & 0.15 \\
    24 months  & \textbf{0.12} & 0.15 \\
    \addlinespace
    27 months  & \textbf{0.12} & 0.15 \\
    30 months  & \textbf{0.12} & 0.15 \\
    33 months  & \textbf{0.11} & 0.15 \\
    Full Test Period & \textbf{0.11} & 0.15 \\
    \bottomrule
    \end{tabular}
    \caption*{\textit{Notes:} The table reports cumulative three-month average losses of the test data for the text-based Transformer and return-only MLP CoVaR models over twelve quarterly evaluation periods for Goldman Sachs. Entries are reported in units of $10^{-2}$, (i.e., values are multiplied by 100).}
\end{table}
\textit{Reuters} news articles during this period documented the intense political conflict between Congress and the Obama administration over raising the debt ceiling, raising fears of a potential default. The textual data from this time contains information relevant to systemic stress: terms such as “default risk,” “debt ceiling,” and “downgrade” become frequent in financial news. The Transformer-based model takes this textual information as input and achieves a smaller increase in loss during this crisis period. 

We report the AVQ losses for all the banks in Table~\ref{tab:loss}. Overall, the two models show broadly similar total losses on test data. The inclusion of the financial news does not dramatically change overall predictive accuracy across the full test sample. The performance differences between the two models are generally small across the eight assets. However, we observe two things: 1) the differences between the text-based CoVaR and the CoVaR without text consistently display a negative dip for all banks during the crisis, regardless of whether the total test losses are higher or lower for either model (See Figure \ref{fig:all_banks}); and 2) most losses in the first 12 months, which include the major market downturn, tend to be lower or at least comparable for the Transformer-based model with text.
\begin{table}[!t]
    \centering
    \caption{Average quantile losses of CoVaR models for eight major banks}
    \label{tab:loss}
    \begin{tabular}{l  S[table-format=1.2] S[table-format=1.2]  S[table-format=1.2] S[table-format=1.2]}
        \toprule
        \textbf{Ticker} & \multicolumn{2}{c}{\textbf{Test Loss in First 12 Months}} & \multicolumn{2}{c}{\textbf{Total Test Loss}} \\
        \cmidrule(lr){2-3} \cmidrule(lr){4-5}
        & \textbf{Transformer} & \textbf{MLP} & \textbf{Transformer} & \textbf{MLP} \\
        \midrule
        GS  & \textbf{0.12} & 0.15 & \textbf{0.11} & 0.15 \\
        \addlinespace
        WFC & \textbf{0.16} & 0.16 & 0.14 & \textbf{0.12} \\
        \addlinespace
        JPM & \textbf{0.11} & 0.14 & \textbf{0.12} & 0.14 \\
        \addlinespace
        BAC & \textbf{0.19} & 0.19 & \textbf{0.16} & 0.16 \\
        \addlinespace
        C   & \textbf{0.19} & 0.20 & \textbf{0.18} & 0.18 \\
        \addlinespace
        BK  & \textbf{0.14} & 0.14 & 0.12 & \textbf{0.12} \\
        \addlinespace
        MS  & 0.21 & \textbf{0.18} & 0.18 & \textbf{0.14} \\
        \addlinespace
        STT & \textbf{0.15} & 0.15 & 0.15 & \textbf{0.15} \\
        \bottomrule
    \end{tabular}
    \caption*{\textit{Notes:} The table reports average quantile losses of CoVaR models with and without textual information. Test losses are shown for the first 12 months (crisis period) and for the total period. Bold numbers indicate the lower loss between the Transformer and MLP models for each bank. Entries are reported in units of $10^{-2}$, (i.e., values are multiplied by 100).}
\end{table}

The generally similar total losses suggest that news content may be noisier and provide limited incremental information for risk prediction during stable periods, which may offset the benefits of text integration in evaluations. However, when we focus on the first 12 months, which cover the crisis period, the lower losses observed during this time indicate that textual information becomes more valuable in more uncertain market environments.

\subsection{Comparison with Sentiment-Based Measures}

We compare our results with sentiment-based CoVaR estimates. Specifically, we use \texttt{FinBERT}, a financial-domain pre-trained language model, to classify each news article into one of three sentiment categories: positive, negative, or neutral. These sentiment information are then incorporated into two alternative modeling frameworks: MLP and Transformer. We estimate CoVaR using each sentiment-based model and compare the resulting MLP- and Transformer-based sentiment CoVaR estimates with our text-based CoVaR estimates.

For the first method, to let the MLP take sentiment as an input variable, we construct the daily sentiment index following a commonly used method in the literature \citep{2007_Tetlock, 2011_Loughran_and_McDonald, 2014_Nassirtoussi_et_al}:
\begin{equation}
    \text{Sentiment Index}_{t}=\frac{\text{\# Positive}_{t} - \text{\#Negative}_{t}}{\text{\# Positive}_{t} + \text{\# Neutral}_{t} + \text{\# Negative}_{t}},
\end{equation}
where ``\#" denotes the number of articles in each sentiment category on day $t$. The index reflects the net sentiment (positive minus negative) scaled by the total number of articles, representing the overall sentiment of financial news on that day (see the top panel of Figure~\ref{fig:mlp_sentiment_GS}). The observed sentiment index fluctuates around zero and lies within the range $[-0.5, +0.5]$. Negative sentiment occurs during crisis periods, but the magnitude is small. We augment the return-only MLP model by adding the sentiment index as an additional input variable.
\begin{figure}[t]
    \centering
    \caption{Differences in CoVaR and $\Delta$CoVaR for Goldman Sachs across models}
    \label{fig:mlp_sentiment_GS}
    \includegraphics[width=0.95\textwidth]{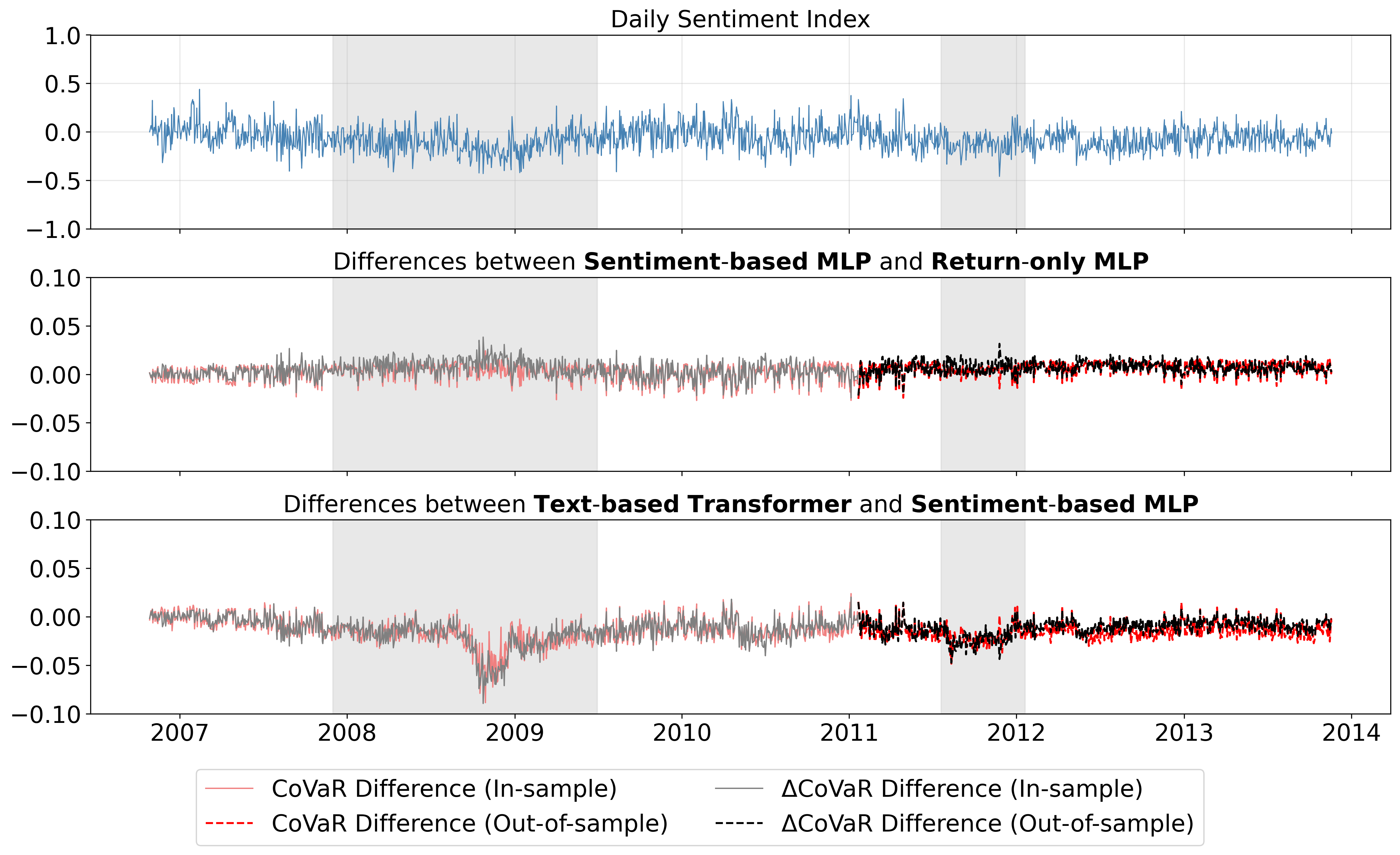}
    \caption*{\textit{Notes:} The \textbf{top panel} shows the daily sentiment index over time. 
    The \textbf{middle panel} shows differences between the sentiment-based MLP estimates and return-only MLP estimates. 
    The \textbf{bottom panel} shows differences between the text-based Transformer estimates and the sentiment-based MLP estimates. 
    Red lines indicate differences in CoVaR, and black lines indicate differences in $\Delta$CoVaR.}
\end{figure}
As shown in the middle panel of Figure~\ref{fig:mlp_sentiment_GS}, the sentiment-based CoVaR adds limited value to systemic risk estimation compared to the return-only MLP CoVaR model. We provide some potential explanations. First, sentiment is only one aspect of human language and cannot fully capture the complexity of language. Relying only on sentiment leads to the loss of important information relevant to financial markets. Second, even if we assume sentiment is the only key information in textual data, investor sentiment still varies significantly across contexts and is not uniformly informative \citep{2024_Wang_and_Ma}. Third, common aggregation methods such as averaging or polarity counts are typically static and equal-weighted. They treat all news sentiments as equally important through the time. This overlooks time-varying importance in news information. These aggregation processes lead to the homogenization of sentiment index \citep{2024_Wang_and_Ma}. This homogenization can be both cross-sectional and temporal: (1) The sentiment scores tend to converge toward the mean and thereby smoothing out the distinctive crucial information in individual news items. The extreme sentiments are diluted, causing predictive information loss; (2) The similarity in average scores across different days homogenizes daily news sentiments, making it difficult for models to capture meaningful time-series variation in investor sentiment. The homogenization mechanism helps explain the small magnitude of day-to-day sentiment fluctuations shown in the top panel of Figure~\ref{fig:mlp_sentiment_GS}. Moreover, news and prices may have interactions in real world, which the common sentiment-based approaches fail to capture. 

The second method we use computes a sentiment-based Transformer CoVaR. Instead of feeding high-dimensional text embeddings, we map each document simply as one of three numerical sentiment labels (e.g., 1, 2, 3) to represent negative, neutral, and positive. Then we apply the same procedure as the text-based Transformer to construct the input matrix.
\begin{figure}[t]
    \centering
    \caption{Differences in CoVaR and $\Delta$CoVaR for Goldman Sachs across models}
    \label{fig:transformer_sentiment_GS}
    \includegraphics[width=0.95\textwidth, keepaspectratio]{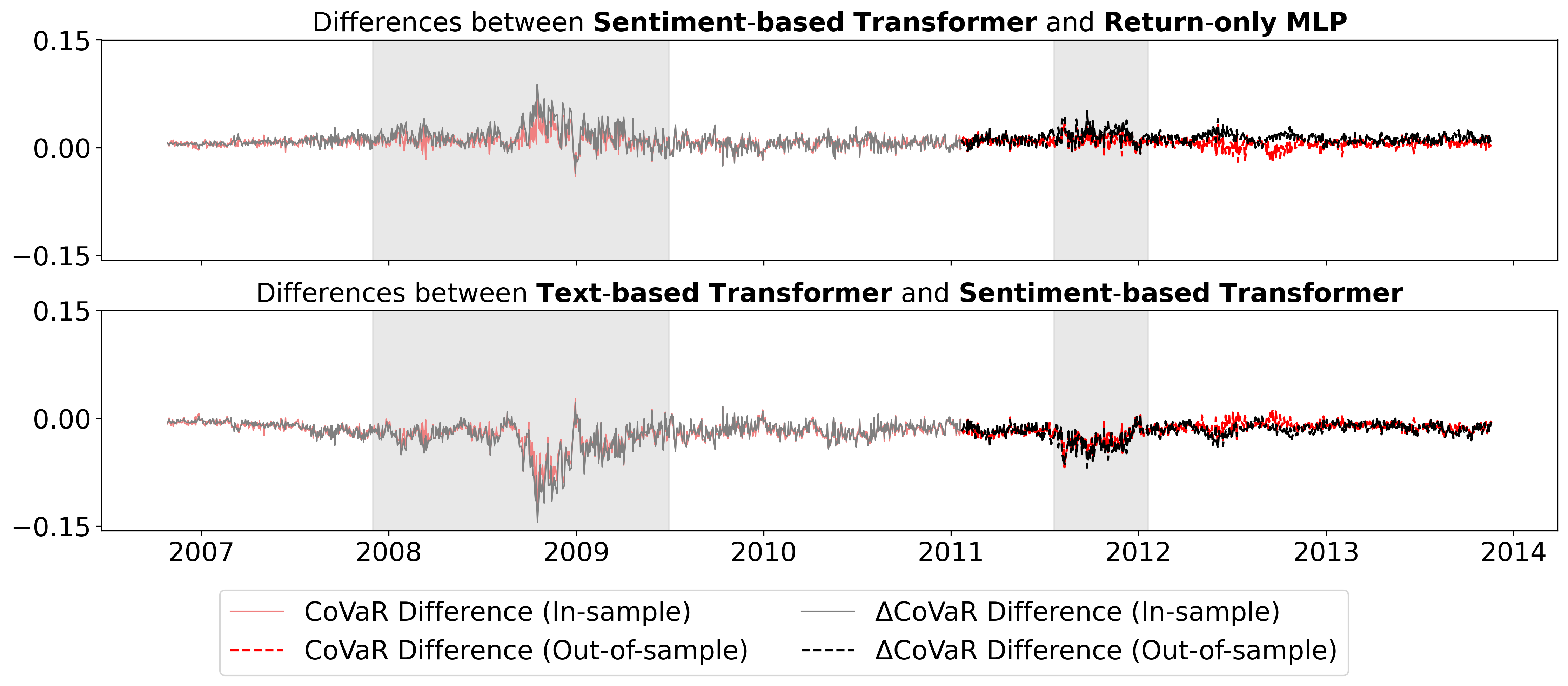}
    \caption*{\textit{Notes:} The \textbf{top panel} shows differences between the sentiment-based Transformer estimates and the return-only MLP estimates. The \textbf{bottom panel} shows differences between the text-based Transformer estimates and the sentiment-based Transformer estimates.}
\end{figure}
We obtain the same conclusion as using the first sentiment-based MLP model: the sentiment-based CoVaR adds limited value to systemic risk estimation compared to the return-only MLP CoVaR model. The largest differences between the two models occur during the 2008 financial crisis. We zoom in and check the sentiment classifications of the news. We zoom in and check the sentiment classifications of the news. Indeed, some news are classified as negative when they are extremely negative, such as ``AIG struggles to survive financial tsunami”. However, there are also news classified as positive, such as ``AIG gets New York's help in accessing \$20 billion”. But is this really a positive news? One can argue that AIG receives governmental support, but on the other hand, when an institution really needs government intervention, it also means the situation is already extremely bad. In addition, ``Chrysler sees U.S. auto market gains toward end 2009” receives a positive label despite the first half of the article describing severe declines in U.S. auto sales, such as \textit{“… auto sales were down more than 11 percent through the first eight months …”}. Classifying articles into only three categories will ignore the information in the first part. Moreover, more than 50\% news during crisis are simply classified as neutral. Thus, collapsing articles into three sentiment categories discards the nuances contained in mixed-tone narratives.

The limitations of the sentiment-based CoVaRs become obvious, when comparing them with the text-based Transformer CoVaR estimates (i.e. the one that uses text embedding). The text embeddings with richer representations that capture deeper linguistic and contextual information, allowing them to detect more relevant information in news, especially during periods of market stress. As shown in the bottom panels of Figures~\ref{fig:mlp_sentiment_GS} and \ref{fig:transformer_sentiment_GS}, the differences between the sentiment-based models and the text-based Transformer are most obvious during crisis periods, where negative dips in CoVaR estimates reappear. This is similar to the differences observed between text-based and return-only CoVaR estimates. This shows that sentiment-based CoVaR modeling overlook important information present in the text, which on the contrary can be captured by using text embeddings.

\subsection{Look-Ahead Bias}

Applying the \textit{GoogleAI} embedding model to transform news articles into embedding vectors may introduce a look-ahead bias. The \texttt{gemini-embedding-001} model may have been trained on data that include future information beyond the analysis period. For instance, if the training data covers the 2008 financial crisis and subsequent bank failures, the model may produce highly similar embeddings for news about different banks. This can be problematic if it occurs even for articles published before the crisis, when news mentioned banks were operating normally. Such bias may skew the embeddings, which overlook subtle but important distinctions in news across different financial institutions.

Mitigating the look-ahead bias is challenging. Ideally, we would use a language model trained only on data preceding our evaluation period. However, since the LLM boom began after 2020, it is difficult to find a model that meets the data requirements of our application and contains sufficient coverage of crisis periods for analysis. Developing a new LLM would be computationally and resource intensive. In addition, \textit{GoogleAI} does not disclose the knowledge cutoff date of \texttt{gemini-embedding-001}. One dataset used in its training is \textit{FRet}, which includes diverse topics such as blog posts, news articles, Wikipedia-like content, and forum discussions \citep{2024_Lee_et_al}. Therefore, there is a high potential for look-ahead bias.

We address this bias by masking the names of the eight banks and other identifying information to prevent the model from exploiting bank-specific cues. Similar to \cite{2025_Breitung_and_Mueller} and \cite{2023_Glasserman_and_Lin}, we apply a Named Entity Recognition (NER) model from the Python package \textit{spaCy} to detect the identifying information of banks. \cite{2025_Breitung_and_Mueller} empirically demonstrate that using this advanced model ensures high masking accuracy.
\begin{figure}[t]
    \centering
    \caption{Differences in CoVaR and $\Delta$CoVaR for Goldman Sachs: with vs. without masking}
    \label{fig:look_ahead_bias}
    \includegraphics[width=0.95\textwidth]{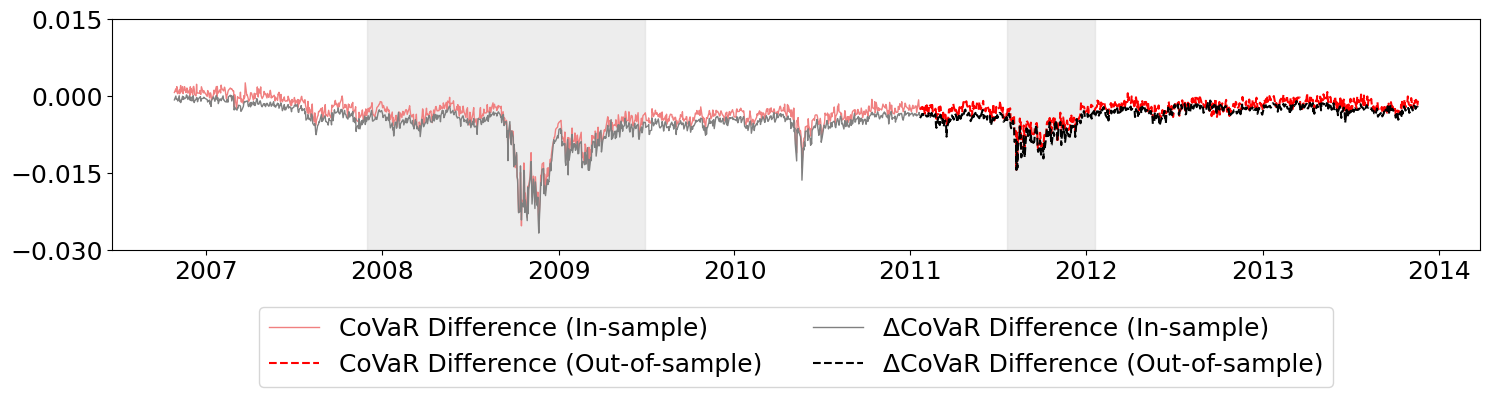}
    \caption*{\footnotesize \textit{Notes:} Differences are computed by subtracting the CoVaR with masking from the CoVaR without masking. Red lines show differences in CoVaR, and black lines show differences in $\Delta$CoVaR. Negative values indicate that the model without masking produces more negative predictions. Shaded regions correspond to identified crisis periods, and dashed vertical lines mark key financial stress events.}
\end{figure}

We find that masking or not masking identifying information does not change the overall conclusion: the CoVaR still becomes more negative during crisis periods when financial news is incorporated. However, there are small numerical differences remain between the models with and without masking (see Figure~\ref{fig:look_ahead_bias}). Negative values in the figure indicate that the model without masking produces more negative CoVaR estimates. The largest absolute differences occur between October and November 2008, when the CoVaR and $\Delta$CoVaR gaps reach about 0.022--0.027. In relative terms, the Transformer-based model without using masking deviates from the Transformer-based model using masking by roughly 12--14\% for both CoVaR and $\Delta$CoVaR. Over the out-of-sample period, the largest gap appears from August to October 2011, with absolute differences of about 0.015 for CoVaR and $\Delta$CoVaR (roughly 9--12\% in relative terms). A likely explanation is that masking reduces how often multiple banks appear in the same article, which weakens the model’s ability to infer relationships between institutions. This may affect the systemic risk estimates.

\section{Conclusion} \label{sec:conclusion_future_work}

We develop a Transformer-based framework for systemic risk modeling that integrates both structured financial data and unstructured textual data. Our approach extends the CoVaR methodology to incorporate high-dimensional, multimodal inputs. We provide empirical evidence of improved tail-risk estimation during periods of financial stress. Specifically, we show that including textual data from financial news results in more negative CoVaR and $\Delta$CoVaR values during crises under a better predictive performance, highlighting the importance of market-related information contained in textual data. In addition, forward-CoVaR, as another systematic risk measure used for early warning, can be interesting for future empirical work.

We establish an upper bound on the $\ell_2$ risk of the Transformer-based quantile regression. In particular, we also provide the convergence rate of our procedure for an explicit case, demonstrating that it can efficiently learn complex nonlinear dependencies in high-dimensional settings.  This complements existing literature on deep learning for financial risk modeling and provides a rigorous econometric foundation for integrating Transformers into systemic risk analysis.


\clearpage
\bibliographystyle{chicago}
\bibliography{ref}



\begin{appendix}
    \section{Transformer Architect}\label{sec:transformer_details}
\begin{figure}[h]
    \centering
    \caption{The Transformer model with residual connection and layer normalization}
    \includegraphics[width=0.45\textwidth, keepaspectratio]{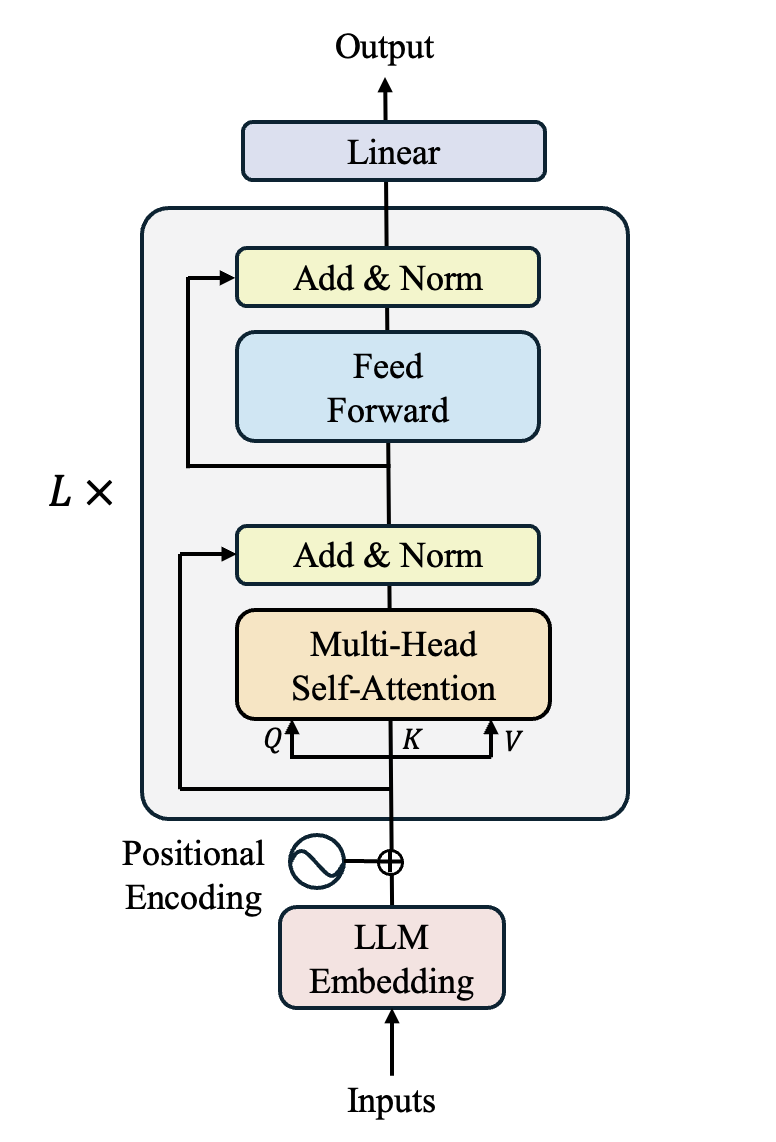}
    \caption*{\textit{Notes:} The Transformer encoder layer, including the embedding layer, positional encoding, multi-head self-attention (MSA), feed-forward network (FFN), residual connections, layer normalization, and a linear output layer.}
    \label{fig:transformer_encoder}
\end{figure}
A Transformer encoder consists of $L$ layers, each comprising a multi-head self-attention sublayer and a feed-forward network sublayer. It is also possible to add residual connection and layer normalization after each sublayer, which will make the training more stable. 

Formally, we define the $l$-th multi-head self-attention class with residual connection $\widetilde{\mathcal{G}}_{l}^{(MSA)}(n,d,H)$ as the set of map $\widetilde{g}_l^{(msa)}: \mathbb{R}^{d \times n} \to \mathbb{R}^{d \times n}$ of the form
\begin{equation}
    \widetilde{g}_l^{(msa)}(Z) = \sum_{h=1}^H W_{l,h}^{(O)} \, W_{l,h}^{(V)} Z \cdot \sigma_S \left(\frac{Z^{\top} {W_{l,h}^{(K)}}^{\top} W_{l,h}^{(Q)} Z}{\sqrt{d}} \right).
\end{equation}
For each head $h \in \{1, \ldots, H\}$, we denote by $W_{l, h}^{(V)}, W_{l, h}^{(K)}, W_{l, h}^{(Q)} \in \mathbb{R}^{\frac{d}{H} \times d}$, and ${W_{l,h}^{(O)}} \in \mathbb{R}^{d \times \frac{d}{H}}$ the value, key, query, and projection matrices, respectively.

The $l$-th feed-forward network class with residual connection $\widetilde{\mathcal{G}}_{l}^{(FF)}(d, d_h)$ is defined as the set of maps $\widetilde{g}_l^{(ff)}:\mathbb{R}^{d \times n} \to \mathbb{R}^{d \times n}$ of the form
\begin{equation}
    \widetilde{g}_l^{(ff)}(Z) = W_l^{(2)} \, \sigma_{R} \left(W_l^{(1)} \cdot Z + b_l^{(1)} \right) + b_l^{(2)}.
\end{equation}
where $d_h$ denotes the hidden dimension of the feed-forward network. The parameters are $W_l^{(1)} \in \mathbb{R}^{d_h \times d}$, $b_l^{(1)} \in \mathbb{R}^{d_h}$ and $W_l^{(2)} \in \mathbb{R}^{d \times d_h}$, $b_l^{(2)} \in \mathbb{R}^{d}$. 

In addition, we define the $l$-th layer normalization function class $\mathcal{G}_{l}^{(LN)}$ as the set of map $g_l^{(ln)}: \mathbb{R}^{d \times n} \to \mathbb{R}^{d \times n}$ of the form
\begin{equation}
    g_l^{(ln)}(Z) = \left( \frac{Z_{j} - \frac{1}{d} \sum_{i} Z_{i,j}}{\sqrt{\frac{1}{d} \sum_{k}(Z_{k,j} - \frac{1}{d} \sum_{i} Z_{i,j})}} \right)_{j=1}.
\end{equation}

Building on the definition of the Transformer with residual connections and layer normalization, we consider a class of Transformer block of the following form:
\begin{equation}
    \begin{split}
        \widetilde{\mathcal{G}}^{(TF)}(n, d, H, d_h, L) := \Big\{ g_{Lb}^{(ln)} \circ \widetilde{g}_L^{(ff)} \circ g_{La}^{(ln)} &\circ \widetilde{g}_L^{(msa)} \circ \cdots \\
        \cdots\circ g_{1b}^{(ln)} &\circ \widetilde{g}_1^{(ff)} \circ g_{1a}^{(ln)} \circ \widetilde{g}_1^{(msa)} \Big\}.
    \end{split}
\end{equation}
Appending an MLP, defined in the same manner as in Section~\ref{sec:estimator}, yields the Transformer-based neural network class with residual connections and layer normalization.
\section{News Content} \label{sec:news_content}
\begin{figure}[!h]
    \centering
    \caption{News timeline for selected clusters.}
    \includegraphics[width=0.9\textwidth, keepaspectratio]{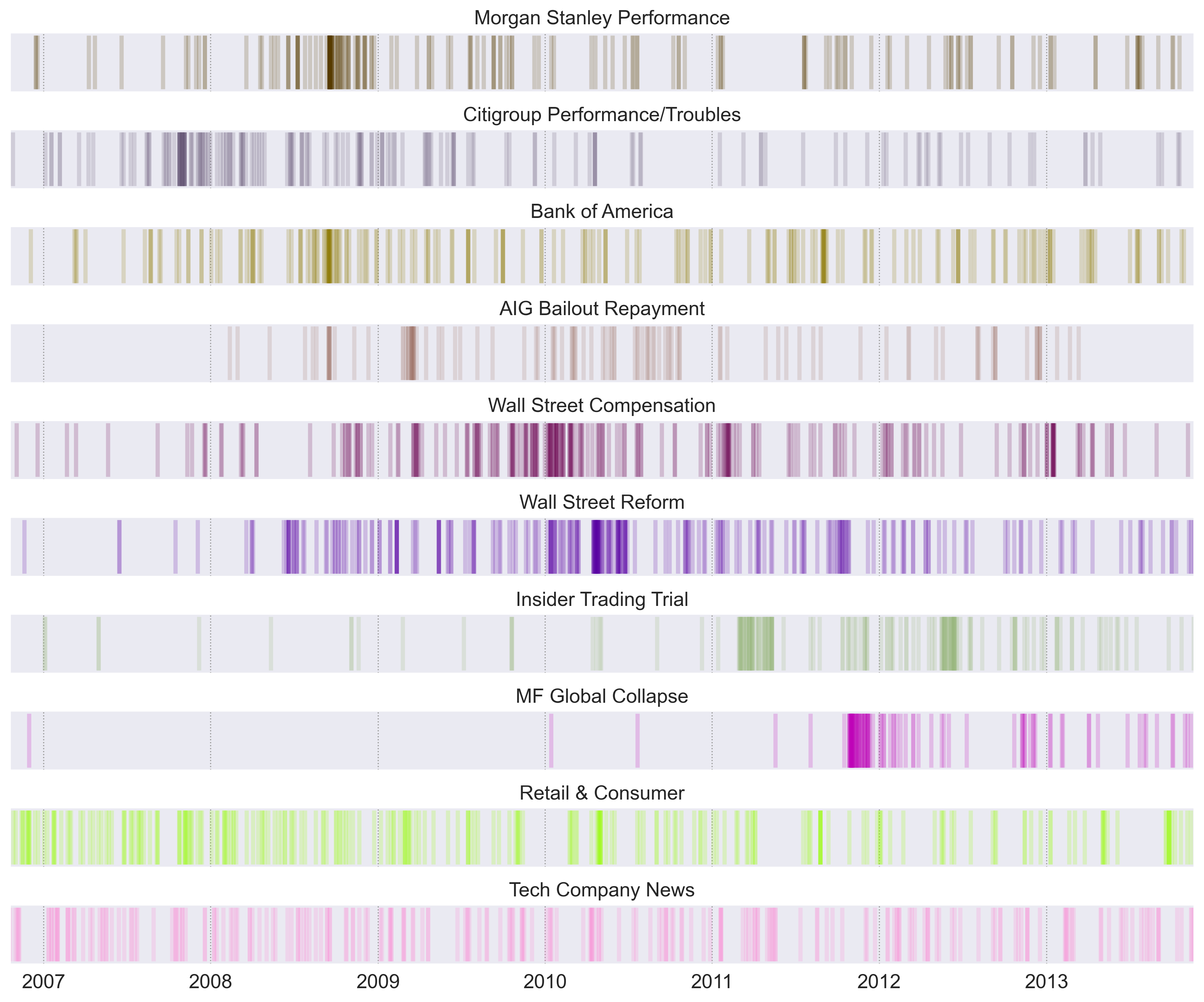}
    \caption*{\textit{Notes:} Each row represents one cluster, with vertical lines indicating the publication dates of news articles belonging to that cluster. Colors correspond to cluster identities, and titles describe the main theme of each cluster. The x-axis shows the timeline, allowing comparison of news activity patterns across topics.}
    \label{fig:timeline}
\end{figure}
To better understand why the inclusion of news embeddings leads to more negative CoVaR predictions during periods of financial distress, we analyze the underlying content of the news articles used in the model. Our aim is to identify the themes of news and analyze their relationship with the observed changes in systemic risk estimates. While this content analysis can provide insights into narrative-based risk dynamics, we acknowledge the limitation that causality between news and market movements cannot be definitively established. News and financial markets may interact as part of an endogenous system, mutually influencing one another.

Figure \ref{fig:timeline} illustrates the temporal distribution of major news clusters. Notably, themes such as Morgan Stanley Performance, Citigroup Performance, and Bank of America Troubles dominate the period between 2007 and 2009, with some appearing even earlier. These topics reflect widespread distress among major financial institutions during the global financial crisis. Articles within these clusters frequently use terms such as “losses,” “fear,” and “woes,” highlighting the pessimistic view surrounding the banking sector at the time. The concentration of such narratives likely contributed to the more negative CoVaR predictions during this period.

Another negative dip CoVaR occurs around 2010, coinciding with a surge in news coverage related to Wall Street Compensation and Financial Reform Efforts. In the wake of the crisis, public criticism of executive compensation practices in the financial sector intensified. The Obama administration introduced major financial reform legislation, notably the Dodd-Frank Act, to address perceived regulatory failures. This policy response and public backlash were extensively covered in the media, potentially reinforcing a more negative systemic risk estimates. These clusters underscore the socio-political consequences of the crisis and how they were captured in the news.

In early 2011, coverage of insider trading allegations involving Goldman Sachs appeared to coincide with another reversal in CoVaR trends. Similarly, the MF Global Collapse cluster spanning 2011 to 2012 demonstrates that news about corporate failures continued to influence perceptions of systemic risk well beyond the immediate aftermath of the 2008 crisis. MF Global, a major derivatives broker, filed for bankruptcy following risky bets on European debt, and its collapse was widely reported as a significant market event.

However, not all news clusters exhibit a strong connection to systemic risk. For instance, clusters related to Retail or Technology companies tend to be more evenly distributed throughout the sample period and display relatively weak correlations with the CoVaR estimates of financial institutions. These findings suggest that only certain types of news, especially those focused on financial sector instability, regulation, and corporate misconduct, meaningfully contribute to systemic risk as captured by our embedding-based CoVaR model.

    \section{Monte Carlo Simulation}
\label{sec:simulation}

We validate our empirical Transformer-based model using a Monte Carlo simulation experiment. The goal is to examine the model's ability to learn the dependence between two financial institutions from noisy textual and numerical data in a finite-sample setting. We generate two synthetic time series representing the stock returns of two hypothetical institutions. The return of the second institution is conditional on the return of the first institution.

\subsection{Simulation Setting}

The baseline setting is linear. The first time series follows a first-order autoregressive process (AR(1)), and the second series is a linear function conditional on the first series with a Gaussian noise:
\begin{subequations}
    \begin{align}
        y_{t}^{(1)} &= \phi \, y_{t-1}^{(1)} + \epsilon_t, \quad \epsilon_t \sim \mathcal{N}(0, \sigma_1^2), \\
        y_{t}^{(2)} &= \beta \, y_{t}^{(1)} + \eta_t, \quad \eta_t \sim \mathcal{N}(0, \sigma_2^2),
    \end{align}
\end{subequations}
with parameters set as $\phi = 0.8$, $\sigma_1 = 0.15$, $\beta = 1.2$, $\sigma_2 = 0.2$, and the initial value $y_{t=0}^{(1)} = 0$. 
Under this setting, we focus on the CoVaR estimates of the second institution, which is conditional on the first institution being at its VaR level. Notably, we feed the Transformer-based model the real-world \textit{Reuters} financial news as noisy textual input alongside the generated numerical return data. Since these news have no real relation to our synthetic returns, we expect that the Transformer-based model will effectively learn the true numerical dependence, despite the noisy textual input. 

To compute the theoretical CoVaR, we first calculate the VaR for the first institution based on the conditional mean and volatility:
\begin{equation}
    \mathrm{VaR}_{t}^{(1)}(\tau) = \phi \, y_{t-1}^{(1)} + \sigma_1 \, z(\tau).
\end{equation}
where $z(\tau) = \Phi^{-1}(\tau)$ and $\Phi^{-1}(\cdot)$ denotes the inverse cumulative distribution function of the standard normal distribution. We set $\mathrm{VaR}_{t=0}^{(1)}(\tau) = \sigma_1 z(\tau)$ for $t = 0$. Then, we derive the theoretical CoVaR for the second institution by conditioning on the VaR of the first series. Specifically,
\begin{equation}
    \mathrm{CoVaR}_t^{(2\mid1)}(\tau) = \beta \, \mathrm{VaR}_t^{(1)}(\tau) + \sigma_2 \, z(\tau).
\end{equation}
For comparison purposes, one can also compute the VaR for the second institution:
\begin{equation}
    \mathrm{VaR}_t^{(2)}(\tau) = \beta \cdot \phi \, y_{t-1}^{(1)} + \sqrt{\beta^2 \, \sigma_1^2 + \sigma_2^2} \cdot z(\tau),
\end{equation}
with the initial value set as $\mathrm{VaR}_{t=0}^{(2)} = \sqrt{\beta^2 \, \sigma_1^2 + \sigma_2^2} \cdot z(\tau)$. The VaR of the second institution incorporates randomness from both $\epsilon_t$ and $\eta_t$.



\subsection{Model Predictive Performance} 

We visualize the simulation results for the Transformer-based model in Figure~\ref{fig:simulation}. The estimated CoVaR from this model closely aligns with the theoretical CoVaR, This indicates the model can capture the underlying conditional dependence between the two institutions, even in the presence of extremely noisy textual data. We compute the prediction error, which is the differences between the estimated and true CoVaR. The prediction error for the 5\% CoVaR is close to zero.
\begin{figure}[!t]
    \centering
    \caption{Simulation result}
    \includegraphics[width=0.95\textwidth, keepaspectratio]{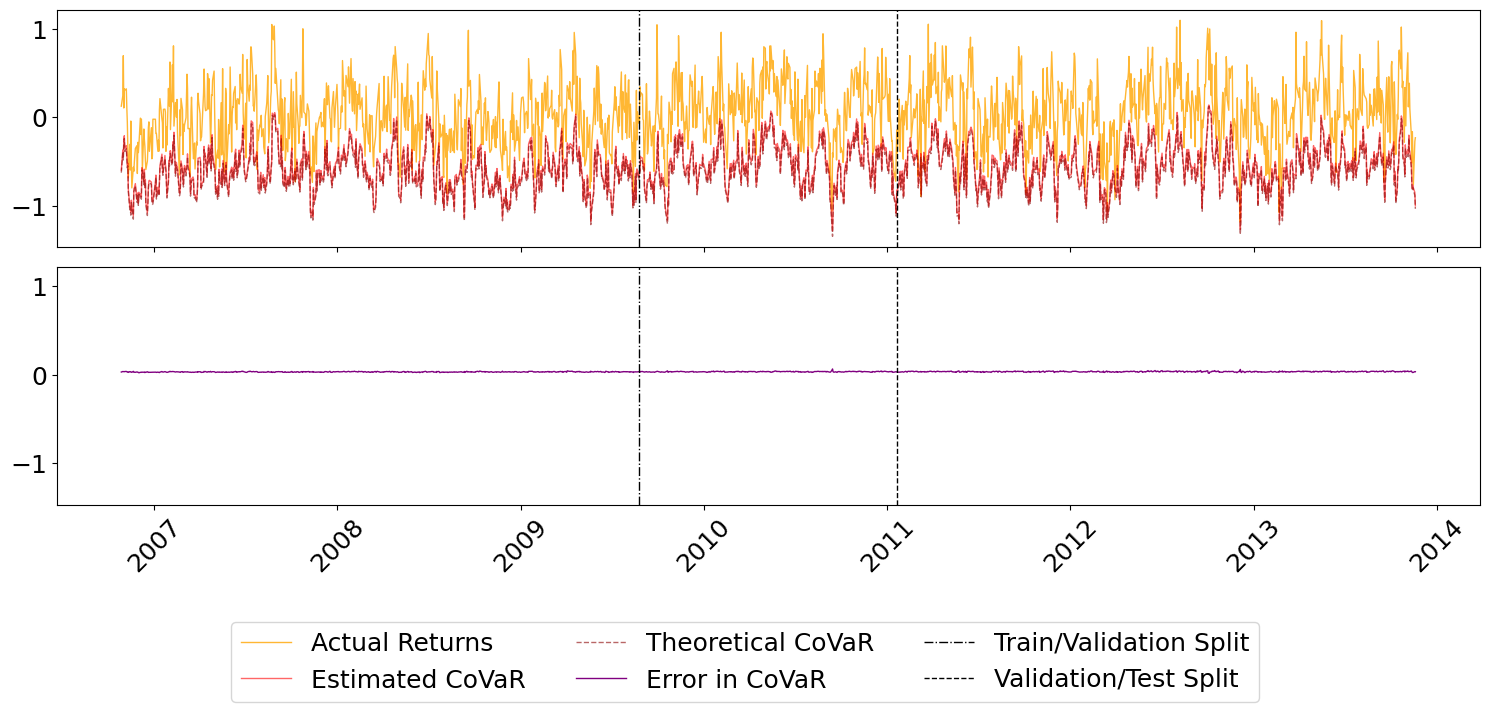}
    \caption*{\textit{Notes:} \textbf{Top panel:} Estimated CoVaR from the Transformer-based model. \textbf{Bottom panel:} Prediction error of CoVaR. Left and right vertical lines are the train-validation and validation-test splits, respectively.}
    \label{fig:simulation}
\end{figure}

We benchmark the Transformer-based model against a plain-vanilla MLP consisting of one hidden layer with 64 ReLU‐activated units (the model by \cite{2022_Keilbar_and_Wang}). We evaluate the predictive performance of both models using the mean absolute error (MAE) between the estimated CoVaR and its theoretical value at quantile $\tau = 5\%$:
\begin{equation}
    \mathrm{MAE}(\tau) = \frac{1}{T+1} \sum_{t=0}^{T} \bigl| \widehat{\mathrm{CoVaR}}_{t}^{(2\mid1)}(\tau) \; - \; {\mathrm{CoVaR}_{t}^{(2\mid1)}}(\tau) \bigr|.
\end{equation}

Table \ref{tab:MAE} presents the MAE results for different models. For the Transformer-based model, the MAE at the 5\% quantiles is close to zero. The MLP models, which rely only on time series inputs, are not affected by the additional noise introduced by the text data. In contrast, the Transformer-based model incorporates noisy textual information, yet it remains capable of capturing the underlying dependency between the two time series.
\begin{table}[ht]
    \centering
    \begin{tabular}{lc}
        \hline
        & \textbf{5\% Quantile} \\
        \hline
        Transformer-based &  0.023 \\
        MLP (ReLU)  & 0.039 \\
        \hline
    \end{tabular}
    \caption{Mean Absolute Error (MAE) at 5\% quantile level for the MLP and Transformer.}
    \label{tab:MAE}
\end{table}
These simulation results support the empirical performance of our Transformer-based model: it does not merely fit noise but can learn meaningful patterns.
    \section{Proof of Theorem \ref{thm:covar_convergence}} 
\label{sec:covar_consistency}
{ 
\begin{assumption}\label{assum:density_VaR}  
For each  $f \in \mathcal{F}$ with the given class $\mathcal{F}$ used for fitting, we have that 
$$
\mathbb{E} [f(\operatorname{VaR}_{-j,t}^{\tau},  E_{j,t}) - f_{j,\tau}^{*}(\operatorname{VaR}_{-j,t}^{\tau}, E_{j,t})])^2 \lesssim \mathbb{E}[f({R}_{-j,t},  E_{j,t}) - f_{j,\tau}^{*}({R}_{-j,t},  E_{j,t})]^2.
$$
\end{assumption}
Assumption~\ref{assum:density_VaR} essentially relates the mean squared error evaluated under VaR values to the $\ell_2$ norm bound derived in Lemma \ref{lemma:last}, and it holds trivially if the joint density of $(\operatorname{VaR}_{-j,t}^{\tau}, E_{j,t})$ upon its support is dominated by a constant times the joint density of $(R_{-j,t}, E_{j,t})$.

\begin{lemma}\label{lem:theo1}
    Suppose that the conclusion of Lemma \ref{lemma:last} applies for the training sample $(X_t, Y_t), t=1,\cdots, T$ with $X_t=(R_{-j,t}, E_{j,t}) , Y_t =R_{j,t}$, $ (R_{-j,t'}, E_{j,t'})$ is an independent copy of $X_t$, $ \delta_T>0,  \delta_T \rightarrow 0, \delta_T^2 T \rightarrow \infty$ as $T$ grows and Assumption \ref{assum:density_VaR} holds
    then with probability at least  $1-\exp (-\delta_T ^2 T)$, then
    \begin{equation}
        \begin{split}  
           \mathbb{E} (\widehat{f}_{j,\tau}( (VaR^\tau_{-j,t'}, E_{j,t'})) - f_{j,\tau}^*((VaR^\tau_{-j,t'}, E_{j,t'})))^2 \; \lesssim \; \inf_{f \in \mathcal{F}} \Bigl\|f - f_{j,\tau}^*\Bigl\|_{\infty}^2 +  \delta_T +  \sqrt{\frac{V_{\left(\mathcal{F},\|\cdot\|_L \infty\right)}(\delta_T)}{T}}. \nonumber
        \end{split}
    \end{equation}
\end{lemma}
\begin{proof}
    The following is directly implied by the conclusion of Lemma \ref{lemma:last}, the definition of $\|\cdot \|_{\ell_2}$ and Assumption \ref{assum:density_VaR}:  for an arbitrary $r_0>\delta >0$, and $\gamma > 0$, with probability at least $1-\exp (-\gamma)$, we have the following upper bound holds, $\left(\frac{\widetilde{C}_f^2}{c^2} + 1\right) \inf_{f \in \mathcal{F}} \Bigl\|f - f_\tau^*\Bigl\|_{\infty}^2 + \frac{\widetilde{C}_f}{c}\Biggl( \delta + r_0 \sqrt{\frac{V_{\left(\mathcal{F},\|\cdot\|_L \infty\right)}(\delta)}{T}} + r_0 \sqrt{\frac{\gamma}{T}} + \frac{C_f \gamma}{T} \Biggl)$,
    where $\widetilde{C}_f, c, V_{(\mathcal{F},\|\cdot\|_\infty)}(\delta)$ are as specified in  Lemma \ref{lemma:last}. Then the conclusion follows once we choose $\delta=\delta_T$ and $\gamma =\delta_T^2 T$.
\end{proof}
}

We now show the proof of Theorem \ref{thm:covar_convergence}.
\begin{proof}
    The proof proceeds in two parts: approximation error and estimation error.

    Recall the definition
    \begin{equation}
        \widehat{\operatorname{CoVaR}}_{j,t}^{\tau} = \widehat{f}_{j,\theta, \tau} \Big(\widehat{\operatorname{VaR}}_{-j,t}^{\tau},\, E_{j,t}\Big),
        \qquad
        \operatorname{CoVaR}_{j,t}^{\tau} = f_{j,\tau}^{*}\Big(\operatorname{VaR}_{-j,t}^{\tau},\,E_{j,t}\Big).
    \end{equation}
    By adding and subtracting $\widehat{f}_{j,\theta,\tau}(\operatorname{VaR}_{-j,t}^{\tau}, E_{j,t})$ and $f_{j,\tau}^{*}(\operatorname{VaR}_{-j,t}^{\tau}, E_{j,t})$, we obtain
    \begin{equation} \label{eq:covar_decomp}
        \begin{split}
            &\widehat{f}_{j,\theta,\tau}(\widehat{\operatorname{VaR}}_{-j,t}^{\tau}, E_{j,t}) - f_{j,\tau}^{*}(\operatorname{VaR}_{-j,t}^{\tau}, E_{j,t}) \\ 
            =\:& \Big[\widehat{f}_{j,\theta,\tau}(\widehat{\operatorname{VaR}}_{-j,t}^{\tau}, E_{j,t}) - \widehat{f}_{j,\theta,\tau}(\operatorname{VaR}_{-j,t}^{\tau}, E_{j,t})\Big] \\
            &\; + \Big[\widehat{f}_{j,\theta,\tau}(\operatorname{VaR}_{-j,t}^{\tau}, E_{j,t}) - f_{j,\tau}^{*}(\operatorname{VaR}_{-j,t}^{\tau}, E_{j,t})\Big] \nonumber \\
            &\; + \Big[f_{j,\tau}^{*}(\operatorname{VaR}_{-j,t}^{\tau}, E_{j,t}) - f_{j,\tau}^{*}(\widehat{\operatorname{VaR}}_{-j,t}^{\tau}, E_{j,t})\Big].
        \end{split}
    \end{equation}
    Taking absolute values and applying the triangle inequality yields
    \begin{equation} \label{eq:covar_triangle}
        \max_j \big| \widehat{\operatorname{CoVaR}}_{j,t}^{\tau}-\operatorname{CoVaR}_{j,t}^{\tau}
        \big| \le A_t + B_t + C_t,
    \end{equation}
    where $A_t, B_t, C_t$ denote the absolute values of the three bracketed terms in \eqref{eq:covar_decomp}. Under Assumption \ref{assp:Liptschiz}, we have both $f_{j,\tau}^{*}(\cdot, E_{j,t})$ and $\widehat{f}_{j,\theta,\tau}(\cdot, E_{j,t})$ are Lipschitz in their first argument with constants $L_r$ and $\widehat{L}_r$, then
    \begin{equation}
        \begin{split}
            A_t &= \max_j \big|\widehat{f}_{j,\theta,\tau}(\widehat{\operatorname{VaR}}_{-j,t}^{\tau}, E_{j,t})-\widehat{f}_{j,\theta,\tau}(\operatorname{VaR}_{-j,t}^{\tau}, E_{j,t})\big| \le \widehat L_r\, \big\|\widehat{\operatorname{VaR}}_{-j,t}^{\tau} - \operatorname{VaR}_{-j,t}^{\tau}\big\|_1, \\
            C_t &= \max_j \big|f_{j,\tau}^{*}(\operatorname{VaR}_{-j,t}^{\tau}, E_{j,t}) - f_{j,\tau}^{*}(\widehat{\operatorname{VaR}}_{-j,t}^{\tau}, E_{j,t})\big| \le L_r\, \big\|\widehat{\operatorname{VaR}}_{-j,t}^{\tau} - \operatorname{VaR}_{-j,t}^{\tau}\big\|_1.
        \end{split}
    \end{equation}
    Assumption \ref{assp:var} implies
    \begin{equation}
        A_t + C_t = O_p(\widehat{L}_r b_T).
    \end{equation}
 
  {Lemma \ref{lem:theo1} together with  Markov's inequality implies that, for any $M > 0$, with probability at least $1-\exp(-\delta_T^2 T)$, 
    \begin{equation}
        \mathrm{P}\left( \max_j \Big|\widehat{f}_{j,\theta,\tau}(\operatorname{VaR}_{-j,t}^{\tau}, E_{j,t}) - f_{j,\tau}^{*}(\operatorname{VaR}_{-j,t}^{\tau}, E_{j,t})\Big| > M \; \right) \le \frac{a_T^2}{M^2},
    \end{equation}
    with $a_T^2 = \inf_{f \in \mathcal{F}} \Bigl\|f - f_\tau^*\Bigl\|_{\infty}^2 +  \delta_T +  \sqrt{\frac{V_{\left(\mathcal{F},\|\cdot\|_L \infty\right)}(\delta_T)}{T}}.$
    }
    Hence,
    \begin{equation}
        B_t = \max_j \Big|\widehat{f}_{j,\theta,\tau}(\operatorname{VaR}_{-j,t}^{\tau}, E_{j,t}) - f_{j,\tau}^{*}(\operatorname{VaR}_{-j,t}^{\tau}, E_{j,t})\Big| = O_p(a_T^2),
    \end{equation}
    provided the evaluation is carried out out-of-sample (e.g. via cross-fitting).
    
    Combining terms 1--3 in \eqref{eq:covar_triangle}, we obtain
    \begin{equation}
    \max_j \Big| \widehat{\operatorname{CoVaR}}_{j,t}^{\tau}-\operatorname{CoVaR}_{j,t}^{\tau} \Big| = O_p(a_T^2) + O_p(b_T).
    \end{equation}
    By assumption, the VaR error is $O_p(T^{-1/2})$. And the convergence rate of the quantile function error $\left\|\widehat{f}_{j,N} - f_{j,\tau}^{*} \right\|_{\ell^2}$ with $\widehat{f}_{j,N}$ the estimator restricted to a finite $\delta$-covering of the function class is shown in Lemma \ref{lemma:last} of the supplementary material. We therefore verify the conditions of Lemma \ref{lemma:last}. 
\end{proof}

    \section{Proof of Corollary \ref{cor:convergence_rate}} \label{sec:corollary}

\begin{proof}
The proof proceeds by analyzing the trade-off between the approximation bias and the variance associated with the covering number of the architecture. Let $\lceil\cdot\rceil$ denotes the ceiling function and $\vee$ denotes the maximum operator. Define that $a_n \asymp b_n$ if $a_n \lesssim b_n$ and $b_n \lesssim a_n$ for two sequences. First, we characterize the Transformer complexity terms $\eta$ and $\widetilde{\eta}$ in Lemma \ref{lem:cover-tf-mlp}. From the definitions provided, both terms scale with the cubic root of the logarithmic dimensions of the input:
\begin{equation}
    \eta \asymp (\log(dn))^{1/3} \quad \text{and} \quad \widetilde{\eta} \asymp (\log(dn))^{1/3}.
\end{equation}
Thus, the cubic term appearing in the covering number bound satisfies:
\begin{equation} \label{eq:eta_log_relation}
    (\widetilde{\eta} + \eta)^3 \asymp \log(dn).
\end{equation}

Based on Theorem~\ref{thm:covar_convergence}, the error bound is governed by $\widetilde{R}_{j,\delta,T}$. Thus, by Lemma \ref{lem:cover-tf-mlp}, \ref{lemma:edelman_B2_continuous}, and let $D \asymp m \asymp \log (d_m)$, we obtain
\begin{equation}
    \widetilde{R}_{\delta,T} \lesssim_P d_{m}^{-2\beta/s} + \delta + r_0 \sqrt{\frac{d_{m} \log (d_m) \cdot \Big(\log \big( \delta^{-1} \, d \big) + \log^2 (d_{m}) \Big) + \frac{C_{\eta}\log(dn)}{\delta^2}}{T}},
\end{equation}
where $C_{\eta}$ is a constant independent of $T, n,$ and $d$. Using the inequality $\sqrt{a+b} \leq \sqrt{a} + \sqrt{b}$, we have:
\begin{equation} \label{eq:widetilde_R_decoupled}
    \widetilde{R}_{\delta, T} \lesssim_P d_{m}^{-2\beta/s} + \left( \delta + \frac{\sqrt{\log(dn)}}{\delta \sqrt{T}} \right) + \sqrt{\frac{d_{m} \log (d_m) \cdot \Big(\log \big( \delta^{-1} \, d \big) + \log^2 (d_{m}) \Big)}{T}}.
\end{equation}
We set $\delta \asymp T^{-1/4} (\log(dn))^{1/4}$. Since $D = 8 + (m+5)\left(1 + \left\lceil\log_2(s \vee \beta)\right\rceil\right)$, we have $D\asymp m\asymp \log(d_m)$. And $d_m$ will be chosen as a power of $T$ in the balancing step. Then, the bound simplifies to:
\begin{equation}
    \widetilde{R}_{\delta, T} \lesssim_P d_{m}^{-2\beta/s} + T^{-1/4}(\log(dn))^{1/4} + \sqrt{\frac{d_{m}  \log (d_m) \log^3 (dT)}{T}}.
\end{equation}

Pick $d_{m} \asymp \left( \frac{T}{\log^3 (dT)} \right)^{\frac{s}{s+4\beta}}$, and plug this into the error expression, we arrive at the final convergence rate:
\begin{equation} \label{eq:lemma_convergence_rate}
    \max_j \left[ \left|\operatorname{CoVaR}_{j,t}^{\tau} - \widehat{\operatorname{CoVaR}}_{j,t}^{\tau}\right|\right] = O_p\left( \left( \frac{T}{\log^3 (dT)} \right)^{-\frac{2\beta}{s+4\beta}} \right) + O_p\left( T^{-1/4}(\log(dn))^{1/4} \right) + L_{\mathcal{F}} b_T.
\end{equation}
This result demonstrates that the rate is governed by the intrinsic dimensionality $s$ and smoothness $\beta$.

In addition, the constant $L_{\mathcal{F}}$ characterizes the Lipschitz continuity of the architecture $\mathcal{F}$. Given a one-layer Transformer block with a $D$-depth MLP (as described in Section \ref{sec:bias} and Theorem 5 in \cite{2020_Schmidt-Hieber}), where $D = 8 + (m+5)\left(1 + \lceil\log_2(s \vee \beta)\rceil\right)$, we have
\begin{equation} 
    L_{\mathcal{F}} \le \left( L_\sigma^{D-1} \prod_{i=1}^{D} \|W_i\|_{\infty} \right) \cdot (B_V B_{KQ}).
\end{equation}
From Lemma \ref{lemma:edelman_B2_continuous}, the weights for each head $h \in \{1, \dots, s\}$ satisfy $\|W^{(Q)}_h\|_{2,1} \le \frac{\log(4n/\gamma)}{1-2\Delta}$, while other Transformer weights $\|W^{(K)}_h\|_{2,1}, \|W^{(V)}_h\|_{2,1}, \|W^{(O)}_h\|_{2,1} \le 1$. The MLP weights satisfy $\|W_i\|_\infty \le B$ for $i =1, \ldots, D$ and $B>0$. And $L_\sigma = 1$, we obtain:
\begin{equation} \label{eq:L_r_final}
    L_{\mathcal{F}} \le B^D \left( \frac{s^3 \log(4n/\gamma)}{1-2\Delta} \right) = O(B^D\log(n)).
\end{equation}
Recall that $D \asymp \log d_m$, and $d_m$ is chosen as a function of $T$. Looking at the last term in \eqref{eq:lemma_convergence_rate}, suppose that $b_T=T^{-\alpha}$ with $\alpha\ge \frac{2\beta}{s+4\beta}$. We choose the depth $D$ such that
\begin{equation}
    D \asymp \frac{\log d_m}{\log(2\vee B)}\cdot \frac{\alpha(s+4\beta)-2\beta}{s}.
\end{equation}
Then $B^D \le (2\vee B)^D \asymp d_m^{(\alpha(s+4\beta)-2\beta)/s}$. And we obtain
\begin{equation}
    B^D b_T \lesssim  \left(\frac{T}{\log^3(dT)}\right)^{\alpha-\frac{2\beta}{s+4\beta}} T^{-\alpha} = T^{-\frac{2\beta}{s+4\beta}}\cdot (\log(dT))^{-3\left(\alpha-\frac{2\beta}{s+4\beta}\right)}.
\end{equation}
Therefore, $L_{\mathcal F}b_T \lesssim T^{-\frac{2\beta}{s+4\beta}}\log(n)$ up to logarithmic factors. Then the rate in \eqref{eq:lemma_convergence_rate} becomes
\begin{equation}
    \max_j \left| \operatorname{CoVaR}_{j,t}^{\tau} - \widehat{\operatorname{CoVaR}}_{j,t}^{\tau} \right| = O_p\left( \left( \frac{T}{\log^3(dT)} \right)^{-\frac{2\beta}{s+4\beta}} + T^{-1/4}(\log(dn))^{1/4} + \log (n) \right).
\end{equation}
\end{proof}
    \section{Error Bounds for Quantile Regression Estimators}
\label{sec:error_bounds_quantile_estimator}

\textbf{General Setup.} In the following proof, let $\{(X_t, Y_t)\}_{t=1}^T$ be an i.i.d. random sample dataset of size $T$ with $(X_t, Y_t) \in \mathbb{R}^{d \times n} \times \mathbb{R}$, and let $(X,Y)$ be independent copy of $(X_t, Y_t)$. The proof is done assuming i.i.d. samples, the extention to, e.g., $m$-dependent series, are straightforward, see, e.g., \cite{li1992stability}.
Now we consider the target model of the following form:
\begin{equation} \label{model_in_lemma}
    Y_t = f_\tau^*(X_t) + U_{\tau, t},
\end{equation}
where $F_{U_{\tau, t} \mid X_t}^{-1}(\tau)=0$ and $\tau \in(0,1)$ is a quantile level. The objective is to approximate $f_\tau^*: \mathbb{R}^{d \times n} \rightarrow \mathbb{R}$, given as
\begin{equation}
    f_\tau^*(x_t)=F_{Y \mid X=x_t}^{-1}(\tau),
\end{equation}
where $F_{Y \mid X=x_t}$ denotes the conditional distribution of $Y$ given $X=x_t$. We estimate $f_\tau^*$ by minimizing the empirical quantile loss over a hypothesis class $\mathcal{F}$, i.e.,
\begin{equation}
    \widehat{f}=\underset{f \in \mathcal{F}}{\argmin} \sum_{t=1}^T \rho_\tau (Y_t - f(X_t)),
\end{equation}
where $\rho_\tau(u) = u\left(\tau-\mathbf{1}{\{u<0\}}\right)$ is the quantile loss function. In addition, let $f_T$ denote the best approximation to $f_\tau^*$ within the function class $\mathcal F$, defined as
\begin{equation}
    f_T := \argmin_{f \in \mathcal{F}} \mathbb{E}\left[\sum_{t=1}^T \rho_\tau\left(Y_t - f(X_t)\right) - \sum_{t=1}^T \rho_\tau\left(Y_t - f_\tau^*(X_t)\right)\right].
\end{equation}

Next, we state the assumptions and lemmas, followed by the proofs. The proof procedure follows the ideas of \cite{2020_Padilla_and_Chatterjee}, \cite{2022_Padilla_et_al}.

\begin{assumption} [\textbf{CDF and Density Boundedness}]
    \label{assp_lemma:cdf}
    There exists a constant $\kappa>0$ such that for any $\delta = (\delta_1, \cdots, \delta_T) \in \mathbb{R}^T$ satisfying $\|\delta\|_{\infty} \leq \kappa$ we have that
    \begin{equation}
        \left| F_{Y \mid X} (f_\tau^*(X) + \delta_t) - F_{Y \mid X} (f_\tau^*(X)) \right| \, \geq \, \underline{p} \cdot \left|{\delta_t} \right|, \quad \text {a.s.}
    \end{equation}
    for $t=1, \ldots, T$ and for some constant $\underline{p}>0$. We also require that
    \begin{equation}
        \sup_{z \in \mathbb{R}} p_{Y \mid X}(z) \leq \bar{p}, \quad \text {a.s.}
    \end{equation}
    for some constant $\bar{p} > 0$, where $p_{Y \mid X}$ is the probability density function of $Y$ conditioning on $X$.
\end{assumption}


\begin{lemma} \label{lemma:15_padilla}
    Suppose Assumption \ref{assp_lemma:cdf} holds, and assume $\left\| f \right\|_{\infty} \le C_f$ for $f \in \mathcal{F}$. Let $\widetilde{C}_f := \max\{1, 2C_f\}$, then we have
    \begin{equation}
        \Bigl\|f - f_T\Bigl\|_{\ell_2}^2 \leq \frac{\widetilde{C}_f}{c} \Biggl( \mathbb{E}\Bigl[ \rho_\tau\bigl(Y-f(X)\bigl) - \rho_\tau\bigl(Y - f_T(X)\bigl)\Bigl] + \Bigl\|f_T - f_\tau^*\Bigl\|_{\infty} \Bigl\|f - f_T\Bigl\|_{\ell_2} \Biggl),
    \end{equation}
    where $c \gtrsim \min\{\underline{p}, \underline{p}\kappa\}$ is a positive constant, and $\left\|f_T - f_\tau^*\right\|_{\infty} \leq c_0$ for a sufficiently small constant $c_0$.
\end{lemma}

\begin{proof}
\cite{1998_Knight}’s identity states that for any two scalars $u$ and $w$, $\rho_{\tau}(u-w) - \rho_{\tau}(u) = -w(\tau - 1\{u \leq 0\}) + \int_0^w(1\{u \leq s\} - 1\{u \leq 0\}) \mathrm{d}s$. Let $w_{T}(X) := f(X)-f_T(X)$. Therefore,
\begin{equation}
    \begin{alignedat}{2}
        &\rho_\tau\bigl(Y - f(X)\bigl) - \rho_\tau\bigl(Y - f_T(X)\bigl)& 
        =& -w_{T}(X) \cdot \Bigl( \tau - 1\bigl\{Y \leq f_T(X)\bigl\} \Bigl) \\
        &&&+ \int_0^{w_{T}(X)} \Bigl( 1\bigl\{Y \leq f_T(X) + z\bigl\} - 1\bigl\{Y \leq f_T(X)\bigl\} \Bigl) \mathrm{d}z \\
        &&=& -w_{T}(X) \Bigl( \tau - 1\bigl\{Y \leq f^*_{\tau}(X)\bigl\} \Bigl) \\
        &&&- w_{T}(X) \Bigl( 1\bigl\{Y \leq f^*_{\tau}(X)\bigl\} - 1\bigl\{Y \leq f_T(X)\bigl\} \Bigl)\\
        &&&+ \int_0^{w_{T}(X)} \Bigl(1\bigl\{Y \leq f_T(X) + z\bigl\} - 1\bigl\{Y \leq f_T(X)\bigl\}\Bigl) \mathrm{d}z.
    \end{alignedat}
\end{equation}
Taking expectations and using Fubini’s theorem yields
\begin{equation} \label{eq:loss-loss_lemma1}
    \begin{alignedat}{2}
        &\mathbb{E}\Bigl[\rho_\tau\bigl(Y - f(X)\bigl) - \rho_\tau\bigl(Y - f_T(X)\bigl) \Bigl]& 
        =& \; \mathbb{E}\Bigl[ -w_{T}(X) \cdot\Bigl( \tau - 1\bigl\{Y \leq f^*_{\tau}(X)\bigl\} \Bigl) \Bigl] \\
        &&&- \mathbb{E}\Bigl[ w_{T}(X) \cdot \Bigl( 1\bigl\{Y \leq f^*_{\tau}(X)\bigl\} - 1\bigl\{Y \leq f_T(X)\bigl\} \Bigl)\Bigl]\\
        &&&+ \mathbb{E}\Biggl[ \int_0^{w_{T}(X)} \Bigl(1\bigl\{Y \leq f_T(X) + z\bigl\} - 1\bigl\{Y \leq f_T(X)\bigl\}\Bigl) \mathrm{d}z \Biggl] \\
        &&=& \; \mathbb{E}\Biggl[ -w_{T}(X) \cdot \mathbb{E}\Bigl[ \tau - 1\bigl\{Y \leq f^*_{\tau}(X)\bigl\} \Bigl| X \Bigl] \Biggl] \\
        &&&- \mathbb{E}\Biggl[ w_{T}(X) \cdot \mathbb{E}\Bigl[ 1\bigl\{Y \leq f^*_{\tau}(X)\bigl\} - 1\bigl\{Y \leq f_T(X)\bigl\} \Bigl|X \Bigl] \Biggl]\\
        &&&+ \mathbb{E}\Biggl[ \int_0^{w_{T}(X)} \mathbb{E}\Bigl[ 1\bigl\{Y \leq f_T(X) + z\bigl\} \Bigl| X \Bigl] - \mathbb{E}\Bigl[1\bigl\{Y \leq f_T(X)\bigl\} \Bigl| X \Bigl] \mathrm{d}z\Biggl] \\
        &&=& \; \mathbb{E}\Biggl[ -w_{T}(X) \cdot \Bigl( \tau - F_{Y|X}\bigl(f^*_{\tau}(X)\bigl) \Bigl) \Biggl] \\
        &&&- \mathbb{E}\Biggl[ w_{T}(X) \cdot \Bigl( F_{Y|X}(f^*_{\tau}(X)) - F_{Y|X}(f_T(X)) \Bigl) \Biggl]\\
        &&&+ \mathbb{E}\Biggl[ \int_0^{w_{T}(X)} \Bigl( F_{Y|X}(f_T(X)+z) - F_{Y|X}(f_T(X)) \Bigl) \mathrm{d}z \Biggl].
    \end{alignedat}
\end{equation}
The first term is zero because $f^*_{\tau}(X)$ is exactly the $\tau$-th quantile. And we can simplify the second term by Jensen's inequality and the triangle inequality:
\begin{equation}
    \begin{split}
        \Biggl| \mathbb{E}\Bigl[ w_{T}(X) \cdot \Bigl( F_{Y|X}(f^*_{\tau}(X)) - F_{Y|X}(f_T(X) \Bigl) \Bigl] \Biggl| 
        &\le \mathbb{E}\Biggl[ \Biggl| w_{T}(X) \cdot \Bigl( F_{Y|X}(f^*_{\tau}(X)) - F_{Y|X}(f_T(X)) \Bigl) \Biggl| \Biggl] \\
        &\le \mathbb{E}\Bigl[ \Bigl|w_{T}(X)\Bigl| \cdot \Bigl| F_{Y|X}(f^*_{\tau}(X)) - F_{Y|X}(f_T(X)) \Bigl| \Bigl] \\
        &\le \bar{p} \; \mathbb{E}\Bigl[ \bigl|w_{T}(X)\bigl| \cdot \bigl|f_\tau^*(X) - f_T(X)\bigl| \Bigl],
    \end{split}
\end{equation}
where $\bar{p} > 0$ is the upper bound constant on the conditional probability density from Assumption \ref{assp_lemma:cdf}. For the third term, we need to consider two cases: $\left|w_T(X)\right| \le \kappa$ and $\left|w_T(X)\right| > \kappa$. When $\left|w_T(X)\right| \le \kappa$, the lower bound assumption in Assumption \ref{assp_lemma:cdf} ensures that the conditional CDF is locally Lipschitz with slope at least $\underline{p} > 0$. Therefore, we have
\begin{equation}
    \begin{split}
        \mathbb{E}\Biggl[ \int_0^{w_{T}(X)} \Bigl( F_{Y|X}(f_T(X)+z) - F_{Y|X}(f_T(X)) \Bigl) \mathrm{d}z \Biggl] 
        &\ge \mathbb{E}\Biggl[ \int_0^{w_{T}(X)} \underline{p} |z| \; \mathrm{d}z \Biggl] \\
        &\ge \mathbb{E}\Biggl[ \Bigl[ \frac{1}{2} \underline{p} z^2 \Bigl]^{w_{T}(X)}_0 \Biggl] \\
        &= \frac{1}{2} \underline{p} \;\mathbb{E}\Bigl[ w_{T}^2(X) \Bigl].
    \end{split}
\end{equation}
Suppose then that $w_T(X) > \kappa$. Then, we restrict the integration to the interval $\left[\kappa / 2, w_T(X)\right]$ to ensure that the lower bound applies uniformly. By monotonicity of the CDF, this gives
\begin{equation}
    \begin{split}
        &\mathbb{E}\Biggl[ \int_0^{w_{T}(X)} \Bigl( F_{Y|X}(f_T(X)+z) - F_{Y|X}(f_T(X)) \Bigl) \mathrm{d}z \Biggl] \\
        \ge& \; \mathbb{E}\Biggl[ \int_{\kappa/2}^{w_{T}(X)} \Bigl( F_{Y|X}(f_T(X)+z) - F_{Y|X}(f_T(X)) \Bigl) \mathrm{d}z \Biggl] \\
        \ge& \; \mathbb{E}\Biggl[ (w_T(X)- \kappa/2) \; \Bigl( F_{Y|X}(f_T(X)+\kappa/2) - F_{Y|X}(f_T(X)) \Bigl) \Biggl] \\
        \ge& \; \mathbb{E}\Biggl[(w_T(X)- \kappa/2) \cdot \underline{p} \left|\frac{\kappa}{2}\right| \Biggl]\\
        \ge& \; \frac{\underline{p} \kappa}{4} \; \mathbb{E}\bigl[w_T(X)\bigl].
    \end{split}
\end{equation}
The first inequality holds because the integrand is nonnegative for positive $z$, the second inequality is due to bounding the integral by the product of its length and the minimum value on the interval, the third inequality applies the lower bound from Assumption \ref{assp_lemma:cdf}, and the last inequality uses $w_T(X) > \kappa \Rightarrow w_T(X) - \kappa/2 > w_T(X)/2$. Note that the case $w_T(X) < -\kappa$ is symmetric, and can be covered by taking the absolute value of $w_T(X)$. Combining both regimes, we obtain
\begin{equation}
    \begin{split}
        \mathbb{E}\Biggl[ \int_0^{w_{T}(X)} \Bigl( F_{Y|X}(f_T(X)+z) - F_{Y|X}(f_T(X)) \Bigl) \mathrm{d}z \Biggl] 
        &\ge \mathbb{E}\Bigl[ \min \bigl\{ \frac{\underline{p}}{2} \, w_T^2(X), \frac{\underline{p} \kappa}{4} \, \left|w_T(X)\right| \bigl\} \Bigl] \\
        &\ge c \cdot \mathbb{E}\Bigl[ \min \bigl\{ w_T^2(X), \left|w_T(X)\right| \bigl\} \Bigl],
    \end{split}
\end{equation}
where $c>0$ and $c \asymp \min \{\underline{p}, \underline{p} \kappa \}$. Recall $w_{T}(X) = f(X) - f_T(X)$, and $ \| w_{T}(X) \|_{\infty} \le 2C_f$. Therefore, 
\begin{equation}
    \begin{split}
        \min \left\{w_T^2(X), \left|w_T(X)\right|\right\} 
        &= w_T^2(X) \cdot \mathbf{1}{\{|w_T(X)| \leq 1\}} + \left| w_T(X) \right| \cdot \mathbf{1}{\{|w_T(X)|>1\}} \\
        &= w_T^2(X) \cdot \mathbf{1}{\{|w_T(X)| \leq 1\}} + \frac{\left| w_T(X) \right|}{w_T^2(X)} w_T^2(X) \cdot \mathbf{1}{\{|w_T(X)|>1\}} \\
        &\ge w_T^2(X) \cdot \mathbf{1}{\{|w_T(X)| \leq 1\}} + \frac{1}{2C_f} w_T^2(X) \cdot \mathbf{1}{\{|w_T(X)|>1\}} \\
        &\ge \min \{1, \frac{1}{2C_f}\} \cdot w_T^2(X).
    \end{split}
\end{equation}
Then, \eqref{eq:loss-loss_lemma1} becomes
\begin{equation}
    \begin{alignedat}{2}
        &\mathbb{E}\Bigl[\rho_\tau\bigl(Y - f(X)\bigl) - \rho_\tau\bigl(Y - f_T(X)\bigl) \Bigl]
        &\ge& - \bar{p} \; \mathbb{E}\Bigl[ \bigl|w_{T}(X)\bigl| \cdot \bigl|f_\tau^*(X) - f_T(X)\bigl| \Bigl] + \frac{c}{\widetilde{C}_f} \, \mathbb{E}\Bigl[w_{T}^2(X)\Bigl] \\
        &&\ge& - \bar{p} \; \sqrt{\mathbb{E}\Bigl[\bigl|w_{T}(X)\bigl|^2\Bigl]} \sqrt{\mathbb{E}\Bigl[\bigl|f_\tau^*(X) - f_T(X)\bigl|^2\Bigl]} \\ 
        &&&+ \frac{c}{\widetilde{C}_f} \, \mathbb{E}\Bigl[w_{T}^2(X)\Bigl] \\
        &&\ge& - \bar{p} \; \Bigl\|f_T - f_\tau^*\Bigl\|_{\infty} \Bigl\|f - f_T\Bigl\|_{\ell_2} + \frac{c}{\widetilde{C}_f} \Bigl\|f - f_T\Bigl\|_{\ell_2}^2.
    \end{alignedat}
\end{equation}
The second inequality is by Cauchy–Schwarz, and $\widetilde{C}_f := \max\{1, 2C_f\}$. Rearranging gives us
\begin{equation}
    \Bigl\|f - f_T\Bigl\|_{\ell_2}^2 \leq \frac{\widetilde{C}_f}{c} \Biggl( \mathbb{E}\Bigl[ \rho_\tau\bigl(Y-f(X)\bigl) - \rho_\tau\bigl(Y - f_T(X)\bigl)\Bigl] + \Bigl\|f_T - f_\tau^*\Bigl\|_{\infty} \Bigl\|f - f_T\Bigl\|_{\ell_2} \Biggl),
\end{equation}
where $c$ absorbs the other constants. This completes the proof.
\end{proof}

\begin{lemma} \label{lemma:17_padilla}
    Suppose Lemma \ref{lemma:15_padilla} holds. Let $M_T(f) := \mathbb E\big[\rho_\tau(Y - f(X)) - \rho_\tau(Y - f_T(X))\big]$ and $\widehat{M}_T(f)$ be its empirical counterpart. Then, we have
    \begin{equation}
        \Bigl\|\widehat{f} - f_{\tau}^*\Bigl\|_{\ell_2}^2 \; \le \; 2\left(\frac{\widetilde{C}_f^2}{c^2} + 1\right) \inf_{f \in \mathcal{F}} \Bigl\|f - f_\tau^*\Bigl\|_{\infty}^2 + \frac{4\widetilde{C}_f}{c}\Biggl( M_T(\widehat{f}) - \widehat{M}_T(\widehat{f}) \Biggl),
    \end{equation}
    where $c \gtrsim \min\{\underline{p}, \underline{p}\kappa\}$, $\widetilde{C}_f = \max\{1, 2C_f\}$, and $\inf_{f\in\mathcal{F}}\|f-f_\tau^*\|_\infty^2 < c_0$ with a sufficiently small $c_0 > 0$.
\end{lemma}

\begin{proof}
By Lemma \ref{lemma:15_padilla} we have that
\begin{equation}
    \begin{alignedat}{2}
        &\Bigl\|\widehat{f} - f_T\Bigl\|_{\ell_2}^2 
        &\leq& \frac{\widetilde{C}_f}{c} \Biggl( \mathbb{E}\Bigl[ \rho_\tau\bigl(Y-\widehat{f}(X)\bigl) - \rho_\tau\bigl(Y - f_T(X)\bigl)\Bigl] + \Bigl\|f_T - f_\tau^*\Bigl\|_{\infty} \Bigl\|\widehat{f} - f_T\Bigl\|_{\ell_2} \Biggl)\\
        &&\leq& \frac{\widetilde{C}_f}{c} \Biggl( \mathbb{E}\Bigl[ \rho_\tau\bigl(Y-\widehat{f}(X)\bigl) - \rho_\tau\bigl(Y - f_T(X)\bigl)\Bigl] + \Bigl\|f_T - f_\tau^*\Bigl\|_{\infty} \Bigl\|\widehat{f} - f_T\Bigl\|_{\ell_2} \\
        &&&- \frac{1}{T} \sum_{t=1}^T \rho_\tau\left(y_t - \widehat{f}(x_t)\right) + \frac{1}{T} \sum_{t=1}^T \rho_\tau\Bigl(y_t - f_T(x_t)\Bigl) \Biggl).
    \end{alignedat}
\end{equation}
The last inequality follows from the fact $\widehat{f}$ is the global empirical risk minimizer over the sample. Denote
\begin{equation}
    M_T(f) := \mathbb E\big[\rho_\tau(Y - f(X)) - \rho_\tau(Y - f_T(X))\big],
\end{equation}
and 
\begin{equation}
    \widehat{M}_T(f) := \frac{1}{T}\sum_{t=1}^T \rho_\tau\bigl(y_t - f(x_t)\bigr) - \rho_\tau\bigl(y_t - f_T(x_t)\bigr).
\end{equation}
Note that $(X,Y)$ is independent copy of $(x_t, y_t)$.
Then, we have
\begin{equation}
    \Bigl\|\widehat{f} - f_T\Bigl\|_{\ell_2}^2 \le \frac{\widetilde{C}_f}{c} \Biggl( M_T(\widehat{f}) - \widehat{M}_T(\widehat{f}) + \Bigl\|f_T - f_\tau^*\Bigl\|_{\infty} \Bigl\|\widehat{f} - f_T\Bigl\|_{\ell_2} \Biggl).
\end{equation}
By Young's inequality,
\begin{equation}
    \Bigl\|f_T - f_\tau^*\Bigl\|_{\infty} \Bigl\|\widehat{f} - f_T\Bigl\|_{\ell_2} \le \frac{\widetilde{C}_f}{2c} \Bigl\|f_T - f_\tau^*\Bigl\|_{\infty}^2 + \frac{c}{2\widetilde{C}_f} \Bigl\|\widehat{f} - f_T\Bigl\|_{\ell_2}^2.
\end{equation}
Substituting this bound and rearranging gives
\begin{equation} \label{eq:f_hat-f_T}
    \begin{split}
        \Bigl\|\widehat{f} - f_T\Bigl\|_{\ell_2}^2 
        &\le \frac{\widetilde{C}_f}{c} \Biggl( M_T(\widehat{f}) - \widehat{M}_T(\widehat{f}) + \frac{\widetilde{C}_f}{2c} \Bigl\|f_T - f_\tau^*\Bigl\|_{\infty}^2 + \frac{c}{2\widetilde{C}_f} \Bigl\|\widehat{f} - f_T\Bigl\|_{\ell_2}^2 \Biggl) \\
        &\le \frac{2\widetilde{C}_f}{c}\Bigl( M_T(\widehat{f}) - \widehat{M}_T(\widehat{f}) \Bigl) + \frac{\widetilde{C}_f^2}{c^2}\Bigl\|f_T - f_\tau^*\Bigl\|_{\infty}^2.
    \end{split}
\end{equation}
By the triangle inequality and $(u+v)^2 \leq 2u^2+2v^2$, we obtain
\begin{equation} \label{eq:f_hat-f*}
    \Bigl\|\widehat{f} - f_{\tau}^*\Bigl\|_{\ell_2}^2 \; = \; \Bigl\|\widehat{f} - f_T + f_T -f^*_{\tau}\Bigl\|_{\ell_2}^2 \; \le \; 2\Bigl\| \widehat{f} - f_T \Bigl\|_{\ell_2}^2 + 2\Bigl\|f_T - f_\tau^*\Bigl\|_{\ell_2}^2.
\end{equation}
Replacing $\Bigl\| \widehat{f} - f_T \Bigl\|_{\ell_2}^2$ in \eqref{eq:f_hat-f*} with \eqref{eq:f_hat-f_T} yields
\begin{equation}
    \Bigl\|\widehat{f} - f_{\tau}^*\Bigl\|_{\ell_2}^2 \; \le \; \frac{4\widetilde{C}_f}{c}\Biggl( M_T(\widehat{f}) - \widehat{M}_T(\widehat{f}) \Biggl) + 2\left(\frac{\widetilde{C}_f^2}{c^2}+1\right) \Bigl\|f_T - f_\tau^*\Bigl\|_{\infty}^2,
\end{equation}
where we use the fact that the $\ell_2$ norm is bounded above by the sup-norm under this case. Finally, since $f_T$ is the best approximation of $f^*_{\tau} \in \mathcal{F}$, we have $\left\|f_T - f_\tau^*\right\|_{\infty}^2 \leq \inf_{f \in \mathcal{F}} \left\|f - f_\tau^*\right\|_{\infty}^2$. Thus,
\begin{equation}
    \Bigl\|\widehat{f} - f_{\tau}^*\Bigl\|_{\ell_2}^2 \; \le \; 2\left(\frac{\widetilde{C}_f^2}{c^2} + 1\right) \inf_{f \in \mathcal{F}} \Bigl\|f - f_\tau^*\Bigl\|_{\infty}^2 + \frac{4\widetilde{C}_f}{c}\Biggl( M_T(\widehat{f}) - \widehat{M}_T(\widehat{f}) \Biggl).
\end{equation}
This completes the proof.
\end{proof}

\begin{lemma}\label{lemma:19_padilla}
    Suppose that Lemmas \ref{lemma:15_padilla}-\ref{lemma:17_padilla} hold. Let $\|\cdot\|_T$ denote the empirical $\ell_2$--norm on the sample $\{(X_t, Y_t)\}_{t=1}^T$ of i.i.d. draws from the joint distribution of $(X,Y)$. Suppose $\max\{\|\widehat f - f_T\|_{\ell_2}, \; \|\widehat f - f_T\|_T\}\le r_0$ for $r_0>0$. Then, with probability at least $1-e^{-\gamma}$, we have
    \begin{equation}
        M_T(\widehat f)-\widehat M_T(\widehat f) \;\lesssim\; \delta \;+\; r_0\sqrt{\frac{V_{(\mathcal F,\|\cdot\|_\infty)}(\delta)}{T}} \;+\; r_0\sqrt{\frac{2\gamma}{T}} \;+\; \frac{16C_f\,\gamma}{3T},
    \end{equation}
    for any $\delta\in(0,r_0)$ and any $\gamma>0$, and $V_{(\mathcal F,\|\cdot\|_\infty)}(\delta):=\log\mathcal N(\delta,\mathcal F,\|\cdot\|_\infty)$ is the covering number under the sup-norm.
\end{lemma}

\begin{proof}
Denote that $\mathcal{G} := \left\{g: g(x, y) = \rho_\tau\bigl(y - f(x)\bigl) - \rho_\tau\bigl(y - f_T(x)\bigl), \; f \in \mathcal{F}, \; \left\|f - f_T\right\|_{T} \leq r_0\right\}$. Since the quantile loss is $1$-Lipschitz\footnote{For all real $u$, $v$, $\left|\rho_\tau(u) - \rho_\tau(v)\right| \leq|u-v|$, because the sub-gradient $\tau-\mathbf{1}{\{u<0\}}$ is bounded in $[-1,1]$.}, we have 
\begin{equation}
    \left|g(x, y)\right| = \left|\rho_\tau(y-f(x)) - \rho_\tau\left(y - f_T(x)\right)\right| \leq\left|f(x) - f_T(x)\right| \le 2C_f.
\end{equation}
So $g$ maps into the interval $\left[-2 C_f, 2 C_f\right]$. Further, $\left\|\widehat{f} - f_T\right\|_{T} \leq r_0$ implies that its corresponding $g_{\widehat{f}}$ lies in $\mathcal{G}$, and trivially,
\begin{equation}
    M_T(\widehat{f}) - \widehat{M}_T(\widehat{f}) \leq \sup_{g \in \mathcal{G}}\left\{\mathbb{E}\left[g(X, Y)\right] - \frac{1}{T} \sum_{t=1} g(x_t, y_t)\right\}.
\end{equation}
Then for $\xi_1, \ldots, \xi_T$ independent Rademacher variables independent of $\bigl\{\left(x_t, y_t\right)\bigl\}_{t=1}^T$, we have the upper bound with probability at least $1-\exp{(-\gamma)}$ by Theorem 2.1 from \cite{2005_Bartlett_et_al}:
\begin{equation}
    \begin{split}
        \sup_{g \in \mathcal{G}}\Bigl\{\mathbb{E}\left[g(X, Y)\right] &-  \frac{1}{T} \sum_{t=1} g(x_t, y_t)\Bigl\} \le \\ 
        &4\cdot \mathbb{E}\left[ \sup_{g \in \mathcal{G}} \frac{1}{T} \sum_{t=1}^T \xi_t \, g(x_t, y_t) \; \Bigl\vert \; (x_1, y_1), \ldots, (x_T, y_T) \right] + r_0 \sqrt{\frac{2\gamma}{T}} + \frac{16 C_f \gamma}{3T}.
    \end{split}
\end{equation}
Then, it remains to bound the conditional Rademacher average (i.e., the first term) in the above equation
\begin{equation}
    \begin{split}
        \mathbb{E}_{\xi}\left[\sup _{g \in \mathcal{G}} \frac{1}{T} \sum_{t=1}^T \xi_t g(x_t, y_t)\right] 
        &\leq \mathbb{E}_{\xi}\left[ \sup_{f \in \mathcal{F}, \|f\|_{\infty} \leq C_f, \left\|f_T - f\right\|_T \leq r_0} \frac{1}{T} \sum_{t=1}^T \xi_t \bigl(f(x_t) - f_T(x_t)\bigl)\right] \\
        &\leq \inf_{0 < \alpha < r_0} \left\{4 \alpha + \frac{12}{\sqrt{T}} \int_\alpha^{r_0} \sqrt{\log \mathcal{N}\left(\delta, \mathcal{F}, \|\cdot\|_T\right)} \; \mathrm{d} \delta \right\} \\
        &\leq \inf_{0 < \alpha < r_0} \left\{4 \alpha + \frac{12}{\sqrt{T}} \int_\alpha^{r_0} \sqrt{\log \mathcal{N}\left(\delta, \mathcal{F}, \|\cdot\|_{\infty} \right)} \; \mathrm{d} \delta \right\} \\
        &\leq \inf_{0 < \alpha < r_0} \left\{4 \alpha + \frac{12 (r_0 - \alpha)}{\sqrt{T}} \sqrt{\log \mathcal{N}\left(\alpha, \mathcal{F}, \|\cdot\|_{\infty} \right)} \; \right\} \\
        &\leq \inf_{0 < \alpha < r_0} \left\{4 \alpha + \frac{12 r_0}{\sqrt{T}} \sqrt{\log \mathcal{N}\left(\alpha, \mathcal{F}, \|\cdot\|_{\infty} \right)} \; \right\} \\
        &\leq \inf_{0 < \alpha < r_0} \left\{4 \alpha + 12 r_0\sqrt{\frac{\log \mathcal{N}\left(\alpha, \mathcal{F}, \|\cdot\|_{\infty} \right)}{T}} \; \right\}.
    \end{split}
\end{equation}
The first inequality follows by Talagrand's contraction\footnote{
Talagrand's contraction: Let $\Phi_1, \cdots, \Phi_T$ be $L$-Lipschitz functions, $\xi_1, \ldots, \xi_T$ be Rademacher random variables, and $S := \left\{x_1, \cdots, x_T\right\}$ be a random i.i.d. sample. Then, for any hypothesis set $\mathcal{F}$ of real-valued functions, the following inequality holds
\begin{equation}
    \frac{1}{T} \mathbb{E}_{\xi}\left[\sup_{f \in \mathcal{F}} \sum_{t=1}^T \xi_t \cdot \left(\Phi_t \circ f\right)(x_t)\right] \leq \frac{L}{T} \mathbb{E}_{\xi}\left[\sup_{f \in \mathcal{F}} \sum_{t=1}^T \xi_t \cdot f(x_t)\right].
\end{equation}}.
We get the second inequality by Dudley’s chaining for any $0 < \alpha < r_0$, where $\mathcal{N}\left(\delta, \mathcal{F},\|\cdot\|_T\right)$ is the covering number of $\mathcal{F}$ with $\|\cdot\|_T$ balls of radius $\delta$. The third is due to that the empirical norm is smaller than the sup-norm, and the remaining steps use the monotonicity of the covering number in $\delta$ and the monotonicity of the integral. Combining the above would lead to the the conclusion.
\end{proof}

\begin{remark}
    For Lemma~\ref{lemma:19_padilla} we only require that the class $\mathcal G=\{(x,y)\mapsto\rho_\tau(y-f(x))-\rho_\tau(y-f_T(x)): f\in\mathcal F\}$ admits a bounded measurable envelope. Under the uniform bound $\sup_{f\in\mathcal F}\|f\|_\infty\le C_f$ we have $|g|\le 2C_f$ for all $g\in\mathcal G$, so the integrability conditions $\mathbb E[\sup_{g\in\mathcal G}|g|]<\infty$
    and $\mathbb E[\sup_{g\in\mathcal G}g^2]<\infty$ hold automatically.
\end{remark}

\begin{lemma} \label{lemma:last}
    Consider Model \eqref{model_in_lemma} with $\tau \in(0,1)$, $u_{\tau, i}$ independent from $S_X$, then under the notations of Lemmas \ref{lemma:15_padilla}-\ref{lemma:19_padilla}, for an arbitrary $r_0>\delta>0$ and $\gamma > 0$, with probability at least $1-\exp (-\gamma)$, we have that
    \begin{equation}
        \begin{split}
            \Bigl\|\widehat{f} - f_{\tau}^*\Bigl\|_{\ell_2}^2 \; \lesssim \; \left(\frac{\widetilde{C}_f^2}{c^2} + 1\right) \inf_{f \in \mathcal{F}} \Bigl\|f - f_\tau^*\Bigl\|_{\infty}^2 + \frac{\widetilde{C}_f}{c}\Biggl( \delta + r_0 \sqrt{\frac{V_{\left(\mathcal{F},\|\cdot\|_L \infty\right)}(\delta)}{T}} + r_0 \sqrt{\frac{\gamma}{T}} + \frac{C_f \gamma}{T} \Biggl).
        \end{split}
    \end{equation}
    where $\widetilde{C}_f = \max \left\{2 C_f, 1\right\}$, $c \asymp \min \{\underline{p}, \underline{p} \kappa\}$, and $V_{(\mathcal{F},\|\cdot\|_\infty)}(\delta) := \log \mathcal{N}(\delta, \mathcal{F}, \|\cdot\|_\infty)$ is the covering number under the sup-norm.
\end{lemma}

\begin{proof}
By Lemma \ref{lemma:17_padilla}, we have 
\begin{equation}
    \Bigl\|\widehat{f} - f_{\tau}^*\Bigl\|_{\ell_2}^2 \; \le \; 2\left(\frac{\widetilde{C}_f^2}{c^2} + 1\right) \inf_{f \in \mathcal{F}} \Bigl\|f - f_\tau^*\Bigl\|_{\infty}^2 + \frac{4\widetilde{C}_f}{c}\Biggl( M_T(\widehat{f}) - \widehat{M}_T(\widehat{f}) \Biggl),
\end{equation}
Plugging the bound of $M_T(\widehat{f}) - \widehat{M}_T(\widehat{f})$ from Lemma \ref{lemma:19_padilla} into the above equation, we finish the proof.
\end{proof}
    \section{Constructing the Cover} \label{sec:proof_cover}

Let $q \in \mathbb{R}^{+}$ and consider a function class $\mathcal{F}$, where $f: \mathbb{R}^a \rightarrow \mathbb{R}^b$ for $f \in \mathcal{F}$. A subset $\widehat{\mathcal{F}} \subset \mathcal{F}$ is said to cover a collection of inputs $\{x_i\}_{i=1}^n$ if, for every $f \in \mathcal{F}$, there exists some $\widehat{f} \in \widehat{\mathcal{F}}$ such that
\begin{equation}
    \sup_{x_i}\left\|f\left(x_i\right) - \widehat{f}\left(x_i\right)\right\|_q < \epsilon.
\end{equation}
The smallest size of such a subset is denoted by
\begin{equation}
    N_{\infty}\left(\mathcal{F}, \epsilon, \{x_i\}_{i=1}^n, \|\cdot\|_q\right)
\end{equation}
and is referred to as the empirical covering number with respect to the samples $\left\{x_i\right\}_{i=1}^n$. The covering number of $\mathcal{F}$ for $n$ samples is then defined as the worst-case (supremum) over all possible sample sets:
\begin{equation}
    N_{\infty}\left(\mathcal{F}, \epsilon, n, \|\cdot\|_q\right) := \sup_{\{x_i\}_{i=1}^n} N_{\infty}\left(\mathcal{F}, \epsilon, \{x_i\}_{i=1}^n, \|\cdot\|_q\right).
\end{equation}

Following \cite{2022_Edelman_et_al} and \cite{2024_Trauger_and_Tewari}, and recalling the definition and notation of the Transformer in section~\ref{sec:transformer}, we define the weights of the $l$-th Transformer layer as
\begin{equation}
    W_l \;:=\; \left\{ W^{(KQ)}_{l}, W^{(V)}_{l}, W^{(1)}_{l}, W^{(2)}_{l}, b^{(1)}_{l}, b^{(2)}_{l}\right\}.
\end{equation}
Note that since the pairs $\left(W_l^{(K)}, W_l^{(Q)}\right)$ and $\left(W_l^{(O)}, W_l^{(V)}\right)$ are not separately identifiable, we group them together in our notation. And let $W_{1:l} := \left\{ W_1, \ldots, W_{l-1} \right\}$. For any $l$, let
\begin{equation}
    \begin{split}
        \left\|W^{(KQ)}_{l}\right\|_2 \le B_{KQ}, &\quad \left\|W^{(V)}_{l}\right\|_2 \le B_{V} \\
        \left\|W^{(1)}_{l}\right\|_2 \le B_{W1}, &\quad \left\|W^{(2)}_{l}\right\|_2 \le B_{W2},
    \end{split}
\end{equation}
and assume $\left\| b_l^{(1)}\right\|_2 \le B_{b1}$ and $\left\| b_l^{(2)}\right\|_2 \le B_{b2}$ for $B_{b1}, B_{b2} > 0$.

We inductively define $g_{l}^{(tf)}: \mathbb{R}^{d \times n} \mapsto \mathbb{R}^{d \times n}$ starting with $g_{1}^{(tf)} \left(X ; W_{1: 1}\right) = X$ with $X \in \mathbb{R}^{d \times n}$ the initial input:
\begin{equation} \label{eq:g_tf}
    g_{l+1}^{(tf)}\left(X ; W_{1:l+1}\right) := \Pi_{norm} \circ g_l^{(ff)} \circ \Pi_{norm} \circ g_l^{(msa)} \circ g_{l}^{(tf)} \bigl(X ; W_{1:l}\bigl).
\end{equation}
where $g_l^{(msa)} \in \mathcal{G}_l^{(MSA)}$ denotes the $l$-th multi-head self-attention sublayer, $g_l^{(ff)} \in \mathcal{G}_l^{(FF)}$ represents the $l$-th feed-forward network following the attention sublayer, and $\Pi_{norm}$ denotes the layer-normalization operator applied to each column. Adding normalization or not does not affect the final results. We adopt a slightly modified version of the normalization as defined in Section~\ref{sec:transformer_details}, following \cite{2022_Edelman_et_al} and \cite{2024_Trauger_and_Tewari}. Specifically, $\Pi_{norm}$ projects each column of the input onto the unit ball. Unrolling \eqref{eq:g_tf} gives us
\begin{equation}
    g_{l+1}^{(tf)}\left(X ; W_{1:l+1}\right) := \Pi_{norm} \Biggl( \sigma_R \Biggl( \Pi_{norm} \Bigl( g^{(msa)} \bigl(g_{l}^{(tf)} \left(X ; W_{1:l}\right) ; W_{l}\bigl) \Bigl) W_{l}^{(1)} + b_{l}^{(1)}\Biggl) W_{l}^{(2)} + b_{l}^{(2)} \Biggl).
\end{equation}
With the above setups, we derive the following Lemmas.

\begin{lemma} \label{lemma:A8_Edelman}
    (Lemma A.8 in \cite{2022_Edelman_et_al}) For $\epsilon, C_l, \beta_l \geq 0$, and $l \in[1,L]$, the solution to
    \begin{equation}
        \min_{\epsilon_1, \ldots, \epsilon_L} \sum_{l=1}^L \frac{C_l}{\epsilon_l^2} \quad \text{ subject to } \quad \sum_{l=1}^L \beta_l \epsilon_l = \epsilon
    \end{equation}
    is 
    \begin{equation}
        \frac{\gamma^3}{\epsilon^2},
    \end{equation} 
    where
    \begin{equation}
        \gamma = \sum_{l=1}^L C_l^{1/3} \beta_l^{2/3} \quad \text{ and } \quad \epsilon_l = \frac{\epsilon}{\gamma}\left(\frac{C_l}{\beta_l}\right)^{1/3}.
    \end{equation}
\end{lemma}
\begin{proof}
    The proof is by using Lagrange multipliers.
\end{proof}

\begin{lemma} \label{lemma:covering_B}
    Let $\mathcal{B} = \{ b \in \mathbb{R}^d : \|b\|_2 \le B_b \}$ denote the Euclidean ball in $\mathbb{R}^d$ of radius $B_b > 0$. Then, we have
    \begin{equation}
        \log N_{\infty} \big(\varepsilon, \mathcal{B},\|\cdot\|_2\big) \le d\log(3B_b) + \frac{d}{\varepsilon^2},
    \end{equation}
    for any $0 < \epsilon \leq B_b$.
\end{lemma}

\begin{proof}
Consider the class of constant functions $\mathcal{F}_b = \{ f_b(x) = b : b \in \mathcal{B} \}$. For any $f_b, f_{b'} \in \mathcal{F}_b$, the supremum error over any input set $\mathcal{X}$ is
\begin{equation}
    \sup_{x \in \mathcal{X}} \| f_b(x) - f_{b'}(x) \|_2
    = \sup_{x \in \mathcal{X}} \| b - b' \|_2
    = \| b - b' \|_2.
\end{equation}
Hence, covering the function class $\mathcal{F}_b$ under the supremum norm is equivalent to covering the Euclidean ball $\mathcal{B}$. The unit ball in $\mathbb{R}^d$ can be covered by at most $(3 / \varepsilon)^d$ balls of radius $\varepsilon$. Scaling by $B_b$ gives
\begin{equation}
    N_{\infty}(\varepsilon, \mathcal{B}, \|\cdot\|_2)
    \le \left( \frac{3 B_b}{\varepsilon} \right)^d.
\end{equation}
Taking logarithms yields
\begin{equation}
    \log N_{\infty}\big(\varepsilon, \mathcal{B}, \|\cdot\|_2\big)
    \le d \log \left( \frac{3 B_b}{\varepsilon} \right)
    = d \log(3B_b) + d \log \frac{1}{\varepsilon}.
\end{equation}
Lastly, note that for any $a > 0$, $\log(1 / \varepsilon) = o(1 / \varepsilon^a)$ as $\varepsilon \to 0^+$.  
Choosing $a = 2$, we obtain the looser but convenient bound
\begin{equation}
    \log N_{\infty}\big(\varepsilon, \mathcal{B}, \|\cdot\|_2\big)
    \le d \log(3B_b) + \frac{d}{\varepsilon^2}.
\end{equation}
\end{proof}

\begin{lemma}[Rework of Lemma A.15 in \cite{2022_Edelman_et_al}] \label{lemma:A15_Edelman}
    Suppose $W_{1:l+1}, \widehat{W}_{1:l+1}$ satisfy the norm bounds $\left\|W^{(KQ)}_{l}\right\|_2 \le B_{KQ}$, $\left\|W^{(V)}_{l}\right\|_2 \le B_{V}$, $\left\|W^{(1)}_{l}\right\|_2 \le B_{W1}$, $\left\|W^{(2)}_{l}\right\|_2 \le B_{W2}$ for all $l$. Then we have
    \begin{equation}
        \begin{split}
            &\left\| g_{l+1}^{(tf)} \bigl(X ; W_{1:l+1}\bigl) - g_{l+1}^{(tf)} \bigl(X ; \widehat{W}_{1:l+1} \bigl) \right\|_{2, \infty} \\
            \le\;& \Biggl\| \Biggl( \sigma_R \left(\Pi_{norm} \left(g^{(msa)} \left(g_{l}^{(tf)} \left(X;\widehat{W}_{1:l}\right); W_{l}\right)\right)  W_{l}^{(1)} + b_{l}^{(1)} \right) \Biggl) \left(W_{l}^{(2)} - \widehat{W}_{l}^{(2)}\right) \Biggl\|_{2, \infty} \\
            &+ \left\| b_{l}^{(2)} - \widehat{b}_{l}^{(2)} \right\|_{2, \infty} \\
            &+ B_{W2} L_\sigma \left\| \Pi_{norm} \left(g^{(msa)} \left(g_{l}^{(tf)} \left(X ; \widehat{W}_{1:l} \right) ; \widehat{W}_{l}\right)\right) \left(W_{l}^{(1)} - \widehat{W}_{l}^{(1)}\right)\right\|_{2, \infty} \\
            &+ B_{W2} L_\sigma \Biggl\| b_{l}^{(1)} - \widehat{b}_{l}^{(1)} \Biggl\|_{2, \infty}\\
            &+ B_{W2} L_\sigma B_{W1} B_{V}\left(1+4 B_{KQ}\right) \Bigl\| g_{l}^{(tf)}\left(X ; W_{1:l}\right) - g_{l}^{(tf)}\left(X ; \widehat{W}_{1:l}\right)\Bigl\|_{2, \infty} \\
            &+ 2 B_{W2} L_\sigma B_{W1} B_{V} \left\| \left(W^{(KQ)} - \widehat{W}^{(KQ)} \right)^{\top} g_{l}^{(tf)}\left(X ; \widehat{W}_{1:l}\right)\right\|_{2, \infty} \\
            &+ B_{W2} L_\sigma B_{W1} \left\|\left(W^{(V)} - \widehat{W}^{(V)} \right) g_{l}^{(tf)}\left(X ; \widehat{W}_{1:l}\right)\right\|_{2, \infty}.
        \end{split} 
    \end{equation}
\end{lemma}

\begin{proof}
Let 
\begin{equation}
    g_{l} := g_{l}^{(tf)}\left(X ; W_{1:l}\right), \quad
    A_l := \sigma_R \left(\Pi_{norm} \left(g^{(msa)} \left(g_{l} ; W_{l}\right)\right)  W_{l}^{(1)} + b_{l}^{(1)} \right).
\end{equation}
By Lemma A.9 in \cite{2022_Edelman_et_al}, we have
\begin{equation} \label{eq:g_tf-g_tf}
    \begin{split}
        &\left\| g_{l+1}^{(tf)} \bigl(X ; W_{1:l+1}\bigl) - g_{l+1}^{(tf)} \bigl(X ; \widehat{W}_{1:l+1} \bigl) \right\|_{2, \infty} \\
        =\;& \left\| \Pi_{norm} \left(A_l W_{l}^{(2)} + b_{l}^{(2)}\right) - \Pi_{norm} \left(\widehat{A}_l \,\widehat{W}_{l}^{(2)} + \widehat{b}_{l}^{(2)}\right) \right\|_{2, \infty} \\
        \le\;& \left\| A_l W_{l}^{(2)} + b_{l}^{(2)} - \widehat{A}_l \, \widehat{W}_{l}^{(2)} - \widehat{b}_{l}^{(2)} \right\|_{2, \infty} \\
        =\;& \left\| A_l W_{l}^{(2)} - \widehat{A}_l \, \widehat{W}_{l}^{(2)} + \widehat{A}_lW_{l}^{(2)} - \widehat{A}_l W_{l}^{(2)} + b_{l}^{(2)} - \widehat{b}_{l}^{(2)} \right\|_{2, \infty} \\
        \le\;& \left\| \widehat{A}_l \left(W_{l}^{(2)} - \widehat{W}_{l}^{(2)}\right)\right\|_{2, \infty} + \left\| \left( A_l - \widehat{A}_l \right) W_{l}^{(2)} \right\|_{2, \infty} + \left\| b_{l}^{(2)} - \widehat{b}_{l}^{(2)} \right\|_{2, \infty},
    \end{split}
\end{equation}
where the last inequality follows from the triangle inequality. We take a closer look at the second term. By the spectral norm $\left\|W_{l}^{(2)}\right\|_2 \leq B_{W2}$, we have
\begin{equation}
    \left\| \left( A_l - \widehat{A}_l \right) W_{l}^{(2)} \right\|_{2, \infty} \le \left\|W_{l}^{(2)}\right\|_2 \left\| A_l - \widehat{A}_l \right\|_{2, \infty} \le B_{W2} \left\| A_l - \widehat{A}_l \right\|_{2, \infty}.
\end{equation}
Then we need to bound $\left\| A_l - \widehat{A}_l \right\|_{2, \infty}$. By Lemma A.9 in \cite{2022_Edelman_et_al} and the assumption that $\sigma_R$ is $L_\sigma$-Lipschitz, we have
\begin{equation} \label{eq:Al}
    \begin{split}
        &\left\| A_l - \widehat{A}_l \right\|_{2, \infty} \\
        \leq\;& L_\sigma \Biggl\| \Pi_{norm} \Bigl(g^{(msa)} \left(g_{l} ; W_{l}\right)\Bigl) W_{l}^{(1)} + b_{l}^{(1)} - \Pi_{norm} \Bigl(g^{(msa)} \Bigl(\widehat{g}_{l} ; \widehat{W}_{l}\Bigl)\Bigl) \widehat{W}_{l}^{(1)} - \widehat{b}_{l}^{(1)} \Biggl\|_{2, \infty} \\
        \leq\;& L_\sigma \Biggl\| \Pi_{norm} \Bigl(g^{(msa)} \left(g_{l} ; W_{l}\right)\Bigl) W_{l}^{(1)} - \Pi_{norm} \Bigl(g^{(msa)} \Bigl(\widehat{g}_{l} ; \widehat{W}_{l}\Bigl) \Bigl) \widehat{W}_{l}^{(1)} \Biggl\|_{2, \infty} \\
        &+ L_\sigma \Biggl\| b_{l}^{(1)} - \widehat{b}_{l}^{(1)} \Biggl\|_{2, \infty}
    \end{split}
\end{equation}
Looking at the first term, let 
\begin{equation}
    B_l := \Pi_{norm} \Bigl(g^{(msa)} \left(g_{l} ; W_{l}\right)\Bigl) ,
\end{equation}
then
\begin{equation}
    \begin{split}
        &L_\sigma \Biggl\| B_l W_{l}^{(1)} - \widehat{B}_{l} \,  \widehat{W}_{l}^{(1)} \Biggl\|_{2, \infty} \\
        \le\;& L_\sigma \left\| \widehat{B}_l \left(W_{l}^{(1)} - \widehat{W}_{l}^{(1)}\right)\right\|_{2, \infty} + L_\sigma \left\| \left( B_l - \widehat{B}_l \right) W_{l}^{(1)} \right\|_{2, \infty} \\
        \le\;& L_\sigma \left\| \widehat{B}_l \left(W_{l}^{(1)} - \widehat{W}_{l}^{(1)}\right)\right\|_{2, \infty} + L_\sigma B_{W1} \left\| B_l - \widehat{B}_l \right\|_{2, \infty}.
    \end{split}
\end{equation}
The last inequality uses the same idea as above by the spectral norm $\left\|W_{l}^{(1)}\right\|_2 \leq B_{W1}$. Next, we bound $\left\| B_l - \widehat{B}_l \right\|_{2, \infty}$ by Lemma A.9 in \cite{2022_Edelman_et_al} as the following
\begin{equation} \label{eq:Bl}
    \begin{split}
        &\left\| B_l - \widehat{B}_l \right\|_{2, \infty} \\
        \leq& \left\| g^{(msa)} \left(g_{l} ; W_{l}\right) - g^{(msa)} \left( \widehat{g}_{l} ; \widehat{W}_{l}\right) \right\|_{2, \infty} \\
        \leq& \left\| g^{(msa)} \left(g_{l} ; W_{l}\right) - g^{(msa)} \left(\widehat{g}_{l} ; W_{l}\right) + g^{(msa)} \left(\widehat{g}_{l} ; W_{l}\right) - g^{(msa)} \left( \widehat{g}_{l} ; \widehat{W}_{l}\right)\right\|_{2, \infty} \\
        \leq& \left\| g^{(msa)} \left(g_{l} ; W_{l}\right) - g^{(msa)} \left(\widehat{g}_{l} ; W_{l}\right) \right\|_{2, \infty} + \left\| g^{(msa)} \left(\widehat{g}_{l} ; W_{l}\right) - g^{(msa)} \left( \widehat{g}_{l} ; \widehat{W}_{l}\right) \right\|_{2, \infty}
    \end{split}
\end{equation}
where the third inequality follows from the triangle inequality. For the first term in \eqref{eq:Bl}, applying Lemma A.14 in \cite{2022_Edelman_et_al}, we have, for $\left\| W_{l}^{(V)} \right\|_2 \le B_{V}$, $\left\| W_{l}^{(KQ)} \right\|_2 \le B_{KQ}$,
\begin{equation}
    \left\| g^{(msa)} \left(g_{l} ; W_{l}\right) - g^{(msa)} \left(\widehat{g}_{l} ; W_{l}\right) \right\|_{2, \infty} \leq B_{V}\left(1+4 B_{KQ}\right) \left\| g_l - \widehat{g}_l \right\|_{2, \infty}.
\end{equation}
For the second term in \eqref{eq:Bl}, by Lemma A.13 in \cite{2022_Edelman_et_al}, we have
\begin{equation}
    \begin{split}
        &\left\| g^{(msa)} \left(\widehat{g}_{l} ; W_{l}\right) - g^{(msa)} \left( \widehat{g}_{l} ; \widehat{W}_{l}\right) \right\|_{2, \infty} \\
        \le\;& 2B_{V} \left\| \left(W^{(KQ)} - \widehat{W}^{(KQ)} \right)^{\top} \widehat{g}_l \right\|_{2, \infty} + \left\|\left(W^{(V)} - \widehat{W}^{(V)} \right) \widehat{g}_l \right\|_{2, \infty}.
    \end{split}
\end{equation}
Thus, \eqref{eq:Bl} becomes
\begin{equation}
    \begin{split}
        \left\| B_l - \widehat{B}_l \right\|_{2, \infty}
        \leq\;& B_{V}\left(1+4 B_{KQ}\right) \left\| g_l - \widehat{g}_l \right\|_{2, \infty} \\
        &+ 2 B_{V} \left\| \left(W^{(KQ)} - \widehat{W}^{(KQ)} \right)^{\top} \widehat{g}_l \right\|_{2, \infty} + \left\|\left(W^{(V)} - \widehat{W}^{(V)} \right) \widehat{g}_l \right\|_{2, \infty}.
    \end{split}
\end{equation}
Then \eqref{eq:Al} becomes
\begin{equation}
    \begin{alignedat}{2}
        &\left\| A_l - \widehat{A}_l \right\|_{2, \infty} \\
        \le\;& L_\sigma \left\| \widehat{B}_l \left(W_{l}^{(1)} - \widehat{W}_{l}^{(1)}\right)\right\|_{2, \infty} + L_\sigma B_{W1} \left\| B_l - \widehat{B}_l \right\|_{2, \infty} + L_\sigma \Biggl\| b_{l}^{(1)} - \widehat{b}_{l}^{(1)} \Biggl\|_{2, \infty} \\
        \le\;& L_\sigma \left\| \widehat{B}_l \left(W_{l}^{(1)} - \widehat{W}_{l}^{(1)}\right)\right\|_{2, \infty} + L_\sigma \Biggl\| b_{l}^{(1)} - \widehat{b}_{l}^{(1)} \Biggl\|_{2, \infty}\\
        &+ L_\sigma B_{W1} B_{V}\left(1+4 B_{KQ}\right) \left\| g_l - \widehat{g}_l \right\|_{2, \infty} \\
        &+ 2 L_\sigma B_{W1} B_{V} \left\| \left(W^{(KQ)} - \widehat{W}^{(KQ)} \right)^{\top} \widehat{g}_l \right\|_{2, \infty} \\
        &+ L_\sigma B_{W1} \left\|\left(W^{(V)} - \widehat{W}^{(V)} \right) \widehat{g}_l \right\|_{2, \infty}
    \end{alignedat}
\end{equation}
Combining the above and plugging into \eqref{eq:g_tf-g_tf} gives us
\begin{equation}
    \begin{split}
        &\left\| g_{l+1}^{(tf)} \bigl(X ; W_{1:l+1}\bigl) -g_{l+1}^{(tf)} \bigl(X ; \widehat{W}_{1:l+1} \bigl) \right\|_{2, \infty} \\
        \le\;& \left\| \widehat{A}_l \left(W_{l}^{(2)} - \widehat{W}_{l}^{(2)}\right)\right\|_{2, \infty} + \left\| b_{l}^{(2)} - \widehat{b}_{l}^{(2)} \right\|_{2, \infty} \\
        &+ B_{W2} L_\sigma \left\| \widehat{B}_l \left(W_{l}^{(1)} - \widehat{W}_{l}^{(1)}\right)\right\|_{2, \infty} \\
        &+ B_{W2} L_\sigma \Biggl\| b_{l}^{(1)} - \widehat{b}_{l}^{(1)} \Biggl\|_{2, \infty}\\
        &+ B_{W2} L_\sigma B_{W1} B_{V}\left(1+4 B_{KQ}\right) \left\| g_l - \widehat{g}_l \right\|_{2, \infty} \\
        &+ 2 B_{W2} L_\sigma B_{W1} B_{V} \left\| \left(W^{(KQ)} - \widehat{W}^{(KQ)} \right)^{\top} \widehat{g}_l \right\|_{2, \infty} \\
        &+ B_{W2} L_\sigma B_{W1} \left\|\left(W^{(V)} - \widehat{W}^{(V)} \right) \widehat{g}_l \right\|_{2, \infty}.
    \end{split} 
\end{equation}
\end{proof}

\begin{lemma}[Rework of Lemma A.16 in \cite{2022_Edelman_et_al}] \label{lemma:A16_Edelman}
    Suppose $W_{1:L+1}, \widehat{W}_{1:L+1}$ satisfy our norm bounds. Then we have
    \begin{equation}
        \begin{split}
            &\left| g^{(tfw)} \bigl(X ; W_{1:L+1}, w\bigl) - g^{(tfw)} \bigl(X ; \widehat{W}_{1:L+1}, \widehat{w}\bigl) \right| \\
            \le\;& \|w\|_2 \cdot \Biggl\| g_{L+1}^{(tf)} \bigl(X ; W_{1:L+1}\bigl) - g_{L+1}^{(tf)} \bigl(X ; \widehat{W}_{1:L+1}\bigl) \Biggl\|_2 \\
        &+ \Biggl| \left( g_{L+1}^{(tf)} \bigl(X ; \widehat{W}_{1:L+1}\bigl) - g_{L+1}^{(tf)} \bigl(X ; \widehat{W}_{1:L+1}\bigl) \right) \cdot \Bigl( w -\widehat{w} \Bigl) \Biggl|.
        \end{split} 
    \end{equation}
\end{lemma}

\begin{proof}
In the last layer, we have $W_{L+1}^{(2)} \in \mathbb{R}^{d \times d_h}, b_{L+1}^{(2)} \in \mathbb{R}^{d}$. We add a linear mapping with the weight $w \in \mathbb{R}^{n}$,
\begin{equation}
    \begin{split}
        &\left| g_{L+1}^{(tfw)} \bigl(X ; W_{1:L+1}, w\bigl) - g_{L+1}^{(tfw)} \bigl(X ; \widehat{W}_{1:L+1}, \widehat{w}\bigl) \right| \\
        =\;& \left| g_{L+1}^{(tf)} \bigl(X ; W_{1:L+1}\bigl) \cdot w - g_{L+1}^{(tf)} \bigl(X ; \widehat{W}_{1:L+1}\bigl) \widehat{w} \right| \\
        =\;& \left| g_{L+1}^{(tf)} \bigl(X ; W_{1:L+1}\bigl) \cdot w - g_{L+1}^{(tf)} \bigl(X ; \widehat{W}_{1:L+1}\bigl) \cdot w + g_{L+1}^{(tf)} \bigl(X ; \widehat{W}_{1:L+1}\bigl) \cdot w - g_{L+1}^{(tf)} \bigl(X ; \widehat{W}_{1:L+1}\bigl) \cdot\widehat{w} \right| \\
        \le\;& \Biggl| \Bigl( g_{L+1}^{(tf)} \bigl(X ; W_{1:L+1}\bigl) - g_{L+1}^{(tf)} \bigl(X ; \widehat{W}_{1:L+1}\bigl) \Bigl) \cdot w \Biggl| \;+\; \Biggl| g_{L+1}^{(tf)} \bigl(X ; \widehat{W}_{1:L+1}\bigl) \cdot \Bigl( w -\widehat{w} \Bigl) \Biggl| \\
        \le\;& \|w\|_2 \cdot \Biggl\| g_{L+1}^{(tf)} \bigl(X ; W_{1:L+1}\bigl) - g_{L+1}^{(tf)} \bigl(X ; \widehat{W}_{1:L+1}\bigl) \Biggl\|_2 \;+\; \Biggl| g_{L+1}^{(tf)} \bigl(X ; \widehat{W}_{1:L+1}\bigl) \cdot \Bigl( w -\widehat{w} \Bigl) \Biggl|.
    \end{split} 
\end{equation}
The first inequality uses the triangle inequality. The second inequality uses the inner product by norms.
\end{proof}


\begin{lemma}[Rework of Theorem 4.2 in \cite{2024_Trauger_and_Tewari}] \label{theorem:4.2_Trauger}
    Suppose we have a log covering numbers in the form of $C_1 / \epsilon^2$ and $C_{B_x} / \epsilon^2$ for the function class $\{x \mapsto W x \mid x \in \mathcal{X}, W \in \mathcal{W}\}$ where $\|x\|_2 \leq 1, \forall x \in \mathcal{X}$ and $\|x\|_2 \leq B_x, \forall x \in \mathcal{X}$ respectively. Suppose we also have $\left\|W_{l}^{(2)}\right\|_2 \leq B_{W2}$, $\left\|W_{l}^{(1)}\right\|_2 \leq B_{W1}$, $\left\|W_{l}^{(V)}\right\|_2 \leq B_{V}$, $\left\|W_l^{(QK)}\right\|_2 \leq B_{QK}$, $\|w\|_2 \leq B_{w}$, $\left\|b_l^{(1)}\right\|_1 \leq B_{b1}$, $\left\|b_l^{(2)}\right\|_1 \leq B_{b2}$. Let $\alpha_l := \prod_{j=l+1}^L B_{W2} L_\sigma B_{W1} B_{V}\left(1+4 B_{KQ}\right)$ and
    \begin{equation}
        \begin{split}
            \beta_l^{(2)} := B_{w} \alpha_l, &\quad 
            \beta_l^{(1)} := B_{w} \alpha_l \cdot B_{W2} L_{\sigma}, \\
            \beta_l^{(KQ)} := B_{w} \alpha_l \cdot 2 B_{W2} L_\sigma B_{W1} B_{V}, &\quad
            \beta_l^{(V)} := B_{w} \alpha_l \cdot B_{W2} L_\sigma B_{W1}.
        \end{split}
    \end{equation}
    Then the covering number of $g_{L+1}^{(tfw)}$ is
    \begin{equation}
        \frac{(\widetilde{\eta} + \eta)^3}{\epsilon^2}.
    \end{equation}
    where 
    \begin{equation}
        \begin{split}
            \eta &= C_1^{\frac{1}{3}} \sum_{l=2}^L \Bigl( 2{\beta_l^{(2)}}^{\frac{2}{3}} + 2{\beta_l^{(1)}}^{\frac{2}{3}} + {\beta_l^{(KQ)}}^{\frac{2}{3}} + {\beta_l^{(V)}}^{\frac{2}{3}} \Bigl), \\
            \widetilde{\eta} &= C_{B_x}^{\frac{1}{3}} \Bigl( \beta_1^{(KQ)} \Bigl) + C_1^{\frac{1}{3}} \Bigl( 2{\beta_1^{(2)}}^{\frac{2}{3}} + 2{\beta_1^{(1)}}^{\frac{2}{3}} + 2{\beta_1^{(V)}}^{\frac{2}{3}} \Bigl) + C_1^{\frac{1}{3}}.
        \end{split}
    \end{equation}
\end{lemma}

\begin{proof}
By Lemma \ref{lemma:A16_Edelman}, we know 
\begin{equation}
    \begin{split}
        &\left| g^{(tfw)} \bigl(X ; W_{1:L+1}, w\bigl) - g^{(tfw)} \bigl(X ; \widehat{W}_{1:L+1}, \widehat{w}\bigl) \right| \\
        \le\;& \|w\|_2 \cdot \Biggl\| g_{L+1}^{(tf)} \bigl(X ; W_{1:L+1}\bigl) - g_{L+1}^{(tf)} \bigl(X ; \widehat{W}_{1:L+1}\bigl) \Biggl\|_2 + \Biggl| g_{L+1}^{(tf)} \bigl(X ; \widehat{W}_{1:L+1}\bigl) \cdot \Bigl( w -\widehat{w} \Bigl) \Biggl|.
    \end{split} 
\end{equation}
Thus, we have
\begin{equation}
    \begin{split}
        &\left| g^{(tfw)} \bigl(X ; W_{1:L+1}, w\bigl) - g^{(tfw)} \bigl(X ; \widehat{W}_{1:L+1}, \widehat{w}\bigl) \right| \\
        \le\;& \|w\|_2 \cdot \Biggl\| g_{L+1}^{(tf)} \bigl(X ; W_{1:L+1}\bigl) - g_{L+1}^{(tf)} \bigl(X ; \widehat{W}_{1:L+1}\bigl) \Biggl\|_2 +\; \epsilon^{(w)}.
    \end{split} 
\end{equation}

Next, by Lemma \ref{lemma:A15_Edelman}, we have 
\begin{equation} \label{eq:by_lemma_A15}
    \begin{split}
        &\left\| g_{l+1}^{(tf)} \bigl(X ; W_{1:l+1}\bigl) - g_{l+1}^{(tf)} \bigl(X ; \widehat{W}_{1:l+1} \bigl) \right\|_{2, \infty} \\
        \le\;& \Biggl\| \Biggl( \sigma_R \left(\Pi_{norm} \left(g^{(msa)} \left(g_{l}^{(tf)} \left(X;\widehat{W}_{1:l}\right); W_{l}\right)\right)  W_{l}^{(1)} + b_{l}^{(1)} \right) \Biggl) \left(W_{l}^{(2)} - \widehat{W}_{l}^{(2)}\right) \Biggl\|_{2, \infty} \\
        &+ \left\| b_{l}^{(2)} - \widehat{b}_{l}^{(2)} \right\|_{2, \infty} \\
        &+ B_{W2} L_\sigma \left\| \Pi_{norm} \left(g^{(msa)} \left(g_{l}^{(tf)} \left(X ; \widehat{W}_{1:l} \right) ; \widehat{W}_{l}\right)\right) \left(W_{l}^{(1)} - \widehat{W}_{l}^{(1)}\right)\right\|_{2, \infty} \\
        &+ B_{W2} L_\sigma \Biggl\| b_{l}^{(1)} - \widehat{b}_{l}^{(1)} \Biggl\|_{2, \infty}\\
        &+ B_{W2} L_\sigma B_{W1} B_{V}\left(1+4 B_{KQ}\right) \Bigl\| g_{l}^{(tf)}\left(X ; W_{1:l}\right) - g_{l}^{(tf)}\left(X ; \widehat{W}_{1:l}\right)\Bigl\|_{2, \infty} \\
        &+ 2 B_{W2} L_\sigma B_{W1} B_{V} \left\| \left(W^{(KQ)} - \widehat{W}^{(KQ)} \right)^{\top} g_{l}^{(tf)}\left(X ; \widehat{W}_{1:l}\right)\right\|_{2, \infty} \\
        &+ B_{W2} L_\sigma B_{W1} \left\|\left(W^{(V)} - \widehat{W}^{(V)} \right) g_{l}^{(tf)}\left(X ; \widehat{W}_{1:l}\right)\right\|_{2, \infty}.
    \end{split} 
\end{equation}
Notice that if $X_t \in \mathbb{R}^{d \times n}$ with $t \in[1,T]$, and $W \in \mathbb{R}^{d \times d}$
\begin{equation}
    \max_{t \in [1,T]} \left\|(W - \widehat{W}) X_t\right\|_{2, \infty} = \max_{t \in[1,T], i \in[1,n]} \left\|(W - \widehat{W}) X_{t,i}\right\|
\end{equation}
Therefore we can use the covering number bounds to bound the values in the $\|\cdot\|_{2, \infty}$. 

First, we create a cover for multilayer Transformer with inputs $X_1, \ldots, X_T$ and $\|X_t\|_2 \le B_x$. Let $\mathcal{W}_{l}^{(V)}$, $\mathcal{W}_{l}^{(KQ)}$, $\mathcal{W}_{l}^{(1)}$, $\mathcal{W}_{l}^{(2)}$, $\mathcal{B}_{l}^{(1)}$, $\mathcal{B}_{l}^{(2)}$ be the sets of all possible values for $W_{l}^{(V)}$, $W_{l}^{(KQ)}$, $W_{l}^{(1)}$, $W_{l}^{(2)}$, $b_{l}^{(1)}$, $b_{l}^{(2)}$ respectively. Let $\widehat{\mathcal{W}}_{l}^{(V)}$ cover the function class 
\begin{equation}
    \left\{x \mapsto {W^{(V)}}^{\top} x \mid W^{(V)} \in \mathcal{W}_{l}^{(V)},\|x\|_2 \leq 1\right\}.
\end{equation} 
Let $\widehat{\mathcal{W}}_{l}^{(KQ)}$ cover the function class 
\begin{equation}
    \left\{x \mapsto W^{(KQ)} x \mid W^{(KQ)} \in \mathcal{W}_{l}^{(KQ)}, \|x\|_2 \leq 1\right\}
\end{equation} 
except for $\widehat{\mathcal{W}}_{1}^{(KQ)}$, which covers the same function class, but with $\|x\| \leq B_x$. Let $\widehat{\mathcal{W}}_{l}^{(1)}$ cover the function class 
\begin{equation}
    \left\{x \mapsto {W^{(1)}}^{\top} x \mid W^{(1)} \in \mathcal{W}_{l}^{(1)},\|x\|_2 \leq 1\right\}.
\end{equation} 
And let $\widehat{\mathcal{W}}_{l}^{(2)}$ cover the function class 
\begin{equation}
    \left\{x \mapsto {W^{(2)}}^{\top} x \mid W^{(2)} \in \mathcal{W}_{l}^{(2)},\|x\|_2 \leq 1\right\}.
\end{equation} 
Let $\widehat{\mathcal{B}}_{l}^{(1)}, \widehat{\mathcal{B}}_{l}^{(2)}$ cover $b_l^{(1)}, b_l^{(2)}$, respectively. Let all of these classes be covered with $Tn$ points and let $\epsilon_{l}^{(V)}, \epsilon_{l}^{(KQ)}, \epsilon_{l}^{(W1)}, \epsilon_{l}^{(W2)}, \epsilon_{l}^{(b1)}, \epsilon_{l}^{(b2)}$ be the resolution for each cover. Also, let $\widehat{\mathcal{W}}$ cover 
\begin{equation}
    \left\{x \mapsto w^{\top} x \mid w \in \mathcal{W},\|x\|_2 \leq 1\right\}
\end{equation}
at resolution $\epsilon^{(w)}$. Thus, our final cover is
\begin{equation}
    \begin{split}
        \widehat{W}^{1:L+1} \in 
        \widehat{\mathcal{W}}_1^{(KQ)} \otimes 
        \widehat{\mathcal{W}}_1^{(V)} &\otimes 
        \widehat{\mathcal{W}}_1^{(1)} \otimes 
        \widehat{\mathcal{B}}_1^{(1)} \otimes 
        \widehat{\mathcal{W}}_1^{(2)} \otimes 
        \widehat{\mathcal{B}}_1^{(2)} \otimes \cdots \\
        &\cdots \otimes \widehat{\mathcal{W}}_L^{(KQ)} \otimes 
        \widehat{\mathcal{W}}_L^{(V)} \otimes 
        \widehat{\mathcal{W}}_L^{(1)} \otimes 
        \widehat{\mathcal{B}}_L^{(1)} \otimes 
        \widehat{\mathcal{W}}_L^{(2)} \otimes 
        \widehat{\mathcal{B}}_L^{(2)} \otimes 
        \widehat{\mathcal{W}}.
    \end{split}
\end{equation}

Then \eqref{eq:by_lemma_A15} becomes
\begin{equation}
    \begin{split}
        &\left\| g_{L+1}^{(tf)} \bigl(X ; W_{1:L+1}\bigl) - g_{L+1}^{(tf)} \bigl(X ; \widehat{W}_{1:L+1} \bigl) \right\|_{2, \infty} \\
        \le\;& \epsilon_{L}^{(W2)} + \epsilon_{L}^{(b2)} + B_{W2} L_{\sigma} \epsilon_{L}^{(W1)} + B_{W2} L_{\sigma} \epsilon_{L}^{(b1)} \\
        &+ B_{W2} L_\sigma B_{W1} B_{V}\left(1+4 B_{KQ}\right) \Bigl\| g_{L}^{(tf)}\left(X ; W_{1:L}\right) - g_{L}^{(tf)}\left(X ; \widehat{W}_{1:L}\right)\Bigl\|_{2, \infty} \\
        &+ 2 B_{W2} L_\sigma B_{W1} B_{V} \epsilon_{L}^{(KQ)} + B_{W2} L_\sigma B_{W1} \epsilon_{L}^{(V)}.
    \end{split}
\end{equation}
We can bound $\left\| g_{L}^{(tf)} \bigl(X ; W_{1:L}\bigl) - g_{L}^{(tf)} \bigl(X ; \widehat{W}_{1:L} \bigl) \right\|_{2, \infty}$ till $\left\| g_{1}^{(tf)} \bigl(X ; W_{1:1}\bigl) - g_{1}^{(tf)} \bigl(X ; \widehat{W}_{1:1} \bigl) \right\|_{2, \infty}$ iteratively. Let
\begin{equation}
    \alpha_l = \prod_{j=l+1}^L B_{W2} L_\sigma B_{W1} B_{V}\left(1+4 B_{KQ}\right),
\end{equation}
then
\begin{equation}
    \begin{split}
        & \max_{t \in T} \left| g^{(tfw)} \bigl(X_t ; W_{1:L+1}, w\bigl) - g^{(tfw)} \bigl(X_t ; \widehat{W}_{1:L+1}, \widehat{w} \bigl) \right| \\
        \leq\;& \epsilon^{(w)} + B_{w} \sum_{l=2}^L \alpha_l \Bigl( \epsilon_{l}^{(W2)} + \epsilon_{l}^{(b2)} + B_{W2} L_{\sigma} \epsilon_{l}^{(W1)} + B_{W2} L_{\sigma} \epsilon_{l}^{(b1)} \\
        &\quad \quad + 2 B_{W2} L_\sigma B_{W1} B_{V} \epsilon_{l}^{(KQ)} + B_{W2} L_\sigma B_{W1} \epsilon_{l}^{(V)} \Bigl) \\
        &+ B_{w}\alpha_1 \Bigl( \epsilon_{1}^{(W2)} + \epsilon_{1}^{(b2)} + B_{W2} L_{\sigma} \epsilon_{1}^{(W1)} + B_{W2} L_{\sigma} \epsilon_{1}^{(b1)} \\
        &\quad \quad+ 2 B_{W2} L_\sigma B_{W1} B_{V} \epsilon_{1}^{(KQ)} + B_{W2} L_\sigma B_{W1} \epsilon_{1}^{(V)} \Bigl).
    \end{split}
\end{equation}
Recall that $\epsilon_{1}^{(KQ)}$ has a different input bound than the rest. Define $\beta_l^{(2)} := B_{w} \alpha_l$, $\beta_l^{(1)} := B_{w} \alpha_l \cdot B_{W2} L_{\sigma}$, $\beta_l^{(KQ)} := B_{w} \alpha_l \cdot 2 B_{W2} L_\sigma B_{W1} B_{V}$, and $\beta_l^{(V)} := B_{w} \alpha_l \cdot B_{W2} L_\sigma B_{W1}$. To get the covering number, we solve
\begin{equation}
    \begin{split}
        \min_{\epsilon_1, \ldots, \epsilon_L} \epsilon^{(w)} &+ \sum_{l=2}^L \Biggl( \beta_l^{(2)} \left(\epsilon_l^{(W2)} + \epsilon_l^{(b2)}\right) + \beta_l^{(1)}\left(\epsilon_l^{(W 1)} + \epsilon_l^{(b1)}\right) + \beta_l^{(KQ)} \epsilon_l^{(KQ)} + \beta_l^{(V)} \epsilon_l^{(V)} \Biggl) \\
        & + \Biggl( \beta_1^{(2)}\left(\epsilon_1^{(W2)} + \epsilon_1^{(b2)}\right) + \beta_1^{(1)}\left(\epsilon_1^{(W1)} + \epsilon_1^{(b1)}\right) + \beta_1^{(KQ)} \epsilon_1^{(KQ)} + \beta_1^{(V)} \epsilon_1^{(V)} \Biggl).
    \end{split}
\end{equation}
Let $i \in \mathcal{I} := \{W, W 1, b1, b2, W 2, K Q, V\}$. Recall that we have log covering numbers in the form of $C_1 / \epsilon^2$ and $C_{B_x} / \epsilon^2$ for the function class $\{x \mapsto W x \mid x \in \mathcal{X}, W \in \mathcal{W}\}$. For $\epsilon_l^{(b1)}$ and $\epsilon_l^{(b2)}$, we apply Maurey's specification (See Theorem 3 in \cite{2002_Zhang}). Therefore, summing across all parameter blocks, the total log covering number can be bounded as
\begin{equation}
    \log N_{\infty}\left(\mathcal{F}, \epsilon, \{x_i\}_{i=1}^n, \|\cdot\|_2\right) \leq \sum_l \sum_{i} \frac{C_l^{(i)}}{\left(\epsilon_l^{(i)}\right)^2},
\end{equation}
To find the tightest possible bound, we optimally distribute the total error budget $\epsilon$ among the individual perturbations $\epsilon_l^{(i)}$. This leads to the optimization problem
\begin{equation}
    \min_{\left\{\epsilon_l^{(i)} > 0\right\}} \sum_l \sum_i \frac{C_l^{(i)}}{\left(\epsilon_l^{(i)}\right)^2} \quad \text{subject to} \quad \sum_l \sum_i \beta_l^{(i)} \epsilon_l^{(i)} = \epsilon.
\end{equation}
Applying Lemma \ref{lemma:A8_Edelman}, we have
\begin{equation}
    \begin{split}
        \eta &= C_1^{\frac{1}{3}} \sum_{l=2}^L \Bigl( 2{\beta_l^{(2)}}^{\frac{2}{3}} + 2{\beta_l^{(1)}}^{\frac{2}{3}} + {\beta_l^{(KQ)}}^{\frac{2}{3}} + {\beta_l^{(V)}}^{\frac{2}{3}} \Bigl), \\
        \widetilde{\eta} &= C_{B_x}^{\frac{1}{3}} \Bigl( \beta_1^{(KQ)} \Bigl) + C_1^{\frac{1}{3}} \Bigl( 2{\beta_1^{(2)}}^{\frac{2}{3}} + 2{\beta_1^{(1)}}^{\frac{2}{3}} + {\beta_1^{(V)}}^{\frac{2}{3}} \Bigl) + C_1^{\frac{1}{3}}.
    \end{split}
\end{equation}
Then the covering number is bounded by
\begin{equation}
    \frac{(\widetilde{\eta} + \eta)^3}{\epsilon^2}.
\end{equation}
\end{proof}

\begin{lemma}
    Under Lemma \ref{theorem:4.2_Trauger}, and suppose we have the norm bounds $\mathcal{W} = \{ W \in \mathbb{R}^{k \times d} : \|W\|_{2,1} \le B_W \}$ with $B_W > 0$ for each $W_l^{(1)}$, $W_l^{(2)}$, $W_l^{(KQ)}$, $W_l^{(V)}$, $w$, and $\mathcal{B} = \{ b \in \mathbb{R}^d : \|b\|_1 \le B_b \}$ with $B_b > 0$ for $b_l^{(1)}$ and $b_l^{(2)}$. Let the maximum of all norm bounds be $B$. Let $x_1, \ldots, x_n \in \mathbb{R}^{d}$ satisfy $\|x_i\|_2 \le B_x$, we have the log covering number of $g_{L+1}^{(tfw)} (X, W_{1:L+1},w)$ as
    \begin{equation}
        \frac{(\widetilde{\eta} + \eta)^3}{\epsilon^2}.
    \end{equation}
    where 
    \begin{equation}
        \begin{split}
            \eta =\;& \sum_{l=2}^L \Bigl(B^2 \log(dn)\Bigl)^{\frac{1}{3}} \Bigl( {\beta_l^{(2)}}^{\frac{2}{3}} + {\beta_l^{(1)}}^{\frac{2}{3}} + {\beta_l^{(KQ)}}^{\frac{2}{3}} + {\beta_l^{(V)}}^{\frac{2}{3}} \Bigl) + \bigl(B^2 \log(2d+1)\bigl)^\frac{1}{3} \Bigl( {\beta_l^{(2)}}^{\frac{2}{3}} + {\beta_l^{(1)}}^{\frac{2}{3}} \Bigl) \\
            \widetilde{\eta} =\;& \Bigl(B^2B_x^2 \log(dn)\Bigl)^{\frac{1}{3}} \Bigl( \beta_1^{(KQ)} \Bigl) + \Bigl(B^2 \log(dn)\Bigl)^{\frac{1}{3}} \Bigl( 1 + {\beta_1^{(2)}}^{\frac{2}{3}} + {\beta_1^{(1)}}^{\frac{2}{3}} + {\beta_1^{(V)}}^{\frac{2}{3}} \Bigl) \\
            &+ \bigl(B^2 \log(2d+1)\bigl)^\frac{1}{3} \Bigl( {\beta_1^{(2)}}^{\frac{2}{3}} + {\beta_1^{(1)}}^{\frac{2}{3}} \Bigl),
        \end{split}
    \end{equation}
    with $\beta_l^{(2)} := B_{w} \alpha_l$, $\beta_l^{(1)} := B_{w} \alpha_l \cdot B_{W2} L_{\sigma}$, $\beta_l^{(KQ)} := B_{w} \alpha_l \cdot 2 B_{W2} L_\sigma B_{W1} B_{V}$, and $\beta_l^{(V)} := B_{w} \alpha_l \cdot B_{W2} L_\sigma B_{W1}$ and $\alpha_l = \prod_{j=l+1}^L B_{W2} L_\sigma B_{W1} B_{V}\left(1+4 B_{KQ}\right)$.
\end{lemma}

\begin{proof}
    By applying Lemma 4.6 from \cite{2022_Edelman_et_al} and Maurey’s specification in Theorem 3 of \cite{2002_Zhang} to Lemma \ref{theorem:4.2_Trauger}, we complete the proof.
\end{proof}

\begin{lemma}[Covering number of Transformer--MLP composition]
\label{lem:cover-tf-mlp}
    Under the assumptions of Lemma~\ref{theorem:4.2_Trauger}, let $\mathcal {G}^{(TF)}=\bigl\{g_{L+1}^{(tfw)}(\cdot;W_{1:L+1},w)\bigl\}$ be the class of Transformer functions with parameters satisfying the stated norm bounds, and let $\mathcal{F}^{(MLP)}$ be the MLP class defined in Section \ref{sec:transformer} satisfying $\max_{i}\left\|W_i\right\|_{\infty} \leq B$ and $\max_{i}\left|b_i\right|_{\infty} \leq B$ for $i=0,\ldots,D$ and $B >0$. Assume that $\mathcal{F}^{(MLP)}$ is uniformly Lipschitz with $\sup_{f\in\mathcal F^{(MLP)}} \mathrm{Lip}(f) \le C$. Define the composed class
    \begin{equation}
        \mathcal{F} := \mathcal F^{(MLP)} \circ \mathcal G^{(TF)} = \{ f\circ g : f\in\mathcal F^{(MLP)},\ g\in\mathcal G^{(TF)} \}.
    \end{equation}
    Then, for any $\varepsilon>0$, the $\|\cdot\|_\infty$-covering number of $\mathcal{F}$ satisfies
    \begin{equation}
        \log N_\infty(\varepsilon,\ \mathcal{F}) \le O (d_{m} \cdot \log \Big( 16\, \varepsilon^{-1} D \, d^2 \, {d_{m}}^{2D}\Big) + \frac{4C^2(\widetilde{\eta} + \eta)^3}{\varepsilon^2}.
    \end{equation}
    where $\widetilde{\eta}$ and $\eta$ are defined as in Lemma~\ref{theorem:4.2_Trauger}.
\end{lemma}

\begin{proof}
For notational simplicity, denote the Transformer block by $g := g_{L+1}^{(tfw)}(X; W_{1:L+1}, w)$, and the MLP network by $f := f^{(mlp)}$. We consider the composed function class
\begin{equation}
    \mathcal F := \{ f \circ g : f \in \mathcal F^{(MLP)},\ g \in \mathcal G^{(TF)} \}.
\end{equation}
Since $f \in \mathcal F^{(MLP)}$ is Lipschitz continuous, for any $g, g' \in \mathcal G^{(TF)}$ we have
\begin{equation}
    \| f \circ g - f \circ g' \|_{\infty} \le \mathrm{Lip}(f)\, \| g - g' \|_{\infty}.
\end{equation}
Under the parameter constraints of \cite{2020_Schmidt-Hieber}, the MLP class is uniformly Lipschitz; that is,
\begin{equation}
    \sup_{f \in \mathcal F^{(MLP)}} \mathrm{Lip}(f) \le C < \infty.
\end{equation}
Hence, the map $g \mapsto f \circ g$ is $C$-Lipschitz uniformly over $f \in \mathcal F^{(MLP)}$.

By Lemma~3.1 in \cite{2023_Petrova_and_Wojtaszczyk}, which controls covering numbers under Lipschitz mappings, it follows that
\begin{equation}
    N_{\infty}(\varepsilon,\ f \circ \mathcal G^{(TF)}) \;\le\; N_{\infty}(\varepsilon/C,\ \mathcal G^{(TF)}).
\end{equation}
To cover the entire composed class $\mathcal{F}$, we discretize both the MLP and Transformer classes. Therefore
\begin{equation}
    N_{\infty}(\varepsilon,\ \mathcal F) \;\le\; N_{\infty}(\varepsilon/2,\ \mathcal F^{(MLP)}) \cdot N_{\infty}(\varepsilon/(2C),\ \mathcal G^{(TF)}),
\end{equation}
and further,
\begin{equation} \label{eq:comp-entropy-sum}
    \log N_{\infty}(\varepsilon,\ \mathcal F) \;\le\; \log N_{\infty}(\varepsilon/2,\ \mathcal F^{(MLP)}) + \log N_{\infty}(\varepsilon/(2C),\ \mathcal G^{(TF)}).
\end{equation}

\paragraph{MLP covering number.}
Define that $a_n \asymp b_n$ if $a_n \lesssim b_n$ and $b_n \lesssim a_n$ for two sequences. Let $\lceil\cdot\rceil$ denotes the ceiling function and $\vee$ denotes the maximum operator. Recalling the output MLP neural network in \eqref{eq:mlp_class}. By Lemma~5 of \cite{2020_Schmidt-Hieber}, for any integers $m \ge 1$, $N \ge (\beta+1)^s \vee(L_{\mathcal{I}}+1) \exp(s)$ and any $\delta>0$, we have $\mathcal{F}^{(MLP)}(D, \widetilde{d}_{m})$ with number of parameters $S^{(net)}$ such that 
\begin{equation}\label{eq:mlp-cover}
    \log N_\infty(\delta,\ \mathcal F^{(MLP)}) \le (S^{(net)}+1)\log\Big( 2\,\delta^{-1}\,(D+1) V^2\Big)
\end{equation}
where $V:= 2(d+1) \cdot \prod_{i=1}^{D-1}\left(d_{m}^{(i)} + 1\right)$, $D$ is the depth, $d_{m}^{(i)}$ the dimension of each hidden layer. Applying \eqref{eq:mlp-cover} with $\delta=\varepsilon/2$ and consider the defined MLP in section \ref{sec:transformer} gives
\begin{equation}\label{eq:mlp-cover-eps}
    \log N_\infty(\varepsilon/2,\ \mathcal F^{(MLP)}) \le (S^{(net)}+1) \log \Big( 4\, \varepsilon^{-1} \, (D+1) V^2\Big).
\end{equation}
Since $D = 8 + (m+5)\left(1 + \left\lceil\log_2(s \vee \beta)\right\rceil\right)$, $d_{m} := d_{m}^{(i)} = 6(s+\lceil\beta\rceil) N$ and $S^{(net)} \leq 141(s+\beta+1)^{3+s} N(m+6)$ by Theorem 5 in \cite{2020_Schmidt-Hieber}, combining the above result for MLP and choosing $m \asymp \log d_m$, we have
\begin{equation}
    \log N_\infty(\varepsilon/2,\ \mathcal F^{(MLP)}) = O\Bigg( \big(d_{m} \log (d_m) + 1\big) \cdot \log \Big( 16\, \varepsilon^{-1} \, (D+1) (d+1)^2 (d_{m}+1)^{2D-2} \Big) \Bigg).
\end{equation}

\paragraph{Transformer covering number.}
By Lemma~\ref{theorem:4.2_Trauger} (stated above), for any $\delta>0$,
\begin{equation}\label{eq:tf-cover}
    \log N_\infty(\delta,\ \mathcal G^{(TF)}) \le \frac{(\widetilde{\eta} + \eta)^3}{\delta^2}.
\end{equation}
Applying \eqref{eq:tf-cover} with $\delta=\varepsilon/(2C)$ yields
\begin{equation}\label{eq:tf-cover-eps}
    \log N_\infty(\varepsilon/(2C),\ \mathcal G^{(TF)}) \le \frac{(\widetilde{\eta}+\eta)^3}{(\varepsilon/(2C))^2} = \frac{4C^2(\widetilde{\eta} + \eta)^3}{\varepsilon^2}.
\end{equation}

Combining both covering numbers by substituting \eqref{eq:mlp-cover-eps} and \eqref{eq:tf-cover-eps} into \eqref{eq:comp-entropy-sum} gives
\begin{equation}
    \log N_\infty(\varepsilon,\ \mathcal{F}) \le O \Bigg(d_{m} \log (d_m) \cdot \log \Big( \varepsilon^{-1} \, D\, d\, d_{m}^{D}\Bigg) + \frac{4C^2(\widetilde{\eta} + \eta)^3}{\varepsilon^2}.
\end{equation}
\end{proof}
    \section{Bias Rate} \label{sec:bias}

We follow the setting of \cite{2022_Edelman_et_al} and slightly adapt their definitions to align with our model architecture. We begin with the following definitions.
\begin{enumerate}
    \item (\textbf{$\mathcal{I}$-sparsity}) A function $f:[0,1]^n \mapsto \mathcal{Y}$ is $\mathcal{I}$-sparse if it only depends on a fixed subset $\mathcal{I} \subseteq \{1, \cdots, n\}$ of its inputs:
    \begin{equation}
        x_i = x_i^{\prime} \quad \forall i \in \mathcal{I} \quad \Longrightarrow \quad f(x) = f\left(x^{\prime}\right).
    \end{equation}

    \item (\textbf{$\gamma$-injective}) An $\mathcal{I}$-sparse function $f$ is $\gamma$-injective if there exists $\gamma>0$ such that
    \begin{equation}
        \| f(x) - f(x') \|_\infty \;\ge\; \gamma \, \| \Pi_{\mathcal{I}} x - \Pi_{\mathcal{I}} x' \|_\infty, \qquad \forall x,x' \in [0,1]^n,
    \end{equation}
    where $\Pi_{\mathcal{I}}$ is a projection onto coordinates indexed by $\mathcal{I}$.
    
\end{enumerate}

We then make the following assumptions.
\begin{assumption}[Target Function]\label{assp:lemma_target_function}
Define $x \in [0,1]^n$ and let $\mathbb Z:[0,1]^n\to[0,1]^{d\times n}$ be a fixed embedding with $Z = \mathbb{Z}(x)$. Assume there exists $\mathcal I\subseteq\{1,\dots,n\}$ with $|\mathcal I|=s\le\min\{n,d\}$, and a $\beta$-H\"older function $\phi:[0,1]^s\to\mathbb R$ with constant $L_{\mathcal I}$, such that the target function $g^*:[0,1]^{d\times n}\to\mathbb R$ satisfies, for all $x \in [0,1]^{n}$,
\begin{equation}
    g^* \circ \mathbb{Z}(x) = \widetilde{\phi} \circ \widetilde{g}_{\mathcal I}(x),
\end{equation}
where $\widetilde{g}_{\mathcal{I}}(x)$ is $\mathcal I$-sparse, $1$-bounded and $\gamma$-injective, and $\widetilde\phi:[0,1]^d\to\mathbb R$ is defined by
\begin{equation}
    \widetilde\phi(u):=\phi\left((u_1,\dots,u_s)\right).
\end{equation}
\end{assumption}

{
\begin{assumption}[Recoverability by Self-Attention] \label{assp:lemma_recoverability}
    Assume $\mathbb{Z}$ satisfies structural conditions similar to those in Appendix B of \cite{2022_Edelman_et_al} (e.g., approximate orthogonality of positional encodings). Specifically, the embedding $Z={\mathbb{Z}(x)}$ satisfies the column-wise decomposition
\begin{equation}\label{eq:Z_decomp_short}
        Z_{:,j} = p_j + x_j u,\qquad j\in[n],
    \end{equation}
    where $\|p_j\|_2=\|u\|_2=1$, $\langle p_j,u\rangle=0$, and $|\langle p_i,p_j\rangle|\le\Delta$ for $i\neq j$. 
\end{assumption}}
Assumption~\ref{assp:lemma_recoverability} essentially requires that the key latent information be appropriately embedded in the input matrix so that attention layers can extract it effectively, as shown later. Following \cite{2022_Edelman_et_al}, we decompose the input columns into two components, and for illustrative purposes, we adopt a specific form, while acknowledging that other forms may also be viable\footnote{One form considered \cite{2022_Edelman_et_al} is that $Z_{:,j} = p_j + x_j u_1 +(1-x_j) u_0$ with $j \in [n],$ where $x_j$ is binary and $u_1,u_0$ are orthogonal to each other. Note that our proof also extends straightforwardly to such a case.}.

\begin{lemma}
    \label{lemma:edelman_B2_continuous}
    Let $x \in [0,1]^{n}$ and $\mathbb{Z} : [0,1]^{n} \to \mathbb{R}^{d \times n}$ with $Z = \mathbb{Z}(x)$ be a fixed embedding satisfying the structural conditions in Assumption~\ref{assp:lemma_recoverability}, with approximate orthogonality parameter $\Delta$ and $\Delta < 1/s$. Under Assumption \ref{assp:lemma_target_function}, there exists $g: \mathbb{R}^{d \times n} \mapsto \mathbb{R}$ {in the Transformer function class $\mathcal{G}$} defined in \eqref{eq:function_class} such that $g$ approximates $g^*$ under $\mathbb{Z}$, i.e.,
    \begin{equation}
        \sup_{Z \in \mathbb{Z}([0,1]^{n})} \big| g(Z) - g^*(Z) \big| \;\le\; \varepsilon.
    \end{equation}
    And the weights satisfy
    \begin{equation}
        \begin{split}
            \|W^{(Q)}_{h}\|_{2,1} \le& \frac{\log\!\left(\frac{4n}{\gamma}\right)}{1 - 2\Delta}, \qquad \|W^{(K)}_{h}\|_{2,1} \le 1, \qquad \|W^{(V)}_{h}\|_{2,1} \le 1, \qquad \|W^{(O)}_{h}\|_{2,1} \le 1
        \end{split}
    \end{equation}
    for $h=1,\dots,s$. And the weights in MLP satisfy $\max_{i}\left\|W_i\right\|_{\infty} \leq B$ and $\max_{i}\left|b_i\right|_{\infty} \leq B$ for $i=0,\ldots,D$ and $B >0$.
\end{lemma}

\begin{proof}
Fix $\mathcal I=\{i_1,\dots,i_s\}\subset[n]$ and let $Z = \mathbb Z(x)$. 

\paragraph{Step 1: (Attention Approximation Error).}
Taking $x \in[0,1]^n$, we aim to approximate the 1-bounded $\gamma$-injective (e.g., with $\gamma=1$) $\mathcal{I}$-sparse function
\begin{equation} 
    g_{\mathcal{I}}(x) := \sum_{i \in \mathcal{I}} x_i \, e_{\pi(i)}^{(d)} \qquad \in \mathbb{R}^d, 
\end{equation}
where $\pi: \mathcal{I} \rightarrow[s]$ is a bijection, and we write $i_h := \pi^{-1}(h)$.  

Set $r$ the index of the readout column. Under Assumption \ref{assp:lemma_recoverability}, we define
\begin{equation}\label{eq:weights_short}
    W_h^{(Q)} := R\, e_1^{(k)} p_r^\top,\quad
    W_h^{(K)} := e_1^{(k)} p_{i_{h}}^\top,\quad
    W_h^{(V)} := e_1^{(k)} u^\top,\quad
    W_h^{(O)} := e_h^{(d)} {e_1^{(k)}}^\top,
\end{equation}
with $R>0$ to be specified. For head $h$, 
\begin{equation}
    \bigl(\alpha_h(Z)\bigl)_{:,r} := \sigma_S \left( Z^{\top} {W_h^{(K)}}^{\top} W_h^{(Q)} Z_{:,r} \right) \qquad \in [0,1]^{n}
\end{equation}
By Lemma B.2 and B.7 in \cite{2022_Edelman_et_al}, we have
\begin{equation}
    \left\|\bigl(\alpha_h(Z)\bigl)_{:,r} - e_{i_h}^{(n)}\right\|_1 \le \frac{2n}{\exp({R(1 - 2\Delta)})}.
\end{equation}
Since $W_h^{(V)}Z = e_1^{(k)}(u^{\top} Z) = e_1^{(k)}(x_1,\dots,x_n) = e_1^{(k)} \sum_{j=1}^n x_j \left(e_j^{(n)}\right)^{\top}$, then
\begin{equation}
    W_h^{(O)}\,W_h^{(V)} Z = e_h^{(d)} {e_1^{(k)}}^\top \, e_1^{(k)} u^\top Z = e_h^{(d)}\sum_{j=1}^n x_j \,(e_j^{(n)})^\top .
\end{equation}
The attention output for head $h$ is
\begin{equation}
    g_h^{(msa)}(\mathbb{Z}(x)) = e_h^{(d)} \sum_{j=1}^n x_j {e_j^{(n)}}^{\top} \bigl(\alpha_h(\mathbb{Z}(x))\bigl)_{:,r} = e_h^{(d)} \sum_{j=1}^n x_j \, \bigl(\alpha_h(\mathbb{Z}(x))\bigl)_{j,r} \qquad \in [0,1]^{d}.
\end{equation}
Summing over $h=1,\dots,s$ gives the MSA output
\begin{equation}
    g^{(msa)}(\mathbb{Z}(x)) := \sum_{h=1}^s f_h^{(msa)}(\mathbb{Z}(x))
    = \sum_{h=1}^s e_h^{(d)}\sum_{j=1}^n x_j\, \bigl(\alpha_h(\mathbb{Z}(x))\bigl)_{j,r} \qquad \in [0,1]^{d}.
\end{equation} 

After the MSA sublayer, the FFN sublayer $\mathcal{G}^{(FF)}$ acts column-wise. We may choose the FFN parameters such that the overall FFN mapping acts as the identity on the readout. As a remark, the FFN can also be used to eliminate the residual connection (see Lemma 5 in \cite{2025_Jiao_et_al}). In what follows, we take $g^{(tf)}(\mathbb Z(x)) := g^{(msa)}(\mathbb Z(x))$. Thus
\begin{equation}
    \begin{split}
        \big\|g^{(tf)}(\mathbb{Z}(x)) - g_{\mathcal I}(x)\big\|_\infty
        &= \max_{h\in[s]}
        \left|\sum_{j=1}^n x_j\,\bigl(\alpha_h(Z)\bigl)_{j,r} - x_{i_h}\right| \\
        &= \max_{h\in[s]} \left|x^\top \Big(\bigl(\alpha_h(Z)\bigl)_{:,r}-e_{i_h}^{(n)}\Big)\right|\\
        &\le\ \max_{h\in[s]} \|x\|_\infty\,\left\|\bigl(\alpha_h(Z)\bigl)_{:,r}-e_{i_h}^{(n)}\right\|_1 \\
        &\le \max_{h\in[s]}\left\|\bigl(\alpha_h(Z)\bigl)_{:,r}-e_{i_h}^{(n)}\right\|_1 \\
        &\le\ \frac{2n}{\exp\!\big(R(1-2\Delta)\big)}.
    \end{split}
\end{equation}

Let $x,x'\in[0,1]^n$ satisfy $\Pi_{\mathcal I}x\neq \Pi_{\mathcal I}x'$. Set $Z:=\mathbb Z(x)$ and $Z':=\mathbb Z(x')$. For each $h\in[s]$, note that the attention weights $\big(\alpha_h(Z)\big)_{:,r}$ depend only on the positional encodings $\{p_j\}_{j=1}^n$ (and the readout index $r$), hence they are independent of $x$.
Moreover, by the explicit form of the construction,
\begin{equation} 
    e_h^{(d)} \cdot g^{(tf)}(Z) = \sum_{j=1}^n x_j\,\big(\alpha_h(Z)\big)_{j,r}, \qquad e_h^{(d)} \cdot g_{\mathcal{I}}(x) = x_{i_h}. 
\end{equation}
we have for each $h\in[s]$,
\begin{equation}\label{eq:inc_bound_head}
    \begin{split}
        &\Big| e_h^{(d)} \cdot \big(g^{(tf)}(Z) - g^{(tf)}(Z')\big) - e_h^{(d)} \cdot \big(g_{\mathcal I}(x) - g_{\mathcal I}(x')\big) \Big| \\
        =\;& \left|\sum_{j=1}^n \big(\alpha_h(Z)\big)_{j,r}(x_j-x'_j) - (x_{i_h}-x'_{i_h})\right| \\
        =\;& \left|\sum_{j=1}^n \Big(\big(\alpha_h(Z)\big)_{j,r}-(e_{i_h}^{(n)})_j\Big)(x_j-x'_j)\right| \\
        =\;& \left|\sum_{j\in\mathcal I} \Big(\big(\alpha_h(Z)\big)_{j,r}-(e_{i_h}^{(n)})_j\Big)(x_j-x'_j)\right| \\
        \le\;& \left\|\big(\alpha_h(Z)\big)_{:,r}-e_{i_h}^{(n)}\right\|_1 \cdot \left\|\Pi_{\mathcal I}x-\Pi_{\mathcal I}x'\right\|_\infty \\
        \le\;& \delta\,\left\|\Pi_{\mathcal I}x-\Pi_{\mathcal I}x'\right\|_\infty.
    \end{split}
\end{equation}
Taking the maximum over $h\in[s]$ yields
\begin{equation}\label{eq:inc_bound_final_pf}
    \sup_{\substack{x, x'\\ \Pi_{\mathcal{I}}x \neq \Pi_{\mathcal{I}}x'}} \frac{\Big\|g^{(tf)}(\mathbb Z(x)) - g^{(tf)}(\mathbb{Z}(x')) - \big(g_{\mathcal{I}}(x) - g_{\mathcal{I}}(x')\big)\Big\|_\infty} {\|\Pi_{\mathcal{I}}x-\Pi_{\mathcal I}x'\|_\infty} \;\le\; \delta.
\end{equation}

Now use the increment control to get injectivity of the Transformer representation. For any $x, x^{\prime}$ with $\Pi_{\mathcal{I}} x \neq \Pi_{\mathcal{I}} x^{\prime}$,
\begin{equation}
    \begin{split}
        &\left\|g^{(tf)}(\mathbb{Z}(x)) - g^{(tf)}\left(\mathbb{Z}\left(x^{\prime}\right)\right)\right\|_{\infty} \\
        = \quad & \left\|\Big(g_{\mathcal{I}}(x) - g_{\mathcal{I}}(x')\Big) + \Big(g^{(tf)}(\mathbb{Z}(x)) - g^{(tf)}(\mathbb{Z}(x')) - \big(g_{\mathcal{I}}(x) - g_{\mathcal{I}}(x')\big)\Big)\right\|_{\infty} \\ 
        \ge \quad & \left\|g_{\mathcal{I}}(x) - g_{\mathcal{I}}(x')\right\|_{\infty} - \left\|g^{(tf)}(\mathbb{Z}(x)) - g^{(tf)}(\mathbb{Z}(x')) - \big(g_{\mathcal{I}}(x) - g_{\mathcal{I}}(x')\big)\right\|_{\infty} \\ 
        \ge \quad & \left\|g_{\mathcal{I}}(x) - g_{\mathcal{I}}(x')\right\|_{\infty} - \delta \left\|\Pi_{\mathcal{I}} x - \Pi_{\mathcal{I}} x^{\prime}\right\|_{\infty}.
    \end{split}
\end{equation}
Since $g_{\mathcal{I}}$ is $\gamma$-injective under Assumption \ref{assp:lemma_target_function},
\begin{equation} \label{eq:tf_gamma_minus_delta_injective}
    \left\|g^{(tf)}(\mathbb{Z}(x)) - g^{(tf)}\left(\mathbb{Z}\left(x^{\prime}\right)\right)\right\|_{\infty} \geq(\gamma-\delta)\left\|\Pi_{\mathcal{I}} x - \Pi_{\mathcal{I}} x^{\prime}\right\|_{\infty} .
\end{equation}
Choose $R$ such that $\delta \leq \gamma / 2$, i.e.
\begin{equation} \label{eq:R_gamma_choice}
    R \geq \frac{\log (4 n / \gamma)}{1-2 \Delta}
\end{equation}
Then $g^{(tf)} \circ \mathbb{Z}$ is $\gamma / 2$-injective w.r.t. $\Pi_{\mathcal{I}} x$.

Define the set
\begin{equation}
    \mathcal S := \left\{ g^{(tf)}(\mathbb Z(x)) : x\in[0,1]^n \right\}\subset\mathbb R^d.
\end{equation}
and define a recovery map
\begin{equation}
    \mathrm{Rec}:\mathcal S\to[0,1]^s, \qquad \mathrm{Rec}\left(g^{(tf)}(\mathbb Z(x))\right) := \Pi_{\mathcal I}x.
\end{equation}
Since \eqref{eq:tf_gamma_minus_delta_injective} implies that $\mathrm{Rec}$ is Lipschitz on $\mathcal S$:
for any $x,x'$,
\begin{equation}\label{eq:Rec_Lipschitz}
    \begin{split}
        \left\|\mathrm{Rec}\left(g^{(tf)}(\mathbb Z(x))\right) - \mathrm{Rec}\!\left(g^{(tf)}(\mathbb Z(x'))\right)\right\|_\infty &= \|\Pi_{\mathcal I}x - \Pi_{\mathcal I}x'\|_\infty \\
        &\le \frac{1}{\gamma-\delta}\, \|g^{(tf)}(\mathbb Z(x)) - g^{(tf)}(\mathbb Z(x'))\|_\infty \\
        &\le \frac{2}{\gamma}\, \|g^{(tf)}(\mathbb Z(x)) - g^{(tf)}(\mathbb Z(x'))\|_\infty.
    \end{split}
\end{equation}

\paragraph{Step 2: (MLP Approximation Error).}
For any $u,v\in\mathcal S$, using the $\beta$-H\"older continuity of $\phi$ and \eqref{eq:Rec_Lipschitz},
\begin{equation}\label{eq:phi_holder_on_S}
    \big|\phi(\mathrm{Rec}(u)) - \phi(\mathrm{Rec}(v))\big| \le L_{\mathcal I}\left(\frac{2}{\gamma}\right)^\beta \|u-v\|_\infty^\beta.
\end{equation}
Thus the function
\begin{equation}
    u \;\mapsto\; \phi(\mathrm{Rec}(u)), \qquad u\in\mathcal S,
\end{equation}
is $\beta$-H\"older on $\mathcal S$ with constant $L_{\mathcal I}(2/\gamma)^\beta$ and intrinsic dimension $s$.

Recalling the output MLP neural network in \eqref{eq:mlp_class}. Define that $a_n \asymp b_n$ if $a_n \lesssim b_n$ and $b_n \lesssim a_n$ for two sequences. Let $\lceil\cdot\rceil$ denotes the ceiling function and $\vee$ denotes the maximum operator. By Schmidt-Hieber's Theorem 5, there exists an MLP $f^{(mlp)}$ with depth $D = 8 + (m+5)\left(1 + \left\lceil\log_2(s \vee \beta)\right\rceil\right)$, hidden layer width $d_{m} := d_{m}^{(i)} = 6(s+\lceil\beta\rceil) N$ for $i = 1,\ldots,D-1$, and number of nonzero parameters $S^{(net)} \le C_0\,(s+\beta+1)^{3+s}\,N(m+6)$, where $m$ and $N$ are any integers satisfying $m \ge 1$ and $N \ge (\beta+1)^s \vee(L_{\mathcal{I}} + 1) \exp(s)$ such that
\begin{equation}\label{eq:mlp_approx_no_psi}
    \sup_{u\in\mathcal S} \big|g^{(mlp)}(u) - \phi(\mathrm{Rec}(u))\big| = O \left(L_{\mathcal{I}}\left(\frac{2}{\gamma}\right)^\beta \bigl(N2^{-m} + N^{-\beta/s})\right).
\end{equation}

Since $\mathrm{Rec}(f^{(tf)}(\mathbb Z(x)))=\Pi_{\mathcal I}x$, for all $x\in[0,1]^n$,
\begin{equation}
    g^{(mlp)}\left(g^{(tf)}(\mathbb Z(x))\right) = \phi(\Pi_{\mathcal I}x) = f^*_{\mathcal I}(x).
\end{equation}
Noting that $d_m^{(i)}\asymp N$, we may equivalently write the bias term as
\begin{equation}\label{eq:bias_dm_m}
    \sup_{x\in[0,1]^n} \left| g^{(mlp)}\left(g^{(tf)}(\mathbb Z(x))\right) - g^*(x) \right| = \;\lesssim\; \left(\frac{2}{\gamma}\right)^\beta \Bigl(d_m\,2^{-m} + d_m^{-\beta/s}\Bigr).
\end{equation}
We choose $m \asymp \log N \asymp \log d_m$, so that $d_m2^{-m}\lesssim d_m^{-\beta/s}$ and hence
\begin{equation}
    \sup_{x\in[0,1]^n} \left| g^{(mlp)}\left(g^{(tf)}(\mathbb Z(x))\right) - g^*(x) \right| = O\left(\left(\frac{2}{\gamma}\right)^\beta \bigl(d_{m}\bigl)^{-\beta/s}\right),
\end{equation}
under the choice of $R$ in \eqref{eq:R_gamma_choice}.

\end{proof}
    \section{Additional Figure(s)} \label{sec:tables_and_figures}

\begin{figure}[h]
    \centering
    \caption{CoVaR and $\Delta$CoVaR Differences: Text vs. No-Text (Eight G-SIBs)}
    \includegraphics[width=0.95\textwidth, keepaspectratio]{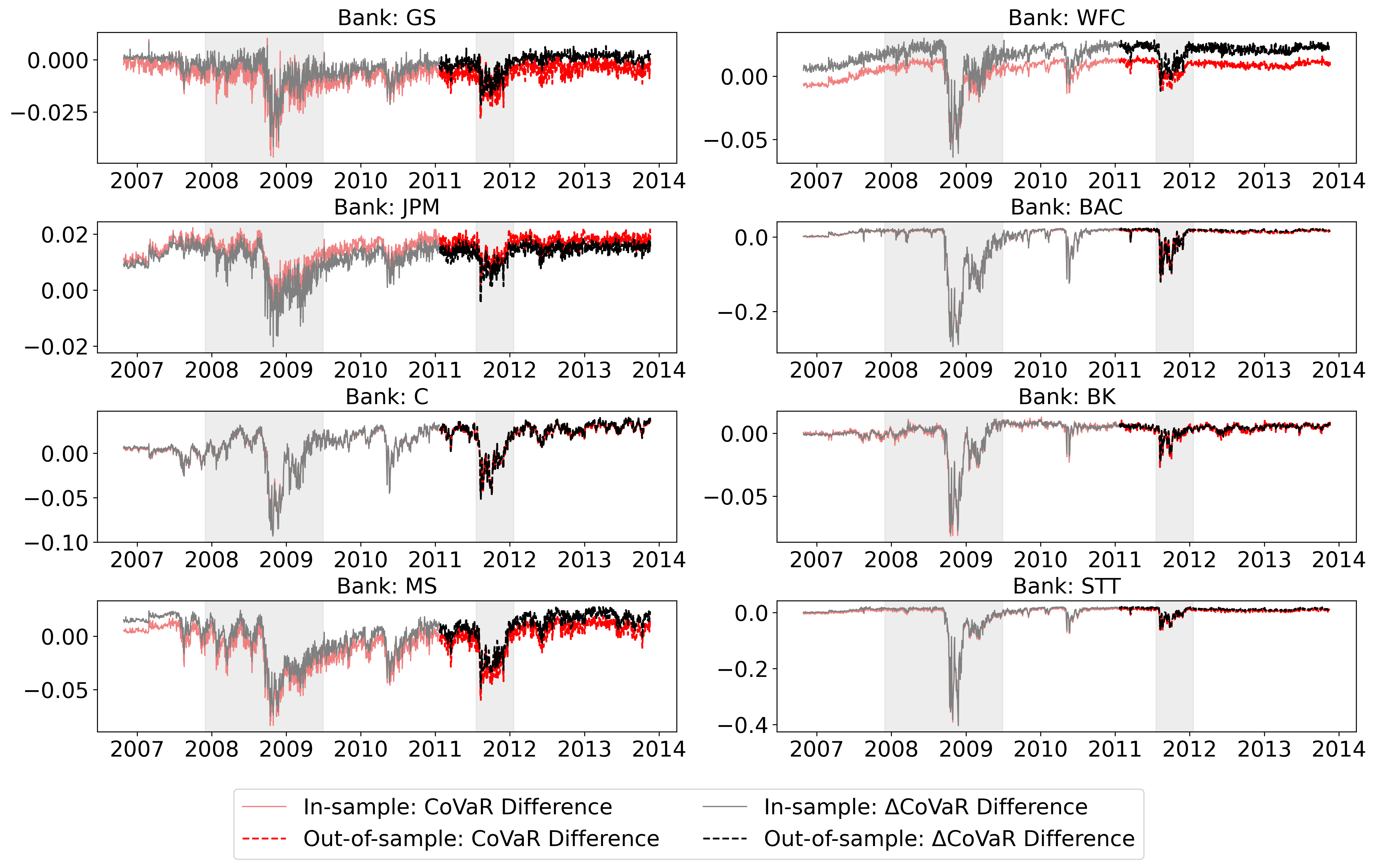}
    \caption*{\textit{Notes:} Differences in CoVaR (red) and $\Delta$CoVaR (black) between models with and without text for the eight Global Systemically Important Banks.}
    \label{fig:all_banks}
\end{figure}
\end{appendix}

\end{document}